\documentclass[letterpaper]{scrartcl}

\usepackage[]{group} 

\newcommand{\conv}{\mathrm{conv}}
\newcommand{\aff}{\mathrm{aff}}
\newcommand{\cl}{\mathrm{cl}}
\usepackage{multirow}
\usepackage{arydshln}
\usepackage{ifthen}
\usepackage{blkarray}
\usepackage{calc}
\usepackage{geometry}
\usepackage{subcaption}
\usepackage{thm-restate}
\usepackage{float}

\usetikzlibrary{arrows.meta,decorations.markings}

\newtheorem{claim}{Claim}

\makeatletter
\renewcommand*\env@matrix[1][*\c@MaxMatrixCols c]{%
  \hskip -\arraycolsep
  \let\@ifnextchar\new@ifnextchar
  \array{#1}}
\makeatother

\renewcommand{\int}[1]{\mathrm{int}({#1})}

\renewcommand{\epsilon}{\varepsilon}
\newcommand{\vnm}{\mathit{vNM}} 
\newcommand{\wl}{\mathit{WL}} 
\newcommand{\dich}{\mathit{dich}} 
\newcommand{\dist}{\mathrm{dist}}

\title{Arrovian Aggregation of Convex Preferences}

\author{Florian Brandl \qquad Felix Brandt\\
Technische Universit\"at M\"unchen\\ 
\texttt{\small \{brandlfl,brandtf\}@in.tum.de}}

\begin{document}

\maketitle

\begin{abstract}
We consider social welfare functions that satisfy Arrow's classic axioms of \emph{independence of irrelevant alternatives} and \emph{Pareto optimality} when the outcome space is the convex hull of some finite set of alternatives. 
Individual and collective preferences are assumed to be continuous and convex, which guarantees the existence of maximal elements and the consistency of choice functions that return these elements, even without insisting on transitivity.
We provide characterizations of both the domains of preferences and the social welfare functions that allow for anonymous Arrovian aggregation. 
The domains admit arbitrary preferences over alternatives, which completely determine an agent's preferences over all mixed outcomes. 
On these domains, Arrow's impossibility turns into a complete characterization of a unique 
social welfare function, which can be readily applied in 
settings involving divisible resources such as probability, time, or money.
\end{abstract}

\section{Introduction}
\label{sec:intro}

A central concept in welfare economics are social welfare functions (SWFs) in the tradition of Arrow, i.e., functions that map a collection of individual preference relations over some set of alternatives to a social preference relation over the alternatives.
Arrow's seminal theorem states that every SWF that satisfies Pareto optimality and independence of irrelevant alternatives is dictatorial \citep{Arro51a}. 
This sweeping impossibility significantly strengthened an earlier observation by \citet{Cond85a} and
sent shockwaves throughout welfare economics, political philosophy, economic theory, and even seemingly unrelated disciplines such as philosophy of science and engineering design \citep[see, e.g.,][]{MaSe14a,Sen17a,Gaer18a}. 
A large body of subsequent work has studied whether more positive results can be obtained by modifying implicit assumptions on the domain of admissible preferences, both individually and collectively. 
Two main approaches can be distinguished.

The first approach, pioneered by \citet{Sen69a}, is based on the observation that ``Arrow's requirement that social preferences be transitive is probably the least defensible of the conditions that lead to his dictatorship result'' \citep{BlPo82a}. Consequently, the idea is to \emph{weaken the assumption of transitivity of collective preferences} to quasi-transitivity, acyclicity, path independence, or similar conditions.
Although this does allow for non-dictatorial aggregation functions that meet Arrow's criteria, these functions turned out to be highly objectionable, usually on grounds of involving weak kinds of dictatorships or violating other conditions deemed to be indispensable for reasonable preference aggregation \citep[for an overview of the extensive literature, see][]{Kell78a,Sen77a,Sen86a,Schw86a,CaKe02a}.  
Particularly noteworthy are results about acyclic collective preference relations \citep[e.g.,][]{MCSo72a,Brow75a,BlDe77a,BlPo82a,BlPo83a,Bank95a} because acyclicity is necessary and sufficient for the existence of maximal elements when there is a finite number of alternatives.
\citet{Sen95a} concludes that ``the arbitrariness of power of which Arrow's case of dictatorship is an extreme example, lingers in one form or another even when transitivity is dropped, so long as \emph{some} regularity is demanded (such as the absence of cycles).''

Another stream of research has analyzed the implications of \emph{imposing additional structure on individual preferences}. This has resulted in a number of positive results for restricted  domains, such as dichotomous or single-peaked preferences, which allow for attractive SWFs \citep[e.g.,][]{Blac48a,Arro51a,Inad69a,SePa69a,EhSt08a}. 
Many domains of economic interest are concerned with infinite sets of outcomes, which satisfy structural restrictions such as compactness and convexity. 
Preferences over these outcomes are typically assumed to satisfy some form of continuity and convexity, i.e., they are robust with respect to minimal changes in outcomes and with respect to convex combinations of outcomes. Various results have shown that Arrow's impossibility remains intact under these assumptions \citep[e.g.,][]{KMS79a,Bord83b,BoLe89a,BoLe90a,BoLe90b,Camp89a,Rede95a}. 
\citet{LeWe11a} provide an overview and conclude that ``economic domain restrictions do not provide a satisfactory way of avoiding Arrovian social choice impossibilities, except when the set of alternatives is one-dimensional and preferences are single-peaked.''

The point of departure for the present approach is the observation that all these impossibilities involve some form of transitivity (e.g., acyclicity), even though no such assumption is necessary to guarantee the existence of maximal elements in domains of continuous and convex preferences. \citet{Sonn71a} has shown that all continuous and convex preference relations admit a maximal element in every non-empty, compact, and convex set of outcomes. Moreover, returning maximal elements under the given conditions satisfies standard properties of choice consistency introduced by \citet{Sen69a,Sen71a}. Continuous and convex preference relations can thus be interpreted as rationalizing relations for the choice behavior of rational agents and 
there is little justification for demanding transitivity.\footnote{Another frequently cited reason to justify transitivity is the \emph{money pump}, where an agent with cyclic preferences over three outcomes is deceived into paying unlimited amounts of money in an infinite series of cyclical exchanges. As \citet{Fish91a} notes, however, the money pump ``applies transitive thinking [in the form of money] to an intransitive world [given the agent's preferences]''. Another issue with the money pump in our framework is that it cleverly avoids convexity of the feasible set by splitting it up into three subsets whose union is not convex. If the agent were confronted with a choice from the convex hull of the three original outcomes, he could simply pick his (unique) most-preferred mixed outcome and would not be tempted to exchange it when offered any other outcome in the future.}  
Transitivity has repeatedly been criticized for being overly demanding \citep[see, e.g.,][]{May54a,Fish70c,BaMa88a,Fish91a,Anan93a,Anan09a}. \citet{Anan09a} concludes, ``once considered a cornerstone of rational choice theory, the status of transitivity has been dramatically reevaluated by economists and philosophers in recent years.''

\subsection*{Summary of Results}

We show that Arrow's theorem ceases to hold for convex outcome sets when dispensing with transitivity, and, moreover, Arrow's axioms characterize a unique anonymous SWF that we refer to as \emph{pairwise utilitarianism}. The SWF is utilitarian because collective preferences are obtained by adding the canonical skew-symmetric bilinear (SSB) utility functions representing the agents' ordinal preferences. SSB utility functions assign a utility value to each \emph{pair} of alternatives and are more general than traditional linear utility functions. The SWF is pairwise in the sense that it merely takes into account the numbers of agents who prefer one alternative to another and also satisfies Condorcet's pairwise majority criterion.

More precisely, we consider a convex set of outcomes consisting of all probability measures on some finite abstract set of alternatives, which we refer to as pure outcomes. 
These outcome sets, for example, arise when allocating a divisible resource (such as probability, time, or money) to alternatives. The canonical example is the standard unstructured social choice setting that also allows for lotteries between alternatives.
Individual and collective preference relations over these outcomes are assumed to satisfy continuity, convexity, and a mild symmetry condition. 
In order to motivate these assumptions, we prove that continuity and convexity are necessary and sufficient for consistent choice behavior, mirroring a classic characterization by \citet{Sen71a} in the finite non-convex choice setting (see \Cref{prop:rationalizable} in \Cref{sec:ratchoice}).
We then show that there is a unique inclusion-maximal Cartesian domain of preference profiles that allows for anonymous Arrovian aggregation and satisfies minimal richness conditions (see \Cref{thm:pcdomain} in \Cref{sec:domain}). 
This domain allows for arbitrary preferences over pure outcomes, which in turn completely determine an agent's preferences over all remaining outcomes. When interpreting outcomes as lotteries, this preference extension has a particularly simple and intuitive explanation: \emph{one lottery is preferred to another if and only if the former is more likely to return a more preferred alternative}. Incidentally, this preference extension, which constitutes a central special case of SSB utility functions as introduced by \citet{Fish82c}, is supported by recent experimental evidence (see \Cref{sec:domain}). We also provide an alternative characterization of SSB utility functions using continuity and convexity, which may be of independent interest (see \Cref{prop:ssb} in \Cref{sec:ssb}).

Our main theorem shows that the only Arrovian SWFs on this domain are affine utilitarian with respect to the underlying SSB utility functions (see \Cref{thm:swf} in \Cref{sec:swf}). As a consequence, there is a unique anonymous Arrovian SWF, which compares outcomes by the sign of the bilinear form given by the pairwise majority margins. The resulting collective preference relation over \emph{pure} outcomes coincides with majority rule and the corresponding choice function is therefore consistent with Condorcet's principle of selecting a pure outcome that is majority-preferred to every other pure outcome whenever this is possible.\footnote{It is therefore in line with \citet{DaMa08a} who, also based on Arrow's axioms, have forcefully argued in favor of majority rule in domains where Condorcet winners are guaranteed to exist. Our arguments extend to unrestricted preferences over pure outcomes.} This relation is naturally extended to mixed outcomes such that every compact and convex set of outcomes admits a collectively most preferred outcome. In the context of lotteries, the collective preference relation admits a very intuitive interpretation: in order to compare two lotteries $p$ and $q$, randomly sample a pure outcome $a$ from $p$, a pure outcome $b$ from $q$, and an agent $i$ from the uniform distribution over agents. Then, $p$ is collectively preferred to $q$ if and only if the probability that agent $i$ prefers $a$ to $b$ is greater than the probability that he prefers $b$ to $a$.

We also show that, when restricting attention to von Neumann-Morgenstern preferences over outcomes, anonymous Arrovian aggregation is only possible for dichotomous preferences and the only SWF to do so corresponds to the well-known approval voting rule (see Theorems \ref{thm:dichdomain} and \ref{thm:dichswf} in \Cref{app:dich}).
Our results can thus be interpreted as a generalization of both the domain of dichotomous preferences and approval voting.

\subsection*{Illustrative Example}

The most natural setting where our assumptions apply is that of preference aggregation when the outcome set consists of all lotteries over a finite set of alternatives. 
By contrast, this subsection discusses the application of our results to a budget allocation setting from public finance. Even though the assumption of convex preferences is somewhat restrictive in this context, we think that the example nicely illustrates the premises and consequences as well as the limitations of our results.

Imagine that a representative body, consisting of 100 delegates, aims at reaching a joint decision on how to divide a nation's tax budget between four departments: education, transportation, health, and military.
The delegates belong to different parties ($A$, $B$, $C$, and $D$) and each party already put forward a favored budget proposal (see \Cref{tab:example}).\footnote{One issue ignored here is how to arrive at the set of budget proposals. In particular, it is required that the proposals are affinely independent.} 
We assume, for simplicity, that there are four groups of delegates with identical preferences, which could---for example---correspond to the four parties, but emphasize that no such assumption is required for our results. In fact, no restrictions whatsoever are imposed on the delegates' preferences over proposals. As a consequence, Arrow's theorem implies that every non-dictatorial Pareto optimal SWF violates independence of irrelevant alternatives (IIA), i.e., collective preferences over pairs of proposals may depend on individual preferences over other, unrelated, proposals. Moreover, in the given profile of preferences, the pairwise majority relation is cyclic and every proposal is unstable in the sense that it can be overthrown by a majority of the delegates.

\begin{table}
\centering
\begin{tabular}{lccccc}
\cmidrule[\heavyrulewidth]{1-5}
& $A$ & $B$ & $C$ & $D$ & \footnotesize\parbox[b]{\widthof{\footnotesize Utilitarianism}}{\centering\sffamily Pairwise\\Utilitarianism}\\
\cmidrule{1-5}
Education & 40\,\% & 30\,\% & 20\,\% & 10\,\% & \footnotesize\sffamily 25.0\,\%\\
Transportation & 30\,\% & 10\,\% & 30\,\% & 30\,\% & \footnotesize\sffamily 26.7\,\%\\
Health & 20\,\% & 40\,\% & 30\,\% & 20\,\% & \footnotesize\sffamily 30.0\,\%\\
Military & 10\,\% & 20\,\% & 20\,\% & 40\,\% & \footnotesize\sffamily 18.3\,\%\\
\cmidrule[\heavyrulewidth]{1-5}
\end{tabular}
\hspace{6ex}
\begin{tabular}{cccc}
\midrule[\heavyrulewidth]
$25$ & $20$ & $45$ & $10$ \\
\midrule
$A$ & $B$ & $C$ & $D$\\
$B$ & $A$ & $A$ & $B$\\
$C$ & $C$ & $D$ & $C$\\
$D$ & $D$ & $B$ & $A$\\
\midrule[\heavyrulewidth]
\end{tabular}
\caption{Public finance example. Four budget proposals ($A$, $B$, $C$, and $D$), the preferences of 100 delegates over these proposals, and the rounded budget allocation returned by pairwise utilitarianism.}
\label{tab:example}
\end{table}

In the described setting, it seems fairly natural to extend the space of possible outcomes by allowing to compromise between the different proposals such that the set of outcomes is now the convex hull of $A$, $B$, $C$, and $D$. A fifty-fifty mixture of proposals $A$ and $B$, for instance, will be written as $\nicefrac{1}{2}\, A + \nicefrac{1}{2}\, B$ and assigns 35\% of the budget to education, 20\% to transportation, 30\% to health, and 15\% to military.
When assuming that the delegates' preferences over this enlarged, infinite, set of outcomes satisfy continuity and convexity and an innocuous symmetry condition, such preferences can be represented by bilinear utility functions. Note, however, that these assumptions do not allow for decreasing marginal returns for increasingly large investments as well as complementarities and substitutabilities among departments.

One may now wonder whether there exist Arrovian SWFs for this modified setting. 
Our first theorem shows that, since we insist on allowing arbitrary preferences over proposals, anonymous Arrovian aggregation is only possible when preferences over proposals are extended to preferences over mixtures of proposals by sampling proposals from each mixture and preferring the mixture which wins more pairwise comparisons. 
According to this preference extension, the 25 delegates who prefer proposals $A$ to $B$ to $C$ to $D$ are, for example, assumed to prefer $\nicefrac{2}{3}\, A + \nicefrac{1}{3}\, C$ to proposal $B$ and to be indifferent between proposal $C$ and $\nicefrac{1}{2}\, A + \nicefrac{1}{2}\, D$.
Even though preferences over proposals are transitive, some preferences over mixed outcomes will be cyclic. This effect is known as the Steinhaus-Trybula paradox (see \Cref{fig:pccycle} on Page \pageref{fig:pccycle}).

Our second theorem shows that there is a unique anonymous Arrovian SWF on this preference domain, which we refer to as pairwise utilitarianism. This SWF is based on the bilinear form given by the matrix of pairwise majority margins
\[
\phi =
\begin{blockarray}{ccccc}
	A & B & C & D\\
	\begin{block}{(cccc)c}
	  	0 & 40 & -10 & 80 & A\\
		-40 & 0 & 10 & -10 & B\\
  	  	10 & -10 & 0 & 80 & C\\
		-80 & 10 & -80 & 0 & D\\
	\end{block}
\end{blockarray}\text.
\]
Here, $\phi(A,B)=40$ because the number of delegates who prefer $A$ to $B$ minus the number of delegates who prefer $B$ to $A$ is $(25+45)-(20+10)=40$. Collective preferences are obtained by checking the sign of the corresponding value of $\phi$: 
proposal $A$ is preferred to proposal $B$ because $\phi(A,B)=40>0$, $B$ is preferred to $C$ because $\phi(B,C)=10>0$, $C$ is preferred to $A$ because $\phi(C,A)=10>0$. Bilinearity implies that, for example, $\nicefrac{1}{2}\, A+\nicefrac{1}{2}\, C$ is preferred to $B$ because $\phi(\nicefrac{1}{2}\, A+\nicefrac{1}{2}\, C, B)=15>0$. It follows from the Minimax Theorem that every convex and closed set contains at least one most-preferred outcome. 
In our example, the unique most-preferred outcome is a convex combination of the first three proposals \[p = \nicefrac{1}{6}\,A + \nicefrac{1}{6}\,B + \nicefrac{2}{3}\,C\text.\]
This corresponds to the budget allocation given in the pairwise utilitarianism column of \Cref{tab:example}. The \emph{choice} function that returns maximal pairwise utilitarian outcomes satisfies contraction consistency. Hence, the optimal allocation is not affected if proposal $D$ is retracted. Moreover, the choice function satisfies expansion consistency, i.e., if allocation $p$ is not only chosen in the example described above, but also in an alternative choice setting in which proposal $D$ is replaced with another proposal $E$, then $p$ would also be chosen if all five proposals (and their convex combinations) were feasible.

The pairwise utilitarian SWF satisfies Pareto optimality. Proposal $C$, for example, is socially preferred to $\nicefrac{1}{2}\, A + \nicefrac{1}{2}\, D$ because the agents represented in columns 1, 2, and 4 are indifferent between both outcomes while those represented in column 3 strictly prefer the former to the latter. 
The SWF also satisfies IIA in the sense that collective preferences between outcomes in the convex closure of some subset of proposals only depend on the individual preferences between these outcomes. For example, all preferences between outcomes in the convex closure of two proposals only depend on the pairwise majority relation between these proposals. Moreover, the collective preferences between all outcomes in the convex closure of any triple of proposals are independent of individual preferences over outcomes that involve the fourth proposal.
Of course, more significant than the observation that pairwise utilitarianism satisfies Pareto optimality and IIA is the fact that it is the \emph{only} such SWF.\footnote{The public finance example can also be used to illustrate two further desirable properties of the maximal pairwise utilitarian choice function. It is \emph{population-consistent} in the sense that merging two bodies of delegates, each of which came to the same conclusion, will not affect the outcome and it is \emph{composition-consistent}, which---among other things---implies that the choice function cannot be manipulated by introducing additional proposals, which are very similar to existing ones. 
\citep[see][]{Bran13a}.}

\section{Related Work}
\label{sec:related}

A special case of our setting, which has been well-studied, concerns individual preferences over lotteries that satisfy the \emph{von Neumann-Morgenstern (vNM) axioms}, i.e., preferences that can be represented by assigning cardinal utilities to alternatives such that lotteries are compared based on the expected utility they produce.  \citet{Samu67a} conjectured that Arrow's impossibility still holds under these assumptions and \citet{KaSc77b} showed that this is indeed the case when there are at least four alternatives. \citet{Hyll80b} later pointed out that a continuity assumption made by \citeauthor{KaSc77b} is not required.
There are other versions of Arrow's impossibility for vNM preferences, which differ in modeling assumptions and whether SWFs aggregate cardinal utilities or the preference relations represented by these utilities.
A detailed comparison of these results is given in Appendix~\ref{app:vnm}.

Our results apply to Arrovian aggregation of preferences over lotteries under much loosened assumptions about preferences over lotteries. In particular, the axioms we presume entail that preferences over lotteries can be represented by \emph{skew-symmetric bilinear (SSB) utility functions}, which assign a utility value to each pair of lotteries. One lottery is preferred to another lottery if the SSB utility for this pair is positive.
SSB utility theory is a generalization of linear expected utility theory due to \citet{vNM47a}, which does not require the controversial independence axiom and transitivity \citep[see, e.g.,][]{Fish82c,Fish84c,Fish88a}.
Independence prescribes that a lottery $p$ is preferred to lottery $q$ if and only if a coin toss between $p$ and a third lottery $r$ is preferred to a coin toss between $q$ and $r$ (with the same coin used in both cases). There is experimental evidence that independence is systematically violated by human decision makers. The Allais Paradox is perhaps the most famous example \citep{Alla53a}.
Detailed reviews of such violations, including those reported by \citet{KaTv79a}, have been provided by \citet{Mach83a,Mach89a} and \citet{McCl88a}. An interesting historical perspective on the independence axiom has been given by \citet{FiWa95a}.

Our characterization of Arrovian SWFs is related to Harsanyi's \emph{Social Aggregation Theorem} \citep{Hars55a}, which shows that, for von Neumann-Morgenstern preferences over lotteries, affine utilitarianism already follows from Pareto indifference \citep[see][for an excellent exposition and various extensions of this theorem]{FSW08a}. Harsanyi's theorem is a statement about \emph{Bergson-Samuelson social welfare functions}, i.e., a single preference profile is considered in isolation. As a consequence, the weights associated with the agents' utility functions may depend on their preferences. This can be prevented by adding axioms that connect the collective preferences across different profiles. The SWF that derives the collective preferences by adding up utility representations normalized to the unit interval is known as \emph{relative utilitarianism} \citep{Dhil98a,DhMe99a,BoCh15a,BoCh15b}. It was characterized by \citet{DhMe99a} using a weakening of IIA and further axioms (see \Cref{fn:ira}).
As shown by \citet{FiGe87a} and further explored by \citet{TuWe99a}, aggregating SSB utility functions is fundamentally different from aggregating von Neumann-Morgenstern utility functions in that Harsanyi's Pareto indifference axiom (and strengthenings thereof) do not imply affine utilitarianism. As we show in this paper, this can be rectified by considering Arrow's multi-profile framework and assuming IIA.
\citet[][Proposition 1]{Mong94a} gives a similar characterization of affine utilitarianism for social welfare functionals (SWFLs), which operate on profiles of cardinal utility functions. The framework of SWFLs (without any invariance property on utility functions) involves interpersonal comparisons of utilities, which significantly weakens the force of IIA (see also \Cref{app:vnm}).

The probabilistic voting rule that returns the maximal elements of the unique anonymous Arrovian SWF is known as \emph{maximal lotteries} \citep{Krew65a,Fish84a} and was recently axiomatized using two consistency conditions \citep{Bran13a}. 
Independently, maximal lotteries have also been studied in the context of randomized matching and assignment \citep[see, e.g.,][]{KMN11a,ABS13a}.

When the set of outcomes cannot be assumed to be convex (for example, because it is finite), a common approach to address the intransitivity of collective preferences is to define alternative notions of maximality, rationalizability, or welfare, leading to concepts such as \emph{transitive closure maximality} or the \emph{uncovered set} \citep[see, e.g.,][]{Lasl97a,BrHa11a,Nish17a,BBSS14a}. Interestingly, the support of maximal lotteries, which is known as the \emph{bipartisan set} or the \emph{essential set} \citep{LLL93b,Lasl00a}, also appears in this literature, even though this approach is fundamentally different from the one pursued in this paper.

\section{Preliminaries}\label{sec:prelim}

Let $\mathcal U$ be a non-empty and finite universal set of alternatives. 
By $\Delta$ we denote the set of all probability measures on $\mathcal U$. 
We assume that $\Delta$ is equipped with the topology induced by the standard Euclidian topology on $\mathbb R^\mathcal U$.
For $X\subseteq \mathcal U$, let $\Delta_X$ be the set of probability measures in $\Delta$ with support in $X$, i.e., $\Delta_X = \{p\in \Delta\colon p(X) = 1\}$. 
We will refer to elements of $\Delta$ as \emph{outcomes} and one-point measures in $\Delta$ as \emph{pure outcomes}.
The \emph{convex hull} of a set of outcomes $X\in \Delta$ will be denoted by $\conv(X)$. 
For $p,q,r\in\Delta$, we write $\conv(p,q,r)$ instead of the more clumsy $\conv(\{p,q,r\})$. 
For $p,q\in\Delta$ and $\lambda\in [0,1]$, we write $[p,q] = \conv(p,q)$, $[p,q) = [p,q]\setminus\{q\}$, $(p,q] = [p,q]\setminus\{p\}$, $(p,q) = [p,q)\cap (p,q]$, and $p\lambda q = \lambda p + (1-\lambda) q$ for short.

The preferences of an agent are represented by an asymmetric binary relation $\succ$ over $\Delta$ called the \emph{preference relation}.
Given two outcomes $p,q\in\Delta$, we write $p\sim q$ when neither $p\succ q$ nor $q\succ p$, and $p\succsim q$ if $p\succ q$ or $p \sim q$.
For $p\in\Delta$, let $U(p) = \{q\in \Delta\colon q\succ p\}$  and $L(p) = \{q\in \Delta\colon p\succ q\}$
be the \emph{strict upper} and \emph{strict lower contour set} of $p$ with respect to $\succ$; $I(p) = \{q\in\Delta\colon p\sim q\}$ denotes the \emph{indifference set} of $p$. 
For $X\subseteq \Delta$, ${\succ}|_X = \{(p,q)\in {\succ}\colon p,q\in X\}$ is the preference relation $\succ$ restricted to outcomes in $X$.
We will consider preference relations that are continuous, i.e., small perturbations to outcomes retain a strict preference, and convex, i.e., preferences are preserved when taking convex combinations of outcomes. This amounts to the following restrictions on contour sets for any $p\in\Delta$.
\begin{gather*}
	U(p) \text{ and } L(p) \text{ are open} \text.\tag{Continuity}\\
	U(p), L(p), U(p)\cup I(p)\text{, and } L(p)\cup I(p) \text{ are convex}\text.\tag{Convexity}
\end{gather*}

\section{Rational and Consistent Choice}
\label{sec:ratchoice}

The existence of maximal elements is usually quoted as the main reason for insisting on transitivity of preference relations.
It was shown by \citet{Sonn71a} that continuity and convexity are already sufficient for the existence of maximal elements in non-empty, compact, and convex sets, even when preferences are intransitive \citep[see also][]{Berg92a,Llin98a}.
We will refer to any such subset of outcomes $X\subseteq \Delta$ as \emph{feasible} and denote the set of all feasible sets by $\mathcal F(\Delta)$. Moreover, define $\max_{\succ}(X)=\{ x\in X \colon x \succsim y \text{ for all }y\in X\}$ for 
any preference relation $\succ$ and feasible set $X$. 

\begin{proposition}\label{thm:sonnenschein}\citep{Sonn71a}
If $\succ$ is a continuous and convex preference relation, then $\max_{\succ}(X)\neq\emptyset$ for every feasible set $X$.\footnote{\citeauthor{Sonn71a} only required that upper contour sets are convex and that lower contour sets are open.} 
\end{proposition}

In a model that assumes the feasibility of all finite non-empty subsets of outcomes, 
\citet{Sen71a} has shown that two intuitive choice consistency conditions, known as \emph{Sen's $\alpha$} (or \emph{contraction}) and \emph{Sen's $\gamma$} (or \emph{expansion}), are equivalent to choosing maximal elements of an underlying preference relation \citep[see also][]{Sen77a}.
Moreover, when the total number of outcomes is finite, such a rationalizing preference relation has to be acyclic because acyclicity is necessary and sufficient for the existence of maximal elements. With convex feasible sets such as the ones we are considering, this is no longer the case, and other properties (such as continuity and convexity) can take over the role of acyclicity. In the following, we show that Sen's theorem can be salvaged in our setting when defining choice functions and choice consistency conditions appropriately.\footnote{Inspired by the classic contributions of \citet{Samu38a}, \citet{Rich66a}, and \citet{Afri67a}, there has been renewed interest in the rationalizability of choice functions in economic domains such as consumer demand \citep[see, e.g.][]{Reny15a,ChEc16a,NOQ17a}. We are, however, not aware of a contraction-expansion-based characterization of rationalizable choice from convex feasible sets.}

A choice function is a function that maps any feasible set to a feasible subset. Formally, we define choice functions as (upper hemi-) continuous functions $S\colon \mathcal F(\Delta)\rightarrow \mathcal F(\Delta)$ such that for all $X\in \mathcal F(\Delta)$, $S(X)\subseteq X$, and for all $p,q\in\Delta$, $S([p,q])\in\{\{p\},\{q\},[p,q])\}$.\todo{FFX: This assumption may be surprising and we maybe should comment on it. FB: Agreed, but I'd rather postpone this to the next version.}
Contraction consistency requires that if an outcome is chosen from some set, then it is also chosen from any subset that it is contained in. A choice function $S$ satisfies \emph{contraction} if for all $X,Y\in \mathcal F(\Delta)$ with $X\cap Y\neq\emptyset$,
\[
 	S(X)\cap Y \subseteq S(X\cap Y)\text.\tag{Contraction}
\]
Expansion consistency demands that an outcome that is chosen from two sets $X$ and $Y$, should also be chosen from their union $X\cup Y$. Since we only consider convex feasible sets, we strengthen this condition by taking the convex hull $\conv(X\cup Y)$ in the consequence.
$S$ satisfies \emph{expansion} if for all $X,Y\in \mathcal F(\Delta)$,
\[
	S(X)\cap S(Y) \subseteq S(\conv(X\cup Y))\text.\tag{Expansion}
\]
Following \citet{Schw76a}, the conjunction of contraction and expansion can be nicely written as a single condition, where for all $X,Y\in\mathcal F(\Delta)$ with $X\cap Y\neq\emptyset$,
\[
	S(X)\cap S(Y) = S(\conv(X\cup Y))\cap X \cap Y\text.\tag{Consistency}
\]
The inclusion from left to right is expansion whereas the inclusion from right to left is equivalent to contraction \citep[see also][]{BrHa11a}.
A choice function that satisfies contraction and expansion will be called \emph{consistent}.

A choice function $S$ is \emph{rationalizable} if there exists a preference relation $\succ$ such that for every $X\in\mathcal{F}(\Delta)$, $S(X)=\max_\succ X$. Let us now consider rationalizability in the context of continuous and convex preference relations.
Any choice function that returns maximal elements of some continuous and convex relation satisfies contraction and expansion.\footnote{There are also stronger versions of expansion, which, together with contraction, are equivalent to the \emph{weak axiom of revealed preference} or \emph{Arrow's choice axiom} \citep{Samu38a,Arro59a}. These conditions imply rationalizability via a \emph{transitive} relation and are therefore not generally satisfied when choosing maximal elements of continuous and convex relations.} 
As we prove in \Cref{app:choice}, the converse holds as well, i.e., if we insist on consistent choice, we may restrict our attention to continuous and convex preference relations.

\begin{restatable}{proposition}{proprationalizable}\label{prop:rationalizable}
	A choice function is rationalizable via a continuous and convex relation if and only if it is consistent.
\end{restatable}

Hence, rational and consistent choice is not only possible without making transitivity (or acyclicity) assumptions, but even Sen's fundamental equivalence between rationality and consistency can be maintained. It follows from \citeauthor{Rich66a}'s~\citeyearpar{Rich66a} theorem that any consistent choice function is rationalized by its revealed preference relation as introduced by \citet{Hout50a}.
This relation moreover satisfies continuity and convexity.

\section{Skew-Symmetric Bilinear Utility Functions}\label{sec:ssb}

Convexity of preferences implies that indifference sets are convex. 
The symmetry axiom introduced by \citet{Fish82c} prescribes that the indifference sets for every triple of outcomes are straight lines that are either parallel or intersect in one point, which may be outside of their convex hull. 
For all $p,q,r\in\Delta$ and $\lambda\in(0,1)$,\footnote{\citeauthor{Fish82c}'s original definition of symmetry requires $p$, $q$, and $r$ to be linearly ordered (cf. \Cref{app:ssb}) and is thus slightly weaker than symmetry as defined here. Since our notion of convexity is weaker than \citeauthor{Fish82c}'s~\citeyearpar{Fish82c} dominance axiom, this stronger formulation of symmetry is required for \Cref{prop:ssb}.\label{fn:symmetry}}

\begin{equation}
\begin{aligned}
	\text{if } q\sim\nicefrac{1}{2}\,p+\nicefrac{1}{2}\,r \text{ and } p \lambda r \sim \nicefrac{1}{2}\,p+\nicefrac{1}{2}\,q
	\text{ then } r \lambda p \sim \nicefrac{1}{2}\,r+\nicefrac{1}{2}\,q\text.
\end{aligned}
\tag{Symmetry}
\end{equation}
\citet{Fish84c} justifies this axiom by stating that ``the degree to which $p$ is preferred to $q$ is equal in absolute magnitude but opposite in sign to the degree to which $q$ is preferred $p$.'' He continues by writing that he is ``a bit uncertain as to whether this should be regarded more as a convention than a testable hypothesis -- much like the asymmetry axiom [\dots], which can almost be thought of as a definitional characteristic of strict preference.''
Without symmetry, continuity and convexity still allow for rather unintuitive preference relations. For example,
let $\mathcal U =\{a,b\}$ and consider the preference relation $\succ$ on $\Delta$ with $p\succ q$ if and only if $p(a) > \nicefrac23$ and $q(a)< \nicefrac13$, which does not even satisfy transitivity in a one-dimensional outcome space.\todo{FFX: The strict relation is transitive, though. We could write ``transitivity of the indifference relation''. FB: I would prefer to leave as it is for simplicity.}

Let $\mathcal{R}$ denote the set of all continuous, convex, and symmetric preference relations over $\Delta$.
Despite the richness of $\mathcal R$, preference relations therein admit a particularly nice representation. A preference relation can be represented by a skew-symmetric and bilinear (SSB) utility function $\phi\colon\Delta\times\Delta\rightarrow \mathbb R$ if, for all $p,q\in\Delta$,
\[p\succ q \text{ if and only if }\phi(p,q) > 0\text.\]
Skew-symmetric requires that $\phi(p,q) = - \phi(q,p)$ for all $p,q\in\Delta$ and bilinearity that $\phi$ is linear in both arguments.
SSB utility was introduced by \citet{Fish82c}, who also gave a complete characterization of preference relations representable by SSB functions.\footnote{Fishburn's characterization uses continuity and symmetry conditions that are weaker than ours while our convexity notion is weaker than his dominance axiom (see \Cref{app:ssb} for more details). Our axioms are arguably more intuitive and match the axioms used in \Cref{prop:rationalizable}.}
We prove the following alternative characterization in \Cref{app:ssb} by reducing it to Fishburn's characterization.

\begin{restatable}{proposition}{propssb}\label{prop:ssb}
	A preference relation $\succ$ can be represented by an SSB function if and only if it satisfies continuity, convexity, and symmetry.
\end{restatable}

SSB functions are unique up to scalar multiplications. 
We therefore write $\phi \equiv \hat\phi$ if and only if there is some $\alpha>0$ such that $\phi = \alpha \cdot \hat\phi$, i.e., if $\phi$ and $\hat\phi$ represent the same preferences.
We will also write ${\succ}\equiv \phi$ if ${\succ}$ is represented by the SSB function $\phi$.
Every preference relation ${\succ}\in \mathcal R$ other than complete indifference can be associated with a unique normalized SSB function on $\Delta\times\Delta$ whose largest positive value is equal to $1$. 
Let $\Phi$ denote the set of all SSB functions that are normalized in this way.\footnote{For vNM utility functions, this normalization boils down to the normalization of relative utilitarianism \citep{DhMe99a,BoCh15a}.}
Since all outcomes have finite support, $\phi(p,q)$ can be written as a convex combination of the values of $\phi$ for pure outcomes \citep{Fish84c}. For this purpose, we identify every alternative $a\in \mathcal U$ with the pure outcome that assigns probability $1$ to $a$. Then, for all $p,q\in\Delta$,
\[
	\phi(p,q) = \sum_{a,b\in \mathcal U} p(a)q(b)\phi(a,b)\text.
\]
We will often represent SSB functions restricted to $\Delta_X$ for $X\subseteq \mathcal U$ as skew-symmetric matrices in $\mathbb R^{X\times X}$.

When requiring transitivity on top of continuity, convexity, and symmetry, the four axioms characterize preference relations that can be represented by \emph{weighted linear (WL)} utility functions as introduced by \citet{Chew83a}.\footnote{A WL function is characterized by a linear utility function and a linear and positive weight function. An outcome $p$ is preferred to another outcome $q$ if the expected utility of $p$ divided by its weight is larger than the same quantity for $q$.
Thus, WL functions are more general than linear utility functions, as every linear utility function is equivalent to a WL function with constant weight function. See also \citet{Fish83a}.} 
We will denote this set of preference relations by $\mathcal R^{\wl}\subset \mathcal R$. 
When additionally requiring independence, then $\phi$ is separable, i.e., $\phi(p,q) = u(p) - u(q)$, where $u$ is a linear von Neumann-Morgenstern utility function representing $\succ$. The corresponding set of preference relations will be denoted by $\mathcal R^{\vnm}\subset \mathcal R^\wl$.
For independently distributed outcomes (as considered in this paper), SSB utility theory coincides with regret theory as introduced by \citet{LoSu82a} \citep[see also][]{LoSu87a,Blav06a}.

Through the representation of ${\succ}\in\mathcal R$ as a skew-symmetric matrix, it becomes apparent that the Minimax Theorem implies the existence of maximal elements of $\succ$ on $\Delta_X$. This was noted by \citet[][Theorem 4]{Fish84c} and already follows from \Cref{thm:sonnenschein}.

\section{Social Welfare Functions}\label{sec:swf}

In the remainder of this paper we deal with the problem of aggregating the preferences of multiple agents into a collective preference relation. The set of agents is $N = \{1,\dots,n\}$ for some $n\ge 2$. The preference relations of agents belong to some \emph{domain} $\mathcal D\subseteq \mathcal R$. A function $R\in\mathcal D^N$ from the set of agents to the domain is a preference profile. We will write preference profiles as tuples $(\succ_1,\dots,\succ_n)$ with indices in $N$. 
Given a preference profile $R$, let $N_{pq} = \{i\in N\colon p\succ_i q\}$ be the set of agents who strictly prefer $p$ over $q$. Also, let $I_{pq} = N\setminus(N_{pq}\cup N_{qp})$ be the set of agents who are indifferent between $p$ and $q$.
A \emph{social welfare function (SWF)} $f\colon \mathcal D^N\rightarrow \mathcal R$ maps a preference profile to a collective preference relation. When considering SWFs that may only map to collective preference relations in $\mathcal R^{\vnm}$, we will refer to SWFs with range $\mathcal R^{\vnm}$.

\citet{Arro51a} initiated the study of SWFs that satisfy two desirable properties: Pareto optimality and IIA. Pareto optimality prescribes that a unanimous preference of one outcome over another in the individual preferences should be reflected in the collective preference. Formally, an SWF $f$ satisfies \emph{Pareto optimality} if, for all $p,q\in\Delta$, $R\in\mathcal D^N$, and $f(R) = {\succ}$,
\begin{align*}
	\begin{aligned}
		&p\succsim_i q \text{ for all } i\in N \text{ implies } p \succsim q\text{, and}\\
		&\text{if additionally } p\succ_i q \text{ for some } i\in N \text{ then } p\succ q\text.
	\end{aligned}
	\tag{Pareto Optimality}
\end{align*}
The indifference part of Pareto optimality, which merely requires that $p \sim_i q$ for all $i\in N$ implies $p \sim q$, is usually referred to as \emph{Pareto indifference}.

Independence of irrelevant alternatives demands that collective preferences over some feasible set of outcomes should only depend on the individual preferences over this set (and not on the preferences over outcomes outside this set). 
In our framework, we will assume that feasible sets are based on the availability of alternatives and are therefore of the form $\Delta_X$ for $X\subseteq \mathcal U$.\footnote{Strengthening IIA by allowing \emph{all} convex subsets of $\Delta$ to be feasible completely ignores the structure of $\Delta$ as the convex hull of $\mathcal{U}$ and is too demanding when paired with our other axioms. Since pairwise utilitarianism violates this strong notion of IIA, \Cref{thm:swf} would turn into an impossibility. Our notion of IIA where only faces of $\Delta$ are feasible was also used in the impossibility theorem by \citet{KaSc77b} and, in an even weaker form, in the characterization of relative utilitarianism by \citet{DhMe99a}; see  \Cref{fn:ira}.}
Formally, we say that an SWF $f$ satisfies \emph{independence of irrelevant alternatives (IIA)} if, for all $R,\hat R\in\mathcal D^N$ and $X\subseteq \mathcal U$,
\[
	R|_{\Delta_X} = \hat R|_{\Delta_X} \text{ implies } f(R)|_{\Delta_X} = f(\hat R)|_{\Delta_X}\text.\tag{IIA}
\]

Any SWF that satisfies Pareto optimality and IIA will be called an \emph{Arrovian} SWF. Arrow has shown that, when no structure---such as convexity---is imposed on preference relations and feasible sets, every Arrovian SWF is dictatorial, i.e., there is $i\in N$ such that for all $p,q\in \Delta$, $R\in \mathcal{D}^N$, and $f(R)={\succ}$, $p \succ_i q$ implies $p \succ q$.
Dictatorships are examples of SWFs that are extremely biased towards one agent. In many applications, \emph{any} differentiation between agents is unacceptable and all agents should be treated equally. This property is known as anonymity. We denote by $\Pi_N$ the set of all permutations on $N$. For $\pi\in\Pi_N$ and a preference profile $R\in\mathcal D^N$, $R^\pi = R \circ \pi$ is the preference profile where agents are renamed according to $\pi$. Then, an SWF $f$ satisfies \emph{anonymity} if, for all $R\in\mathcal D^N$ and $\pi\in\Pi_N$,
\[
	f(R) = f(R^\pi)\text.\tag{Anonymity}
\]  
Anonymity is obviously a stronger requirement than non-dictatorship. 

Two straightforward aggregation rules that satisfy Pareto optimality, IIA, and anonymity are \emph{majority rule} ($p \succ q$ if and only if $|N_{pq}|>|N_{qp}|$) and \emph{Pareto rule} ($p \succ q$ if and only if $|N_{pq}|>0$ and $N_{qp}=\emptyset$). However, both rules do not constitute well-defined SWFs because they do not map to $\mathcal R$. While majority rule may not even produce maximal elements \citep{Zeck69a}, Pareto rule violates continuity and convexity. 

A natural subclass of SWFs can be defined by computing the weighted sum of the normalized individual utility representations. An SWF is called \emph{affine utilitarian} if and only if there are weights $w_1, \dots, w_n\in\mathbb{R}$ such that for all $R\in\mathcal{D}^N$ and $(\phi_i)_{i\in N}\in \Phi^N$ with $(\phi_i)_{i\in N}\equiv R$,
	\[f(R) \equiv \sum_{i\in N} w_i \phi_i \tag{Affine Utilitarianism}\text.\]
Affine utilitarian SWFs satisfy Pareto indifference. 
This still allows for constant SWFs (by setting all weights to $0$) or dictatorial SWFs (by setting all weights but one to $0$).
When requiring that all weights are positive, these SWFs are ruled out and all resulting SWFs satisfy Pareto optimality.
The unique anonymous and affine utilitarian SWF with positive weights is defined by setting all weights to $1$ and will simply be referred to as the \emph{utilitarian SWF}.
On the domain of vNM preferences, the utilitarian SWF coincides with \emph{relative utilitarianism} as introduced by \citet{DhMe99a}.
The utilitarian SWF violates IIA on domain $\mathcal R^\vnm$.
For example, for two agents and three alternatives, consider a preference profile where the first agent assigns normalized utilities $1$, $0$, and $0$ to the alternatives and the second agent assigns $\nicefrac{1}{3}$, $1$, and $0$.
Then, when deriving the collective preferences by employing utilitarianism, the first alternative is preferred to the second alternative. 
If instead the utility of the first agent for the second alternative was $\nicefrac{1}{2}$, the latter would be collectively preferred to the first alternative, even though the individual preferences over those two alternatives have not changed.
Consequently, the utilitarian SWF also violates IIA on the full domain $\mathcal R$.\footnote{\citeauthor{DhMe99a} introduced a weakening of IIA called \emph{independence of redundant alternatives}, which only considers feasible sets for which every infeasible outcome is \emph{unanimously} indifferent to some feasible outcome, and show that this condition is satisfied by the utilitarian SWF on domain $\mathcal R^\vnm$. However, it can be shown that this is no longer the case when considering the full domain $\mathcal R$.\label{fn:ira}}

\section{Characterization of the Domain}
\label{sec:domain}

When restricting attention to SWFs with range $\mathcal R^\vnm$, it is known that anonymous Arrovian aggregation on the full domain $\mathcal R$ is impossible because it is already impossible in the subdomain of vNM preferences (see \Cref{app:vnm}). 
On the other hand, interesting possibilities emerge in restricted domains such as in that of dichotomous vNM preferences $\mathcal R^\dich$ where each agent can only assign two different vNM utility values, say, $0$ and $1$ \citep{Inad69a}.
In \Cref{thm:dichdomain} (\Cref{app:dich}), we show that anonymous Arrovian aggregation of vNM preferences is \emph{only} possible in subdomains of $\mathcal R^\dich$.
In such domains, every affine utilitarian SWF satisfies IIA since for every $X\subseteq\mathcal{U}$, the individual preferences over all outcomes in $\Delta_X$ only depend on which alternatives in $X$ receive utility $1$. When all individual weights are positive, these SWFs furthermore satisfy Pareto optimality and thus constitute a natural class of Arrovian SWFs.
The utilitarian SWF corresponds to \emph{approval voting} and ranks pure outcomes by the number of approvals they receive from the agents \citep[see, e.g.,][]{BrFi07c}. This ranking is identical to majority rule, which happens to be transitive for dichotomous preferences, and is extended to all outcomes by comparing expected utilities. 
In \Cref{thm:dichswf} (\Cref{app:dich}), we prove that this SWF is in fact the \emph{only} anonymous Arrovian SWF with range $\mathcal R^\vnm$ when $\mathcal D\subseteq \mathcal R^\dich$ and $|\mathcal U|\ge 4$.\footnote{\label{fn:MaMo}\citet{MaMo15a} show a similar statement in the classic non-convex social choice setting. Since they only consider pure outcomes, their notions of Pareto optimality and IIA are weaker than ours. However, the consequence of their statement is also weaker because it only implies that pure outcomes are ranked according to their approval scores.}

 The goal of this section is to characterize the unique inclusion-maximal domain $\mathcal D\subseteq \mathcal R$ for which anonymous Arrovian SWFs exist.
To this end, we need to assume that $\mathcal D$ satisfies four richness conditions.
First, we require that it is neutral in the sense that it is not biased towards certain alternatives. 
For $\pi\in\Pi_\mathcal U$ and $p\in\Delta$, let $p^\pi\in\Delta$ such that $p^\pi(\pi(a)) = p(a)$ for all $a\in \mathcal U$.
Then, for ${\succ}\in\mathcal{R}$, we define ${\succ}^{\pi}$ such that $p^\pi\succ^{\pi} q^\pi$ if and only if $p\succ q$ for all $p,q\in\Delta$. 
It is assumed that 
\[
{\succ}\in\mathcal{D} \text{ if and only if }{\succ}^{\pi}\in\mathcal D \text{ for all }\pi\in\Pi_\mathcal U \text{ and } {\succ}\in\mathcal D\text. 
\tag{R1}
\label{eq:richneutral}
\]
Second, we require that it is possible to be completely indifferent, i.e.,
\[
	\emptyset\in\mathcal D\text.
	\tag{R2}
	\label{eq:richindiff}
\]
Third, it should also be possible for agents to declare completely opposed preferences. For ${\succ}\in \mathcal D$, ${\succ}^{-1}$ is the inverse of $\succ$, i.e., $p\succ^{-1} q$ if and only if $q\succ p$ for all $p,q\in\Delta$. Then,
\[
{\succ}\in\mathcal D \text{ implies }{\succ}^{-1}\in\mathcal D \text{ for all }{\succ}\in\mathcal R\text.
\tag{R3}
\label{eq:richinverse}
\]
Note that this condition is not implied by the previous neutrality condition because it allows the inversion of preferences over all outcomes, not only over pure outcomes.
Finally, we demand that for every preference relation in $\mathcal D$ and every set of up to four pure outcomes, there is a relation in $\mathcal D$ with the same preferences over pure outcomes such that these outcomes are all preferred to a fifth pure outcome.
\begin{align*}
	\begin{split}
		&\text{for all } {\hat\pref}\in\mathcal D\text{ and }X\subseteq \mathcal U\text{, }|X|\le 4 \text{, there is } {\pref}\in\mathcal D \text{ and }a\in\mathcal U\text{ such that }
		{\pref}|_X = {\hat\pref}|_X \text{ and}\\ &x\pref a\text{ for all } x\in X\text.
	\end{split}
	\tag{R4}
	\label{eq:richbottom}
\end{align*}
Any domain $\mathcal{D}\subseteq \mathcal{R}$ that satisfies~\ref{eq:richneutral},~\ref{eq:richindiff},~\ref{eq:richinverse}, and~\ref{eq:richbottom} is called \emph{rich}.
Note that any rich domain allows for arbitrary transitive preferences over up to five pure outcomes.

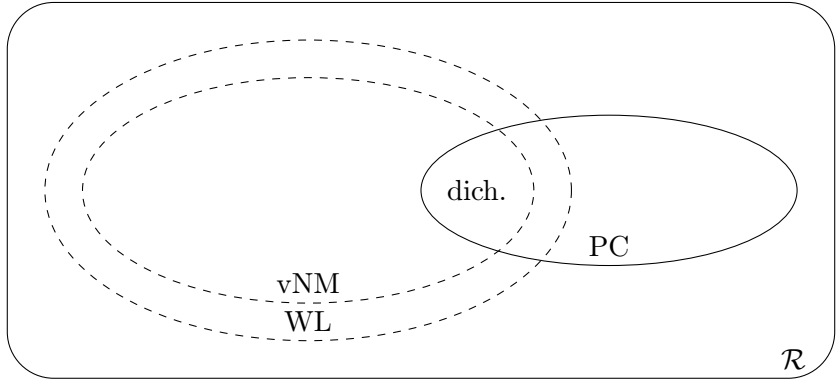
\begin{figure}[tb]
	\centering
	\begin{tikzpicture}[scale=0.5]
		\node[draw,rounded corners=0.618cm,minimum width=11cm,minimum height=5cm,label={[anchor=south east, above left, inner xsep=1em]south east:$\mathcal{R}$}] at (0,0) {};
		\node[draw,dashed,ellipse,minimum width=6cm,minimum height=3cm,label={[anchor=south,above]270:vNM}] at (-3,0) {};
		\node[draw,dashed,ellipse,minimum width=7cm,minimum height=4cm,label={[anchor=south,above]270:WL}] at (-3,0) {};		
		\node[draw,ellipse,minimum width=5cm,minimum height=2cm,label={[anchor=south,above]270:PC}] at (5,0) {};		
		\node (dich) at (1.5,0) {dich.};		
	\end{tikzpicture}
	\caption{Venn diagram showing the inclusion relationships between preference domains. The intersection of $\mathcal{R}^\vnm$ and $\mathcal{R}^\pc$ (pairwise comparison) contains exactly $\mathcal{R}^\dich$. The intersection of $\mathcal{R}^\wl$ (weighted linear utility) and $\mathcal{R}^\pc$ contains exactly the domain of PC preferences based on trichotomous weak orders (see \Cref{fig:pctrans} for an example). An example of a preference relation in $\mathcal{R}^\pc \setminus \mathcal{R}^\wl$ is given in \Cref{fig:pccycle}.
\Cref{thm:pcdomain} shows that $\mathcal{R}^\pc$ is the unique inclusion-maximal rich domain for which anonymous Arrovian aggregation is possible within $\mathcal{R}$. This, for example, implies impossibilities for $\mathcal{R}^\wl$ and $\mathcal{R}^\vnm$. 
	}
	\label{fig:venn}
\end{figure}
 
A rich domain that will turn out to be important for our characterization is defined as follows.
We say that $\phi\in\Phi$ is based on \emph{pairwise comparisons} if $\phi(a,b)\in\{-1,0,1\}$ for all $a,b\in \mathcal U$ and denote the set of SSB functions that are based on pairwise comparisons by $\Phi^{\pc}\subset \Phi$ and the corresponding set of preference relations by $\mathcal{R}^\pc=\{ {\succ} \in \mathcal{R} \colon {\succ} \equiv \phi \text{ for some } \phi\in\Phi^{\pc}\}$. $\mathcal{R}^\pc$ contains $\mathcal{R}^\dich$ (see \Cref{fig:venn}).
Other rich domains include $\mathcal R$, $\mathcal R ^\wl$, $\mathcal R ^\vnm$, and the subset of $\mathcal R^\pc$ in which all preferences over pure outcomes are transitive.

Anonymous Arrovian aggregation is only possible on rich subdomains of $\mathcal{R}^\pc$.

\begin{restatable}{theorem}{thmpcdomain}\label{thm:pcdomain}
	Let $f$ be an anonymous Arrovian SWF on some rich domain $\mathcal{D}$ with $|\mathcal U|\ge 4$.
	Then, $\mathcal{D}\subseteq \mathcal{R}^{\pc}$.
\end{restatable}

The proof of \Cref{thm:pcdomain} is given in \Cref{app:domain}. 

\pc preferences are quite natural and can be seen as the canonical SSB representation consistent with a given ordinal preference relation over alternatives. For a preference relation ${\succ}\in \mathcal R^\pc$ and two outcomes $p,q\in\Delta$ we have that 
\[
p \succ q \quad\text{if and only if}\quad
\sum\limits_{a,b\colon a \succ b} p(a)\cdot q(b) > \sum\limits_{a,b\colon a \succ b} q(a)\cdot p(b)\text.
\]
If $p$ and $q$ are interpreted as independent lotteries, $p$ is preferred to $q$ if and only if $p$ is more likely to return a more preferred alternative than $q$. Alternatively, the terms in the inequality above can be associated with the probability of \emph{ex ante} regret. Then, $p$ is preferred to $q$ if its choice results in less \emph{ex ante} regret.
Since \pc preferences are completely determined by preferences over pure outcomes and transitive preferences over pure outcomes can be conveniently represented by weak rankings, we will compactly represent \pc preferences over some set of alternatives by putting the weak ranking of these alternatives in brackets (see Figures \ref{fig:pctrans} and \ref{fig:pccycle} for examples).

\begin{figure}[tb]
	\centering
	$
	\begin{bmatrix}
		a\\b\\c
	\end{bmatrix}\equiv
			\phi =
			\begin{blockarray}{cccc}
				a & b & c\\
				\begin{block}{(ccc)c}
		  		  	0 & 1 & 1 & a \\
		  			-1 & 0 & 1 & b \\
		  	  	  	-1 & -1 & 0 & c \\
				\end{block}
			\end{blockarray}
	$\hspace{14ex}
		\begin{tikzpicture}[baseline=(g.base),every node/.style={inner sep=0,outer sep=0}]
			\def\r{15ex}
			\def\dist{3ex}
			\def\arrow{1.5ex}
			\tikzstyle{myline}=[line width = .08ex]
		
			\node[label={[label distance=\dist]-150:$a$}] (a) at (-150:\r) {};
			\node[label={[label distance=\dist]90:$b$}] (b) at (90:\r) {};
			\node[label={[label distance=\dist]-30:$c$}] (c) at (-30:\r) {};
		
			\coordinate (f) at ($(-90:2*\r)$) {};
		
			\node (g) {};
		
			\draw (a.center) -- (b.center) -- (c.center) -- (a.center);
			
			\coordinate (q0) at (intersection of a--c and b--f);
			\draw[myline] (b) -- (q0);

			\coordinate (o1) at ($(f) + (110:4*\r)$);
			\coordinate (p1) at (intersection of f--o1 and a--b);
			\coordinate (q1) at (intersection of a--c and p1--f);
			\draw[myline] (p1) -- (q1);

			\coordinate (o2) at ($(f) + (70:4*\r)$);
			\coordinate (p2) at (intersection of f--o2 and c--b);
			\coordinate (q2) at (intersection of a--c and p2--f);
			\draw[myline] (p2) -- (q2);

			\coordinate (o3) at ($(f) + (100:4*\r)$);
			\coordinate (p3) at (intersection of f--o3 and a--b);
			\coordinate (q3) at (intersection of a--c and p3--f);
			\draw[myline] (p3) -- (q3);

			\coordinate (o4) at ($(f) + (80:4*\r)$);
			\coordinate (p4) at (intersection of f--o4 and c--b);
			\coordinate (q4) at (intersection of a--c and p4--f);
			\draw[myline] (p4) -- (q4);

			\foreach \x in {110,100,90,80,70}{
			
				\draw[myline,-latex] ($(f) + (\x:1.85*\r)$) -- ($(f) + (\x:1.85*\r) + (\x-90:\arrow)$);
			
			}
		\end{tikzpicture}
	\caption{Illustration of preferences based on pairwise comparisons for three alternatives when preferences over pure outcomes are given by the transitive relation $a\succ b\succ c$. The left-hand side shows the preference relation and the SSB function and the right-hand side the Marschak-Machina probability triangle. The arrows represent  normal vectors to the indifference curves (pointing towards the lower contour set). Each indifference curve separates the corresponding upper and lower contour set.}
	\label{fig:pctrans}
\end{figure}
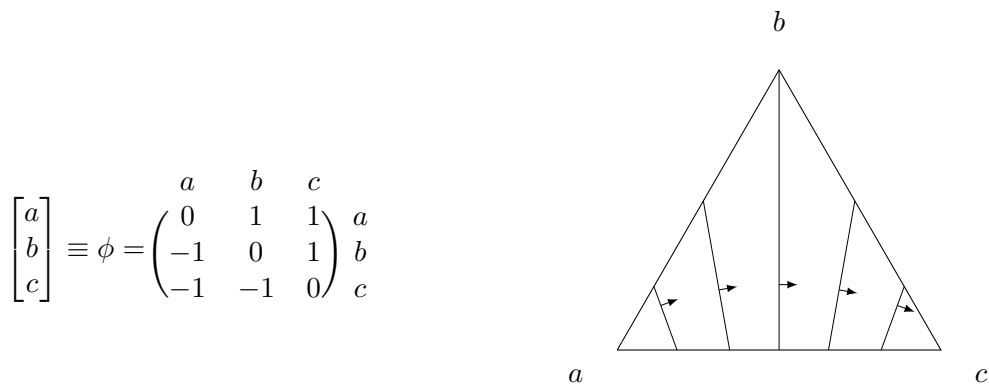

\pc preferences have previously been considered in decision theory \citep{Blyt72a,Pack82a,Blav06a}. \citet{Pack82a} calls them the \emph{rule of expected dominance} and \citet{Blav06a} refers to them as a \emph{preference for the most probable winner}. \citet{ABB14b,ABBB15a} and \citet{BBH15c} have studied Pareto efficiency, strategyproofness, and related properties with respect to these preferences. \citet{Blav06a} gives an axiomatic characterization of \pc preferences using the SSB axioms (continuity, convexity, and symmetry), and an additional axiom called \emph{fanning-in}, which essentially prescribes that indifference curves are not parallel, but fanning in at a certain rate (see \Cref{fig:pctrans}).
As a corollary of \Cref{thm:pcdomain}, fanning-in is implied by Fishburn's SSB axioms and Arrow's axioms.
 \citeauthor{Blav06a} cites extensive experimental evidence for the fanning-in of indifference curves.

\Cref{fig:pctrans} illustrates \pc preferences for three transitively ordered pure outcomes.\footnote{For three alternatives, \pc preferences as depicted in \Cref{fig:pctrans} can be represented by a WL function with utility function $u(a) = u(b) = 1$ and $u(c) = 0$ and weight function $w(a) = 0$ and $w(b) = w(c) = 1$.}
When there are at least four alternatives, \pc preferences can be cyclic even when preferences over pure outcomes are transitive. This phenomenon, known as the \emph{Steinhaus-Trybula paradox}, is illustrated in \Cref{fig:pccycle} \citep[see, e.g.,][]{StTr59a,Blyt72a,Pack82a,RuSe12a,BuPo18a}. \citet{BuPo18a} have conducted an extensive experimental study of the Steinhaus-Trybula paradox and found significant evidence for \pc preferences. 

\begin{figure}[tb]
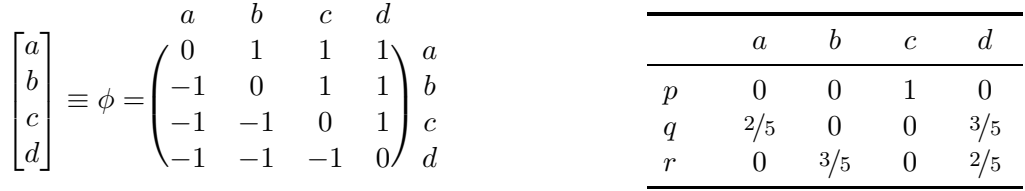

	\centering
	$
		\begin{bmatrix}
			a\\b\\c\\d
		\end{bmatrix}\equiv
		\phi =
		\begin{blockarray}{ccccc}
			a & b & c & d\\
			\begin{block}{(cccc)c}
	  		  	0 & 1 & 1 & 1& a \\
	  			-1 & 0 & 1 & 1& b \\
	  	  	  	-1 & -1 & 0 & 1& c \\
				-1 & -1 & -1 & 0 &d\\
			\end{block}
		\end{blockarray}
	$\hspace{14ex}
	\newcolumntype{C}{>{\centering\let\newline\\\arraybackslash\hspace{0pt}}m{\widthof{999}}} 
	\begin{tabular}{p{\widthof{999}} C C C C} 
	\toprule
	& $a$ & $b$ & $c$ & $d$\\
	\midrule
	$p$ & $0$ & $0$ & $1$ & $0$\\
	$q$ & $\nicefrac{2}{5}$ & $0$ & $0$ & $\nicefrac{3}{5}$ \\
	$r$ & 0 & $\nicefrac{3}{5}$ & $0$ & $\nicefrac{2}{5}$ \\
	\bottomrule
	\end{tabular}
	\caption{Illustration of preferences based on pairwise comparisons for four alternatives when preferences over pure outcomes are given by the transitive relation $a\succ b\succ c\succ d$. The left-hand side shows the preference relation and the SSB function. The preferences between the three outcomes $p$, $q$, and $r$, defined in the table on the right-hand side, are cyclic: $\phi(p,q)=\nicefrac{3}{5} - \nicefrac{2}{5}=\nicefrac{1}{5}>0$, $\phi(q,r)=\nicefrac{2}{5} - (\nicefrac{3}{5})^2 =\nicefrac{1}{25}>0$, and $\phi(r,p)=\nicefrac{3}{5} - \nicefrac{2}{5}=\nicefrac{1}{5}>0$. Hence, $p\succ q\succ r\succ p$.
}
	\label{fig:pccycle}
\end{figure}

\section{Characterization of the Social Welfare Function}

\Cref{thm:pcdomain} has established that anonymous Arrovian aggregation is only possible if individual preferences are based on pairwise comparisons, i.e., $\mathcal D\subseteq \mathcal R^\pc$. This raises the question \emph{which} SWFs (if any) are Arrovian on $\mathcal D$. 

It turns out that for any ${\succ}\in \mathcal D$, ${\pref}\equiv \phi\in\Phi$, and $X\subseteq \mathcal U$, ${\succ}|_{\Delta_X}$ uniquely determines $\phi|_X$ (not only up to a positive scalar).
Hence, affine utilitarian SWFs satisfy IIA.
Since any affine utilitarian SWF with positive weights furthermore satisfies Pareto optimality, these SWFs are Arrovian. 
Our next theorem shows that these are indeed the \emph{only} Arrovian SWFs. More precisely, we show that SWFs on domain $\mathcal D$ satisfy Pareto indifference and IIA if and only if they are affine utilitarian. 
Affine utilitarian SWFs may assign negative or null weights to agents. 
As mentioned in \Cref{sec:swf}, this, for example, allows for dictatorial SWFs where the collective preferences are identical to the preference relation of one pre-determined agent.
However, when assuming full Pareto optimality, the weights assigned to these SSB functions have to be positive, which rules out dictatorial SWFs.

\begin{restatable}{theorem}{thmswf}\label{thm:strongpareto}\label{thm:swf}
	Let $f$ be an Arrovian SWF on some rich domain $\mathcal{D}\subseteq \mathcal{R}^\pc$ with $|\mathcal U|\ge 5$.
	Then, $f$ is affine utilitarian with positive weights.
\end{restatable}

The proof of \Cref{thm:swf} is given in \Cref{app:swf}.

On subdomains of $\mathcal R^\pc$, affine utilitarianism with positive weights admits an intuitive probabilistic interpretation: in order to compare two lotteries $p$ and $q$, randomly sample a pure outcome $a$ from $p$, a pure outcome $b$ from $q$, and an agent $i$ with probabilities proportional to the agents' weights. Then, $p$ is collectively preferred to $q$ if and only if the probability that agent $i$ prefers $a$ to $b$ is greater than the probability that he prefers $b$ to $a$.

\Cref{thm:swf} can be seen as a multi-profile version of Harsanyi's Social Aggregation Theorem (see \Cref{sec:related}) for SSB preferences, where IIA allows us to connect weights across different profiles.
When furthermore assuming anonymity, the weights of all SSB functions have to be identical and we obtain the following complete characterization.

\begin{corollary}\label{cor:ml}
	Let $|\mathcal U|\ge 5$ and $\mathcal D$ be a rich domain.
	An anonymous SWF is Arrovian if and only if it is the utilitarian SWF and $\mathcal{D}\subseteq \mathcal{R}^\pc$.
\end{corollary}

\begin{figure}[tb]
	\centering
	$
			\phi =
			\begin{blockarray}{cccc}
				a & b & c\\
				\begin{block}{(ccc)c}
		  		  	0 & 1 & -1 & a \\
		  			-1 & 0 & 1 & b \\
		  	  	  	1 & -1 & 0 & c \\
				\end{block}
			\end{blockarray}
	$\hspace{14ex}
	\begin{tikzpicture}[baseline=(g.base),every node/.style={inner sep=0,outer sep=0}]
		\def\r{15ex}
		\def\dist{3ex}
		\def\arrow{1.5ex}
		
		\tikzstyle{myline}=[line width = .08ex]
		
		\node[label={[label distance=\dist]-150:$a$}] (a) at (-150:\r) {};
		\node[label={[label distance=\dist]90:$b$}] (b) at (90:\r) {};
		\node[label={[label distance=\dist]-30:$c$}] (c) at (-30:\r) {};
		
		\node (f) {\tiny\textbullet};
		
		\draw (a.center) -- (b.center) -- (c.center) -- (a.center);

		\foreach \x in {10,50,90,130,170,210,250,290,330}{

		\ifthenelse{\x > -1}{
		\ifthenelse{\x < 30}{

			\coordinate (p) at ($(f) + (\x:2*\r)$);
			\coordinate (q) at (intersection of f--p and c--b);
			\coordinate (r) at (intersection of f--p and a--b);
			\draw[myline] (q) -- (r);

		}{}
		}{}
		
		\ifthenelse{\x > 29}{
		\ifthenelse{\x < 90}{

			\coordinate (p) at ($(f) + (\x:2*\r)$);
			\coordinate (q) at (intersection of f--p and c--b);
			\coordinate (r) at (intersection of f--p and a--c);
			\draw[myline] (q) -- (r);

		}{}
		}{}
		
		\ifthenelse{\x > 89}{
		\ifthenelse{\x < 170}{

			\coordinate (p) at ($(f) + (\x:2*\r)$);
			\coordinate (q) at (intersection of f--p and a--b);
			\coordinate (r) at (intersection of f--p and a--c);
			\draw[myline] (q) -- (r);

		}{}
		}{}
		
		\ifthenelse{\x > 149}{
		\ifthenelse{\x < 210}{

			\coordinate (p) at ($(f) + (\x:2*\r)$);
			\coordinate (q) at (intersection of f--p and a--b);
			\coordinate (r) at (intersection of f--p and b--c);
			\draw[myline] (q) -- (r);

		}{}
		}{}
		
		\ifthenelse{\x > 209}{
		\ifthenelse{\x < 270}{

			\coordinate (p) at ($(f) + (\x:2*\r)$);
			\coordinate (q) at (intersection of f--p and a--c);
			\coordinate (r) at (intersection of f--p and b--c);
			\draw[myline] (q) -- (r);

		}{}
		}{}
		
		\ifthenelse{\x > 269}{
		\ifthenelse{\x < 330}{

			\coordinate (p) at ($(f) + (\x:2*\r)$);
			\coordinate (q) at (intersection of f--p and a--c);
			\coordinate (r) at (intersection of f--p and b--a);
			\draw[myline] (q) -- (r);

		}{}
		}{}
		
		\ifthenelse{\x > 329}{
		\ifthenelse{\x < 361}{

			\coordinate (p) at ($(f) + (\x:2*\r)$);
			\coordinate (q) at (intersection of f--p and b--c);
			\coordinate (r) at (intersection of f--p and b--a);
			\draw[myline] (q) -- (r);

		}{}
		}{}
		
		\coordinate (y) at (\x:5ex);
		\draw[myline,-latex] (y) -- ($(y) + (\x-90:\arrow)$);
		
		}
			
	\end{tikzpicture}

	\caption{Illustration of collective preferences returned by the unique anonymous Arrovian SWF in the case of Condorcet's paradox. The left-hand side shows the collective SSB function and the right-hand side the Marschak-Machina probability triangle.
	The arrows represent normal vectors to the indifference curves (pointing towards the lower contour set). Each indifference curve separates the corresponding upper and lower contour set.
	The unique most preferred outcome is $\nicefrac{1}{3}\, a + \nicefrac{1}{3}\, b + \nicefrac{1}{3}\, c$.
	}
	\label{fig:cond}
\end{figure}
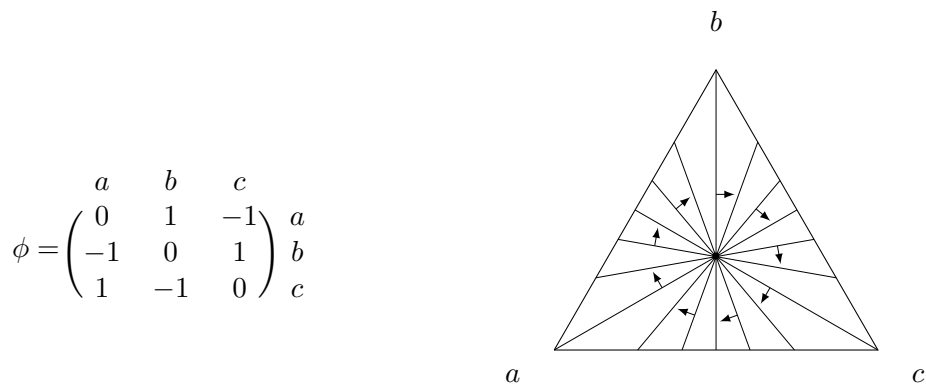

We refer to the utilitarian SWF on PC preferences as \emph{pairwise utilitarianism}. 
Pairwise utilitarianism is computationally tractable: two outcomes can be compared by straightforward matrix-vector multiplications while a maximal outcome can be found using linear programming.
For illustrative purposes, consider the classic Condorcet preference profile 
\[
		R=(
		\begin{bmatrix}
			a\\b\\c
		\end{bmatrix},
		\begin{bmatrix}
			b\\c\\a
		\end{bmatrix},
		\begin{bmatrix}
			c\\a\\b
		\end{bmatrix}
		)
		\equiv
		(
			\begin{pmatrix}
	  		  	0 & 1 & 1 \\
	  			-1 & 0 & 1 \\
	  	  	  	-1 & -1 & 0 
			\end{pmatrix},
			\begin{pmatrix}
	  		  	0 & -1 & -1 \\
	  			1 & 0 & 1 \\
	  	  	  	1 & -1 & 0 
			\end{pmatrix},
			\begin{pmatrix}
	  		  	0 & 1 & -1 \\
	  			-1 & 0 & -1 \\
	  	  	  	1 & 1 & 0
			\end{pmatrix}
		)=(\phi_1,\phi_2,\phi_3)\text.
\]
Note that the pairwise majority relation is cyclic, since there are majorities for $a$ over $b$, $b$ over $c$, and $c$ over $a$. The unique anonymous Arrovian SWF $f$ aggregates preferences by adding the individual utility representations, i.e.,
\[
		f(R) \equiv 
		\sum_{i\in N} \phi_i=
			\begin{pmatrix}
	  		  	0 & 1 & -1 \\
	  			-1 & 0 & 1 \\
	  	  	  	1 & -1 & 0 
			\end{pmatrix}
		\text.
\]
\Cref{fig:cond} shows the collective preference relation represented by this matrix.
The unique most preferred outcome is $\nicefrac{1}{3}\, a + \nicefrac{1}{3}\, b + \nicefrac{1}{3}\, c$.\footnote{This outcome represents a somewhat unusual unique maximal element because it is not \emph{strictly} preferred to any of the other outcomes. This is due to the contrived nature of the example and only happens if the support of a maximal outcome contains all alternatives. Also, in this example, collective preferences happen to be \pc preferences, which clearly is not the case in general.}

\section{Discussion}

Our results challenge the traditional---transitive---way of thinking about individual and collective preferences, which has been largely influenced by the pervasiveness of scores and grades. While our theorems do hold for transitive individual preferences over \emph{pure} outcomes, the preference domain we characterize does admit preference cycles over \emph{mixed} outcomes. Pareto optimality then immediately implies the same for collective preferences.
Even though the collective preference relation we characterize does not provide a ranking of all possible outcomes, it nevertheless allows for the comparison of arbitrary pairs of outcomes and identifies maximal (and minimal) elements in each feasible set of outcomes.\footnote{\citet{BeRa09a} have recently also put forward a relaxed---intransitive---notion of welfare and defended it as ``a viable welfare criterion'' because it guarantees the existence of maximal elements for finite sets. In fact, \citeauthor{BeRa09a} write that ``to conduct useful welfare analysis, one does not require transitivity'' \cite[see also][]{Bern09a}.}

Some readers may be concerned by, say, deriving the decisions of a government from an intransitive collective preference relation. We believe that this concern is largely based on the common fallacy of equating transitivity with rationality. Transitivity certainly appears to be a desirable property of preference relations. However, in notoriously difficult settings such as social choice, it can be unnecessarily restrictive: \emph{unnecessary} because basic principles of \emph{rational} choice (such as Propositions \ref{thm:sonnenschein} and \ref{prop:rationalizable}) hold without making this assumption and \emph{restrictive} because \emph{collective} choice is impossible when insisting on transitivity. Arrow's theorem implies that every Pareto optimal social choice function \emph{(i)} cannot be rationalized by a transitive collective preference relation or \emph{(ii)} takes into account irrelevant information (by violations of IIA). Consider, for example, Borda's voting rule which assigns equidistant scores to alternatives based on the agents' individual rankings and returns the alternatives with maximal accumulated score. \citet[][pp.~78f]{Sen77a} provides an illuminating discussion of two different interpretations of Borda's rule for variable feasible sets that highlight the tension between (transitive) rationalizability and IIA. In a choice-theoretic context, IIA demands that the choice set only depends on preferences over elements contained in the feasible set, and rationalizability requires that all choices can be rationalized by a single collective preference relation ranging over all alternatives in the universe.
Now, the \emph{broad} Borda rule first assigns Borda scores to all alternatives in the universe and then returns the alternatives with maximal scores within the feasible set. By contrast, the \emph{narrow} Borda rule directly assigns Borda scores to alternatives in the feasible set and then returns those with maximal score. The broad Borda rule can be rationalized by a transitive collective preference relation (the ranking of all alternatives by their Borda score), but clearly violates IIA while the narrow Borda rule satisfies IIA, but cannot be rationalized by any binary preference relation (it violates contraction consistency). Arrow’s theorem shows that this tradeoff concerns \emph{all} Pareto optimal social choice functions. Moreover, Sen observed that transitivity rationalizability can be replaced with contraction consistency in Arrow's theorem and many related results.
While voters could justifiably complain that, under the broad Borda rule, the social choice from feasible set $\{a,b,c\}$ depends on their preferences over other unrelated alternatives, say, $d$ or $e$ (a violation of IIA), they could be similarly concerned about the narrow Borda rule, under which it is possible that alternative \{a\} is chosen from $\{a,b,c\}$ but not from $\{a,b\}$. Both failures are troubling: the lack of IIA because seemingly irrelevant information is taken into account for the social choice, and the lack of contraction because introducing or removing, say, clearly inferior alternatives can influence the social choice.

Within the setting of convex outcome sets as described in our paper, affine utilitarian social choice functions on the domain of \pc preferences simultaneously satisfy IIA and rationalizability by a binary preference relation (and thus contraction) as well as Pareto optimality.
If one takes offense at intransitivities, these functions can also be interpreted as mappings from individual choice functions (which may already be aggregates of individual opinions) to a collective choice function.
\todo{FB: How should we define Pareto optimality here or should we omit this sentence altogether? FFX: I would prefer to keep this sentence but omit PO. The discussion above is concerned with the conflict between IIA and rationalizability, so PO does not seem very relevant here. FB: Ok, but let's think about this model again when preparing the final version.}
A compelling opinion on transitivity, which matches the narrative of our paper, is expressed in the following quote by decision theorist Peter C. Fishburn:

\begin{quote}
Transitivity is obviously a great practical convenience and a nice thing to have for mathematical purposes, but long ago this author ceased to understand why it should be a cornerstone of normative decision theory. [\dots] The presence of intransitive preferences complicates matters [\dots] however, it is not cause enough to reject intransitivity. An analogous rejection of non-Euclidean geometry in physics would have kept the familiar and simpler Newtonian mechanics in place, but that was not to be. Indeed, intransitivity challenges us to consider more flexible models that retain as much simplicity and elegance as circumstances allow. It challenges old ways of analyzing decisions and suggests new possibilities.\hfill \citep[][pp.~115--117]{Fish91a}
\end{quote}

\Cref{thm:swf}, the main result of this paper, can be viewed as an intermediary between Harsanyi's Social Aggregation Theorem and Arrow's Impossibility Theorem: it uses Arrow's axioms to derive Harsanyi's utilitarian consequence. 
Clearly, the form of utilitarianism characterized in \Cref{thm:swf} is rather restrictive as, due to \Cref{thm:pcdomain}, it does not allow for intensities of individual preferences.\footnote{One may even question whether this form of preference aggregation really qualifies as \emph{utilitarianism}. However, in a similar vein, one could also question whether Harsanyi's Social Aggregation Theorem 
or Dhillon and Mertens' relative utilitarianism 
are concerned with cardinal utilitarianism because vNM utilities are merely a compact representation of ordinal preferences over lotteries \citep[see, e.g.,][for a discussion of this issue in the context of the so-called Harsanyi-Sen debate]{Weym91a,FSW08a,FlMo16a,MoPi16a}. \Citet[][p. 16]{vNM47a} themselves warn against cardinal interpretations of their utility theory (see also ``Fallacy 3'' by \citet[][p.~32]{LuRa57a} and \citet{Fish89b}).}
In fact, it is no more ``utilitarian'' than approval voting or Borda's rule, which are also based on the summation of scores in purely ordinal contexts. In contrast to Borda's rule, however, pairwise utilitarianism respects majority rule on pure outcomes and thereby reconciles Borda's and Condorcet's seemingly conflicting views on preference aggregation \citep[see, e.g.,][]{Blac58a,Youn88a,Youn95a}.
While \Cref{thm:pcdomain} shows that Arrow's axioms rule out intensities of \emph{individual} preferences over pure outcomes, \Cref{thm:swf} implies that intensities of \emph{collective} preferences are required.

To conclude, we would like to highlight a remarkable quote from Kenneth J. Arrow's influential monograph, which draws the reader's attention precisely to the 
avenue pursued in this paper.

\begin{quote}
It seems that the essential point is, and this is of general bearing, that, if conceptually we imagine a choice being made between two alternatives, we cannot exclude any probability distribution over those two choices as a possible alternative. The precise shape of a formulation of rationality which takes the last point into account or the consequences of such a reformulation on the theory of choice in general or the theory of social choice in particular cannot be foreseen; but it is at least a possibility, to which attention should be drawn, that the paradox to be discussed below might be resolved by such a broader concept of rationality [\dots]
Many writers have felt that the assumption of rationality, in the sense of a one-dimensional ordering of all possible alternatives, is absolutely necessary for economic theorizing [\dots] There seems to be no logical necessity for this viewpoint; we could just as well build up our economic theory on other assumptions as to the structure of choice functions if the facts seemed to call for it.
\hfill \citep[][pp.~20--21]{Arro51a}
\end{quote} 

\section{Remarks}

This section contains a number of technical remarks concerning Theorems \ref{thm:pcdomain} and \ref{thm:swf}.
\medskip

\begin{remark}[Transitivity]\label{rem:trans}
When requiring transitivity of individual preferences over \emph{all} outcomes, we immediately obtain an impossibility because $|\mathcal U|\ge 4$ and~\ref{eq:richbottom} imply that we have to admit a strict ranking of four pure outcomes. According to \Cref{thm:pcdomain}, these preferences are extended to all outcomes using the PC extension. The example given in \Cref{fig:pccycle} shows that these preferences violate transitivity.
Hence, we also have the impossibility of anonymous Arrovian aggregation of WL preferences (and thereby of vNM preferences), even when collective preferences need not be transitive.\footnote{When collective preferences have to be transitive as well, this impossibility directly follows from Arrow's theorem by only considering pure outcomes. Pareto optimality and IIA only become weaker while non-dictatorship is strengthened (dictators only need to be able to dictate strict preferences over pure outcomes). The latter is implied by anonymity.}
\end{remark}

\begin{remark}[Anonymity]
	\Cref{thm:pcdomain} does not hold without assuming anonymity.
	Let $\mathcal U = \{a,b,c,d\}$, $N = \{1,2,3\}$, 
	\[
		\phi=
		\begin{pmatrix}
			0 & 1 & 1 & 1+\epsilon\\
			-1 & 0 & 1 & 1\\
			-1 & -1 & 0 & 1\\
			-(1+\epsilon) & -1 & -1 & 0
		\end{pmatrix}
	\]	
	for some $\epsilon \in(0,\nicefrac{1}{4})$, and $\mathcal D \equiv \Phi^\pc \cup \{ \phi^\pi \colon \pi\in \Pi_\mathcal U \}$. $\mathcal D$ satisfies our richness assumptions and the SWF $f\colon\mathcal D\rightarrow \mathcal R$, $f(R) \equiv 2\phi_1 + 3\phi_2 + 4\phi_3$ satisfies Pareto optimality and IIA, but violates anonymity.
	Note that $f$ is not dictatorial.
	Hence, \Cref{thm:pcdomain} does not hold when weakening anonymity to non-dictatorship.
\end{remark}

\begin{remark}[Tightness of Bounds]\label{rem:bounds}
	\Cref{thm:pcdomain} does not hold if $|\mathcal U|<4$, which is the same bound as for the result by \citet{KaSc77a}.
	This stems from the fact that for $|\mathcal U| = 3$, IIA only has consequences for feasible sets of the form $\Delta_{\{a,b\}}$ for some $a,b\in \mathcal U$.
	For every possible preference between $a$ and $b$, there is exactly one SSB preference relation on $\Delta_{\{a,b\}}$ consistent with it.
	Hence, IIA only has consequences for the collective preferences over \emph{pure} outcomes.
	However, even for three alternatives, the domains of preferences satisfying~\ref{eq:richneutral},~\ref{eq:richindiff}, and~\ref{eq:richinverse} which allow for anonymous Arrovian aggregation are severely restricted. 
	In particular, \Cref{lem:indiff,lem:semidecisive,lem:uniquerelation} and the cases~\ref{item:domain1} and~\ref{item:domain4} of \Cref{lem:pcdomain} still hold.
	Any such domain contains exactly one SSB function $\phi$ for every strict order over $\mathcal U$, which takes the form
	\[
		\phi = 
		\begin{pmatrix}
			0 & 1 & \lambda \\
			-1 & 0 & 1\\
			-\lambda & -1 & 0\\
		\end{pmatrix}
	\]
	for some $\lambda\in\mathbb R_{>0}$ that is fixed across all strict orders.
	For $1 < \lambda < 1 + \nicefrac{1}{n}$, the utilitarian SWF constitutes an Arrovian SWF on the corresponding domain.

	\Cref{thm:swf} does not hold if $|\mathcal U|< 5$.
	Let $\mathcal U = \{a,b,c,d\}$, $\mathcal D = \mathcal{R}^\pc$,
	\[
		\hat \phi = 
		\begin{pmatrix}
			0 & 1 & 0 & 0\\
			-1 & 0 & 0 & 0\\
			0 & 0 & 0 & 1\\
			0 & 0 & -1 & 0
		\end{pmatrix}\text{, and }
	\hat R=(
	\begin{bmatrix}
		a\\b\\c\\d
	\end{bmatrix},
	\begin{bmatrix}
		a\\b\\c\\d
	\end{bmatrix},
	\begin{bmatrix}
		c\\d\\a\\b
	\end{bmatrix},
	\begin{bmatrix}
		d\\c\\a\\b
	\end{bmatrix}
	,\dots)
	\]
	such that every preference relation in $\mathcal D\setminus\{\emptyset\}$ appears exactly once in $(\hat R_i)_{i\in N\setminus\{1,2,3,4\}}$.
	Then, Pareto optimality has no implications for $\hat R$.
	Let $f\colon\mathcal D^N\rightarrow \mathcal R$ be the utilitarian SWF except that $f(\hat R) \equiv \hat\phi$.
	$f$ satisfies Pareto optimality and IIA.
	The proof of \Cref{thm:swf} fails at \Cref{lem:mlscheme2}.
\end{remark}

\begin{remark}[SWFs with Range $\mathcal D$]
When defining SWFs by requiring that collective preferences have to belong to the same domain as individual preferences, one obtains an impossibility. \Cref{thm:pcdomain} only becomes weaker if we restrict the range of SWFs. Hence, both individual and collective preferences have to belong to $\mathcal R^\pc$ and \Cref{thm:swf} implies that any Arrovian SWF is affine utilitarian with positive weights. However when letting
	\[
	\phi_1=
	\begin{pmatrix}
		0 & 1 & 1 \\
		-1 & 0 & 0 \\
		-1 & 0 & 0 \\
	\end{pmatrix}
	\text{and }
	\phi_2=
	\begin{pmatrix}
		0 & 0 & 1 \\
		0 & 0 & 1 \\
		-1 & -1 & 0 \\
	\end{pmatrix}
	\text{,}
	\]
	then $\phi_1,\phi_2\in \Phi^\pc$ while there is no $w_1, w_2\in \mathbb{R}_{>0}$ such that $w_1 \phi_1+ w_2 \phi_2\in \Phi^\pc$. Since $\phi_1$ and $\phi_2$ represent dichotomous preference relations, \Cref{thm:dichswf} also turns into an impossibility.
\end{remark}

\begin{remark}[More General Preference Relations]
Our results could be strengthened by making even less assumptions about individual and collective preferences. 
Whether symmetry is required for \Cref{thm:pcdomain} and \Cref{thm:swf} is open. 
A more drastic generalization would only require the existence of maximal elements in all feasible sets with respect to the \emph{collective} preference relation. Such a generalization of \Cref{thm:pcdomain} does not hold. Consider the domain of individual preferences containing all transitive and complete relations. Pareto rule (see \Cref{sec:swf}) is anonymous and Arrovian and always returns a relation that permits maximal elements. Hence, \Cref{thm:swf} does not hold either. 
For the full domain of \emph{individual} preferences that admit maximal elements in all feasible sets, Pareto optimality alone can result in collective preferences that do not admit a maximal element.
\end{remark}

\section*{Acknowledgements}
This material is based on work supported by the Deutsche Forschungsgemeinschaft under grant {BR~2312/11-1}. Early versions of this paper were presented at the Yale Microeconomic Theory Reading Group (September, 2016) and the
5th Annual D-TEA Workshop in Paris, dedicated to the memory of Kenneth J. Arrow (May, 2017).
The authors thank Tilman B{\"o}rgers, Paul Harrenstein, Ehud Kalai, Michel Le Breton, Philippe Mongin, Klaus Nehring, Dominik Peters, Mirek Pi{\v s}t{\v e}k, Marcus Pivato, Clemens Puppe, Larry Samuelson, and Gerelt Tserenjigmid for helpful comments. 
The authors are furthermore indebted to the editor and anonymous referees for their constructive feedback.

\def\bibfont{\small}

\newpage

\appendix

\section*{APPENDIX}

\section{Arrovian Impossibilities for vNM Preferences}\label{app:vnm}

As mentioned in \Cref{sec:related}, there are a number of Arrovian impossibilities when preferences over lotteries satisfy the von Neumann-Morgenstern (vNM) axioms and thus can be represented by assigning cardinal utilities to alternatives such that lotteries are compared based on the expected utility they produce. We believe that a detailed comparison of these results which have appeared in different branches of social choice theory and welfare economics is in order.

The literature on \emph{economic domains} uses a framework very similar to the one studied in this paper \citep[see][]{LeWe11a}. A key question is whether Arrow's impossibility remains intact if the domain of admissible preference profiles is subject to certain structural restrictions. Many results in this area rely on the so-called local approach due to \citet{KMS79a}, who proposed a simple domain condition that is sufficient for Arrow's impossibility. \citet{LeBr86a} has shown that this condition is satisfied by the domain of vNM preferences, which implies Arrow's impossibility (see also \citet[][p.~214]{LeWe11a}). The corresponding IIA condition is defined for arbitrary pairs of lotteries, or---equivalently---arbitrary feasible sets of size two (which implies IIA for arbitrary feasible sets of lotteries).
In view of the structure of the set of lotteries, weaker IIA conditions (for example, restricted to convex feasible sets) seem natural.

\citet{Sen70a} has initiated the study of so-called \emph{social welfare functionals (SWFLs)}, which map a profile of cardinal utilities to a transitive and complete collective preference relation \citep[see also][]{AsGe02a}. The definitions of Pareto optimality and IIA can be straightforwardly extended to SWFLs. Note, however, that IIA takes into account the absolute values of utilities (rather than only ordinal comparisons between these values). This allows for Pareto optimal SWFLs that satisfy IIA, for example by adding individual utilities (utilitarianism).

vNM utilities are invariant under positive affine transformations. To account for this, \citet{Sen70a} introduced the axiom of \emph{cardinality and non-comparability}, which prescribes that collective preferences returned by the SWFL are invariant under positive affine transformations of the individual utility functions. However, this assumption effectively turns the problem into a problem of ordinal preference aggregation because the utility values assigned to two different alternatives in two different utility profiles can be made identical across profiles by applying a positive affine transformation. Hence, IIA implies an ordinal version of IIA which only takes into account the ordinal comparisons between utility values and Arrow's original theorem holds \citep[][Theorem 8*2]{Sen70a}. 

There are two ways to interpret this result. First, one can view the set of alternatives as the set consisting of only degenerate lotteries. This leads to weak notions of Pareto optimality and IIA because they are only concerned with degenerate lotteries. Non-dictatorship, on the other hand, becomes much stronger because a dictator can only enforce his (strict) preferences over degenerate lotteries, rather than all lotteries. Alternatively, one can define the set of alternatives as the set of all lotteries. This gives rise to stronger notions of Pareto optimality and IIA based on pairs of lotteries, rather than pairs of degenerate lotteries. In this model, non-dictatorship is defined by excluding agents who can enforce their (strict) preferences over lotteries. Like Arrow's theorem, \citeauthor{Sen70a}'s result assumes an unrestricted domain of preferences (or ordinal utilities, respectively). Expected utility functions over a set of lotteries, however, are subject to certain structural constraints (described by the vNM axioms independence and continuity). This gap is filled by \citet[][Proposition 3]{Mong94a}, who has shown that \citeauthor{Sen70a}'s result still holds when the set of alternatives is a convex subset of some vector space with mixture-preserving (i.e., affine) utility functions, which includes the domain of lotteries over some finite set of alternatives as a special case.\footnote{The vector space is required to be at least of dimension 2, which corresponds to the set of lotteries over at least three degenerate lotteries.} The Pareto condition used by \citeauthor{Mong94a} is identical to the one used in his paper and therefore slightly stronger than the one used by Arrow, Sen, and Le Breton. 

A very strong impossibility for vNM preferences was given by \citet{KaSc77b} (and later improved by \citet{Hyll80b}). \citeauthor{KaSc77b} consider ``cardinal'' preference relations represented by equivalence classes of utility functions that can be transformed into each other using positive affine transformations and \emph{cardinal social welfare functions}, which map a profile of cardinal preference relations to a collective cardinal preference relation. The set of alternatives is defined as the set of degenerate lotteries like in the first interpretation of Sen's result above. Preferences over lotteries are implicit in each equivalence class of utility functions.
When interpreted in our ordinal framework, they prove an Arrovian impossibility when individual and collective preferences over lotteries are subject to the vNM axioms and there are at least four alternatives. In contrast to the results by Le Breton, Sen, and Mongin, IIA is only required for feasible sets given by the convex combination of degenerate lotteries and non-dictatorship only rules out projections. The theorem thus uses weaker notions of Pareto optimality, IIA, and non-dictatorship at the expense of also requiring the vNM axioms for the collective preference relation.
When replacing non-dictatorship with anonymity, our \Cref{thm:pcdomain} implies a similar impossibility, even without requiring collective preferences to be transitive (see Remark~\ref{rem:trans}). We use \citeauthor{KaSc77b}'s weak IIA notion, but \citeauthor{Mong94a}'s strong notion of Pareto optimality.

\section{Characterization of the Domain}
\label{app:domain}

We start by showing that every continuous relation satisfies a weaker notion of continuity known as Archimedean continuity.
This statement will also be used in Appendices \ref{app:choice} and \ref{app:ssb}.
A preference relation $\succ$ satisfies \emph{Archimedean continuity} if for all $p,q,r\in\Delta$,
\[
	p\succ q\succ r \text{ implies } p \lambda r \sim q \text{ for some } \lambda\in(0,1) \text.\tag{Archimedean continuity}
\]
It is well-known that continuity implies Archimedean continuity \citep[see, e.g.,][]{Karn07a}.
We give a proof for completeness below.

\begin{lemma}
	If a preference relation $\succ$ satisfies continuity, then it satisfies Archimedean continuity.
	\label{lem:archcont}
\end{lemma}

\begin{proof}
	Let $\succ$ be a preference relation satisfying continuity, $p,q,r\in \Delta$ such that $p\succ q\succ r$, and
	$\lambda^\ast = \sup_{\lambda\in[0,1]} \{q\succ p\lambda r\}$.
	Since $q\succ r$, $\lambda^\ast$ is well defined.
	Continuity of $\succ$ implies that $L(q)$ and $U(q)$ are open and hence, $L(q)\cap [p,r]$ and $U(q)\cap [p,r]$ are open in $[p,r]$.
	If $q\succ p\lambda^\ast r$, then $L(q)\cap [p,r]$ is not open in $[p,r]$, since $p\lambda r\succsim q$ for all $\lambda > \lambda^\ast$, which is a contradiction.
	If $p\lambda^\ast r\succ q$, then $U(q)\cap [p,r]$ is not open in $[p,r]$, since, by definition of $\lambda^\ast$, every open neighborhood of $p\lambda^\ast r$ in $[p,r]$ contains $p\lambda r$ for some $\lambda < \lambda^\ast$ such that $q\succ p\lambda r$. 	
	This is again a contradiction.
	Hence, $q\sim p\lambda^\ast r$.
\end{proof}

Two vNM preference relations with the same symmetric part have to be equal up to orientation, since any two linear functions representing them have to have the same null space.
\citet[][Theorem 2]{FiGe87a} have shown that this statement extends to SSB preference relations, which will be useful in subsequent proofs.\footnote{\citet[][Lemma 8.8]{Bran18a} shows that \Cref{lem:indiff} even holds when $\succ$ and $\hat\succ$ are only required to satisfy Archimedean continuity and \citeauthor{Fish82c}'s~\citeyearpar{Fish82c} dominance axiom.}

\begin{lemma}[\citealp{FiGe87a}]\label{lem:indiff}
	Let ${\succ}, {\hat\succ}\in\mathcal R$ such that ${\sim}\subseteq{\hat\sim}$. Then, $\hat\succ\in\{{\succ},{\succ^{-1}}, \emptyset\}$.
\end{lemma}

The next lemma is reminiscent of what is known as the \emph{field expansion lemma} in traditional proofs of Arrow's theorem \citep[see, e.g.][]{Sen86a}.\footnote{In contrast to \Cref{lem:semidecisive}, the consequence of the original field expansion lemma uses a stronger notion of decisiveness.}
	Let $f\colon \mathcal{D}^N\rightarrow\mathcal{R}$ be an SWF, $G,H\subseteq N$, and $a,b\in \mathcal U$.
	We say that $(G, H)$ is decisive for $a$ against $b$, denoted by $a\mathrel{D_{G,H}} b$, if, for all $R\in\mathcal D^N$, $a\succ_i b$ for all $i\in G$, $a\sim_i b$ for all $i\in H$, and $b\succ_i a$ for all $i\in N\setminus (G\cup H)$ implies $a\succ b$.	
	Hence, $D_{G,H}$ is a relation on $\mathcal U$.

\begin{lemma}\label{lem:semidecisive}
	Let $f$ be an Arrovian SWF on some rich domain $\mathcal{D}$ with $|\mathcal U|\ge 3$, $G,H\subseteq N$, and $a,b\in \mathcal U$.
	Then, $a\mathrel{D_{G,H}} b$ implies that $D_{G,H} = \mathcal U\times \mathcal U$.
\end{lemma}
\begin{proof}
	First we show that $a\mathrel{D_{G,H}} x$ and $b\mathrel{D_{G,H}} x$ for all $x\in \mathcal U\setminus\{a,b\}$.
	To this end, let $x\in \mathcal U\setminus\{a,b\}$ and ${\succ_x}\in\mathcal D$ be a preference relation such that $a\succ_x b\succ_x x$ and $a\succ_x x$, which exists by richness assumption \ref{eq:richbottom} (cf. \Cref{sec:swf}).
	Consider the preference profile 
	\[
		R = (\underbrace{\succ_x,\dots,\succ_x}_G,\underbrace{\emptyset,\dots,\emptyset}_H,{\succ_x^{-1}},\dots,{\succ_x^{-1}})\text,
	\]
	which exists by \ref{eq:richindiff} and \ref{eq:richinverse}.
	Since ${\succ_x} \cap {{\succ_x^{-1}}} = \emptyset$, it follows from Pareto indifference and \Cref{lem:indiff} that ${\succ} = f(R)\in \{{\succ_x},{{\succ_x^{-1}}},\emptyset\}$.
	Since $a\mathrel{D_{G,H}} b$, ${\succ} = {\succ_x}$ remains as the only possibility.
	Hence, $a\succ x$ and $b\succ x$.
	By IIA, it follows that $a\mathrel{D_{G,H}} x$ and $b\mathrel{D_{G,H}} x$.
	
	Repeated application of the second statement implies that $D_{G,H}$ is a complete relation.
	To show that $D_{G,H}$ is symmetric, let $x,y,z\in \mathcal U$ such that $x\mathrel{D_{G,H}} y$.
	The first part of the statement implies that $x\mathrel{D_{G,H}} z$.
	Two applications of the second part of the statement yield $z\mathrel{D_{G,H}} y$ and $y\mathrel{D_{G,H}} x$.
	Hence, $D_{G,H} = \mathcal U\times \mathcal U$.
\end{proof}

Now we show that anonymous Arrovian aggregation is only possible on rich domains in which preferences over outcomes are completely determined by preferences over pure outcomes.

\begin{lemma}\label{lem:uniquerelation}
	Let $f$ be an anonymous Arrovian SWF on some rich domain $\mathcal D$ with $|\mathcal U|\ge 3$.
	Then, ${\succ}|_A = {\hat\succ}|_A$ implies ${\succ}|_{\Delta_A} = {\hat\succ}|_{\Delta_A}$ for all ${\succ}, {\hat\succ}\in\mathcal D$ and $A\subseteq \mathcal U$. 
\end{lemma}

\begin{proof}
	Let ${\succ_0}, {\hat\succ_0}\in\mathcal D$ and $A\subseteq \mathcal U$ such that ${\succ_0}|_{A} = {\hat\succ_0}|_{A}$.
	Consider the preference profile 
	\[
		R = ({\succ_0},{\hat\succ_0^{-1}},\emptyset,\dots,\emptyset)\text,
	\]
	which exists by \ref{eq:richindiff} and \ref{eq:richinverse}.
	Assume that there are $a,b\in A$ such that $a\succ_0 b$ and define $\bar R=R_{(12)}$ to be identical to $R$ except that the preferences of agents $1$ and $2$ are exchanged. 
	Anonymity of $f$ implies that ${\bar\succ} =f(\bar R) = f(R)= {\succ}$. 
	Assume for contradiction that $a\succ b$.
	Then, by IIA, $(\{1\},N\setminus\{1,2\})$ is decisive for $a$ against $b$.
	\Cref{lem:semidecisive} implies that $(\{1\},N\setminus\{1,2\})$ is also decisive for $b$ against $a$.
	Hence $b\mathrel{\bar\succ} a$, which contradicts $\bar\succ = {\succ}$.
	Hence, we get that $a\sim b$ for all $a,b\in A$ such that $a\mathrel{\hat\succ_0} b$.
	For $a,b\in A$ such that $a\sim_0 b$ and $a\mathrel{\hat\sim_0} b$, it follows from Pareto indifference that $a\sim b$.
	Hence, $a\sim b$ for all $a,b\in A$.
	
	Since by convexity of $\succ$, indifference sets are convex, we get that ${\succ}|_{\Delta_A} = \emptyset$.
	If ${\succ}_0|_{\Delta_A} \neq {\hat\succ}_0|_{\Delta_A}$, there are $p,q\in\Delta_A$ such that $p\succ_0 q$ and not $p\mathrel{\hat\succ_0} q$, i.e., $p\mathrel{\hat\succsim_0^{-1}} q$.
	The strict part of Pareto optimality of $f$ implies that $p\succ q$.
	This contradicts ${\succ}|_{\Delta_A} = \emptyset$.
	Hence, ${\succ}_0|_{\Delta_A} = {\hat\succ}_0|_{\Delta_A}$.
\end{proof}

\Cref{lem:uniquerelation} is the only part of the proof of \Cref{thm:pcdomain} that requires anonymity. A much weaker condition would also suffice: there has to be $R\in \mathcal{D}^N$, $x,y\in \mathcal U$, $i\in N$, and $f(R)={\succ}$ such that $x \succ_i y$ and $x \sim y$.

Next, we show that intensities of preferences between pure outcomes have to be identical.

\begin{lemma}\label{lem:pcdomain}
	Let $f$ be an anonymous Arrovian SWF on some rich domain $\mathcal D$ with $|\mathcal U|\ge 4$.
	Then, for all ${\succ_0}\in\mathcal D$ and $a,b,c\in \mathcal U$ with $a\succ_0 b$, 
	\begin{enumerate}[label=(\roman*)]
		\item $b \succ_0 c$ implies $\phi_0(a,b) = \phi_0(b,c)$,\label{item:domain1}
		\item $a \succ_0 c$ implies $\phi_0(a,b) = \phi_0(a,c)$, \label{item:domain2}
		\item $c \succ_0 b$ implies $\phi_0(a,b) = \phi_0(c,b)$, and \label{item:domain3}
		\item $c \succ_0 a$ implies $\phi_0(a,b) = \phi_0(c,a)$. \label{item:domain4}
	\end{enumerate}
\end{lemma}

\begin{proof}
	Ad \ref{item:domain1}: Since, by \Cref{lem:archcont}, $\succ_0$ satisfies Archimedean continuity, it follows that $b\sim_0 a \lambda c$ for some $\lambda\in(0,1)$.
	Observe that ${\succ}_0^{(ac)}|_{\{a,b,c\}} = {\succ}_0^{-1}|_{\{a,b,c\}}$, where $(ac)$ denotes the permutation that swaps $a$ and $c$ and leaves all other alternatives fixed.
	\Cref{lem:uniquerelation} implies that ${\succ}_0^{(ac)}|_{\Delta_{\{a,b,c\}}} = {\succ}_0^{-1}|_{\Delta_{\{a,b,c\}}}$.
	Hence, we have $b\sim_0 c \lambda a$.
	Convexity of $\succ_0$ then implies that $I(b)$ is convex and hence, $b\sim_0 \nicefrac{1}{2}\, a + \nicefrac{1}{2}\, c$.
	This is equivalent to $\phi_0(a,b) = \phi_0(b,c)$.
	
	Ad \ref{item:domain2}: we distinguish two cases.
	
	\em{Case 1} ($b\sim_0 c$)\em : Consider the preference profile 
	
	\[R = (\succ_0,(\succ_0^{(bc)})^{-1},\emptyset,\dots,\emptyset ) \text{,}\]
	which exists by \ref{eq:richneutral}, \ref{eq:richindiff}, and \ref{eq:richinverse}.
	
	Let ${\succ} = f(R)$. 
	As in the proof of \Cref{lem:uniquerelation}, we get that ${\succ}|_{\Delta_{\{a,b,c\}}} = \emptyset$.
	Without loss of generality, assume that $\phi_0(a,b) = 1$ and $\phi_0(a,c) = \lambda$ for some $\lambda\in(0,1]$.
	Let $p = \nicefrac{1}{2}\, a + \nicefrac{1}{2}\, c$ and $q = \nicefrac{1}{2}\, a + \nicefrac{1}{2}\, b$ and denote by $\phi_1$ and $\phi_2$ the SSB functions representing the preference relations $\succ_0$ and $(\succ_0^{(bc)})^{-1}$, respectively.
	Then, $\phi_1(p,q) = \phi_2(p,q) = \nicefrac{1}{4}\,(1-\lambda)$.
	If $\lambda < 1$, the strict part of Pareto optimality of $f$ implies that $p\succ q$.
	This contradicts ${\succ}|_{\{a,b,c\}} = \emptyset$.
	Hence, $\lambda = 1$.

	\em{Case 2} ($b\succ_0 c$)\em : Assume without loss of generality that $\phi_0(a,b) = 1$. 
	By \ref{item:domain1}, we get $\phi_0(a,b) = \phi_0(b,c) = 1$.
	By \ref{eq:richbottom}, there is $\hat\succ_0\in\mathcal D$ with $a\mathrel{\hat\succ_0} b \mathrel{\hat\succ_0} c,\, a \mathrel{\hat\succ_0} c$, and $c\mathrel{\hat\succ_0} x$ for some $x\in \mathcal U$.
	\Cref{lem:uniquerelation} implies that $\phi_0|_{\{a,b,c\}} = \hat\phi_0|_{\{a,b,c\}}$.
	Hence, it suffices to show that $\hat\phi_0(a,c) = 1$.
	By \ref{item:domain1}, we get that $\hat\phi_0(a,c) = \hat\phi_0(c,x)$ and $\hat\phi_0(b,c) = \hat\phi_0(c,x) = 1$.
	Hence, $\hat\phi_0(a,c) = 1$.
		
	Ad \ref{item:domain3}: The proof is analogous to the proof of \ref{item:domain2}.
	
	Ad \ref{item:domain4}: The proof is analogous to the proof of \ref{item:domain1}.
\end{proof}

\thmpcdomain*

\begin{proof}
	Let ${\succ_0}\in\mathcal D$ and $a,b,c,d\in \mathcal U$ such that $a\succ_0 b$ and $c\succ_0 d$.
	We have to show that $\phi_0(a,b) = \phi_0(c,d)$.
	First assume there are $x\in\{a,b\}$ and $y\in\{c,d\}$ such that $x\succ_0 y$ or $y\succ_0 x$. 
	Then, \Cref{lem:pcdomain} implies that $\phi_0(a,b) = \phi_0(x,y) = \phi_0(c,d)$ or $\phi_0(a,b) = \phi_0(y,x) = \phi_0(c,d)$, respectively.
	Otherwise, $x\sim_0 y$ for all $x\in\{a,b\}$ and $y\in\{c,d\}$.
	This implies that ${\succ_0}|_{\{a,b,c,d\}} = {\succ_0}^{(ac)(bd)}|_{\{a,b,c,d\}}$.
	From \Cref{lem:uniquerelation} it follows that ${\succ_0}|_{\Delta_{\{a,b,c,d\}}} = {\succ_0}^{(ac)(bd)}|_{\Delta_{\{a,b,c,d\}}}$.
	Hence, $\phi_0|_{\{a,b,c,d\}} = \phi_0^{(ac)(bd)}|_{\{a,b,c,d\}}$ and $\phi_0(a,b) = \phi_0(c,d)$.
\end{proof}

\section{Characterization of the Social Welfare Function}
\label{app:swf}

Except for \Cref{thm:swf}, all results in this section only require Pareto indifference rather than Pareto optimality.

The following lemmas show that for all preference profiles $R$ and all alternatives $a$ and $b$, $\phi(a,b)$ only depends on the set of agents who prefer $a$ to $b$, whenever $R$ is from the domain of \pc preferences and $\phi$ represents $f(R)$. 
We first prove that, if an alternative is strictly Pareto dominated, then the intensities of collective preferences between each of the dominating alternatives and the dominated alternative are identical.

\begin{lemma}\label{lem:pareto}
	Let $f$ be an SWF satisfying Pareto indifference and IIA on some rich domain $\mathcal D\subseteq \mathcal{R}^\pc$ with $|\mathcal U|\ge 4$.
	Let $a,b,c\in \mathcal U$ and $R\in\mathcal D^N$ such that $N_{ac}=N_{bc}=N$.

	 Then, $\phi(a,c) = \phi(b,c)$ where $\phi\equiv f(R)$.
\end{lemma}
	
\begin{proof}
	The idea of the proof is to introduce a fourth alternative, which serves as a calibration device for the intensity of pairwise comparisons, and eventually disregard this alternative using IIA.
	To this end, let $x\in \mathcal U$ and consider a preference profile $\hat R\in\mathcal{D}^N$ such that $R|_{\{a,b,c\}} = \hat R|_{\{a,b,c\}}$ and $\hat N_{ax}=\hat N_{bx}=\hat N_{cx}=N$
	which exists by \ref{eq:richbottom}.
	Let $\hat\phi \equiv f(\hat R)$.
	The Pareto indifference relation with respect to $\hat R|_{\{a,c,x\}}$ is identical to ${\sim_1}|_{\{a,c,x\}}$. 
	The analogous statement holds for the Pareto indifference relation with respect to $\hat R|_{\{b,c,x\}}$. 
	Hence, Pareto indifference, \Cref{lem:indiff}, and IIA imply that there are $\alpha,\beta \in \mathbb{R}$ such that
\[
\hat\phi|_{\{a,c,x\}} = \alpha
	   \begin{pmatrix}
	 	0 & 1 & 1 \\
		-1 & 0 & 1 \\
		-1 & -1 & 0 \\
	\end{pmatrix}
\qquad\text{and}\qquad
\hat\phi|_{\{b,c,x\}} = \beta
	   \begin{pmatrix}
	 	0 & 1 & 1 \\
		-1 & 0 & 1 \\
		-1 & -1 & 0 \\
	\end{pmatrix}
	\text.
\]
	As a consequence, $\alpha=\beta$ and $\hat\phi(a,c) = \hat\phi(b,c)$.
	Since $R|_{\{a,b,c\}} = \hat R|_{\{a,b,c\}}$, \Cref{lem:uniquerelation} and IIA imply that $\phi|_{\{a,b,c\}} \equiv \hat\phi|_{\{a,b,c\}}$.
	Hence, $\phi(a,c) = \phi(b,c)$.\footnote{Pareto optimality also implies that $\phi(a,c),\phi(b,c)>0$.}
\end{proof}

\Cref{lem:mlscheme1} shows that for a fixed preference profile, $\phi(a,b)$ only depends on $N_{ab}$ and $I_{ab}$ (and not on the names of the alternatives).

\begin{lemma}\label{lem:mlscheme1}
	Let $f$ be an SWF satisfying Pareto indifference and IIA on some rich domain $\mathcal D\subseteq \mathcal{R}^\pc$ with $|\mathcal U|\ge 5$, $a,b,c,d\in \mathcal U$, and $R\in\mathcal{D}^N$ such that $N_{ab} = N_{cd}$ and $N_{ba} = N_{dc}$.
	Then, $\phi(a,b) = \phi(c,d)$ where $\phi\equiv f(R)$.
\end{lemma}

\begin{proof}
	We first prove the case when all of $a,b,c,d$ are distinct.
	Let $e\in \mathcal U$ and consider a preference profile $\hat R\in\mathcal{D}^N$ such that $R|_{\{a,b,c,d\}} = \hat R|_{\{a,b,c,d\}}$ and $\hat N_{ae}=\hat N_{be}=\hat N_{ce}=\hat N_{de}= N$.
	Such a profile exists by~\ref{eq:richbottom}.
	Then, by \Cref{lem:pareto}, we can assume without loss of generality that $\hat\phi(a,e)=\hat\phi(b,e)=\hat\phi(c,e)=\hat\phi(d,e)= \lambda\in\mathbb{R}$.
	Now consider a preference profile $\mathring R\in\mathcal D^N$ such that
	\[
	\mathring R|_{\{a,b,c,d,e\}}=(
	\underbrace{
	\begin{bmatrix}
		a\\b\\c\\d\\e
	\end{bmatrix},\dots
	}_{N_{ab}}\,,
	\underbrace{
	\begin{bmatrix}
		d\\c\\b\\a\\e
	\end{bmatrix},\dots
	}_{N_{ba}}\,,
	\begin{bmatrix}
		a,b,c,d\\e
	\end{bmatrix},\dots)\text,
	\]
	which exists by \ref{eq:richbottom}.
	Note that $\hat R|_{\{a,b,e\}}= \mathring R|_{\{a,b,e\}}$ and $\hat R|_{\{c,d,e\}} = \mathring R|_{\{c,d,e\}}$ because $N_{ab} = N_{cd}$ and $N_{ba} = N_{dc}$ by assumption.
	Now, let $\hat\phi\equiv f(\hat R)$ and $\mathring \phi\equiv f(\mathring R)$.
	Since $\hat R|_{\{a,b,e\}} = \mathring R|_{\{a,b,e\}}$, we have $\hat\phi|_{\{a,b,e\}} \equiv \mathring\phi|_{\{a,b,e\}}$ by IIA.
	Moreover, $\hat R|_{\{c,d,e\}} = \mathring R|_{\{c,d,e\}}$ and IIA yield $\hat\phi|_{\{c,d,e\}} \equiv \mathring\phi|_{\{c,d,e\}}$.
	\Cref{lem:pareto} implies that $\mathring\phi(a,e) = \mathring\phi(b,e) =\mathring\phi(c,e) =\mathring\phi(d,e) =\lambda$ for some $\lambda\in\mathbb{R}$.
	Thus, for some $\mu, \sigma\in\mathbb{R}$, $\mathring\phi$ takes the form
	\[
	\mathring\phi|_{\{a,b,c,d,e\}} =
	\begin{pmatrix}
		0 & \mu &  &  & \lambda\\
		-\mu & 0 & & &  \lambda\\
		& & 0 & \sigma & \lambda\\
		& & -\sigma & 0 & \lambda\\ 
		-\lambda & -\lambda & -\lambda & -\lambda & 0
	\end{pmatrix}\text{.}
	\]
	Note that $\mathring R|_{\{a,b,c,d\}}$ only consists of one fixed preference relation, its inverse, and complete indifference.
	Hence, Pareto indifference and \Cref{lem:indiff} imply that $\mathring\phi|_{\{a,b,c,d\}} = \alpha \mathring\phi_i|_{\{a,b,c,d\}}$ for some $\alpha\in\mathbb R$, $i\in N$, and $\mathring\phi_i\equiv \mathring R_i$.
	As a consequence, we get that $\mu = \sigma$.
	Since $\phi|_{\{a,b,c,d\}} = \hat\phi|_{\{a,b,c,d\}}$, it follows that $\phi(a,b) = \phi(c,d)$.
	
	The cases when $a = c$ and $b = c$ follow from repeated application of the above case.
	All other cases are symmetric to one of the covered cases.
\end{proof}

\begin{lemma}\label{lem:mlscheme2}
	Let $f$ be an SWF satisfying Pareto indifference and IIA on some rich domain $\mathcal D\subseteq \mathcal{R}^\pc$ with $|\mathcal U|\ge 5$, $a,b,c,d\in \mathcal U$, $R,\hat R\in\mathcal{D}^N$, $\phi\equiv f(R)$, and $\hat\phi\equiv f(\hat R)$.
	If $R|_{\{a,b\}} = \hat R|_{\{a,b\}}$ and $R|_{\{c,d\}} = \hat R|_{\{c,d\}}$, then $\phi(a,b) = \alpha \cdot \hat \phi(a,b)$ and $\phi(c,d) = \alpha \cdot \hat \phi(c,d)$ for some $\alpha>0$.
\end{lemma}

\begin{proof}
	Let $e\in \mathcal U\setminus \{a,b,c,d\}$ and $R',\hat R'\in\mathcal{D}^N$ such that $R'|_{\{a,b,c,d\}} = R|_{\{a,b,c,d\}}$, $\hat R'|_{\{a,b,c,d\}} = \hat R|_{\{a,b,c,d\}}$, and 
	$N'_{ae}=N'_{be}=N'_{ce}=N'_{de}=\hat N'_{ae}=\hat N'_{be}=\hat N'_{ce}=\hat N'_{de}=N$.
	The profiles $R',\hat R'$ exist by \ref{eq:richbottom}.
	By $\phi' \equiv f(R')$ and $\hat\phi' \equiv f(\hat R')$ we denote the corresponding collective SSB functions.
	Since $f$ satisfies IIA, we have that $\phi|_{\{a,b,c,d\}} \equiv \phi'|_{\{a,b,c,d\}}$ and $\hat\phi|_{\{a,b,c,d\}} \equiv \hat\phi'|_{\{a,b,c,d\}}$.
	\Cref{lem:pareto} implies that without loss of generality, $\phi'$ and $\hat\phi'$ take the following form for some $\lambda,\mu,\hat\mu,\sigma,\hat\sigma\in\mathbb{R}$. 
	Note that we can choose suitable representatives such that $\phi'(a,e) = \hat\phi'(a,e) = \lambda$.
	\[
	\phi'|_{\{a,b,c,d,e\}} =
	\begin{pmatrix}
		0 & \mu & &  & \lambda\\
		-\mu & 0 &  &  & \lambda\\
		 & & 0 & \sigma & \lambda\\
		 & & -\sigma & 0 & \lambda\\
		-\lambda & -\lambda & -\lambda & -\lambda & 0
	\end{pmatrix}
	\qquad
	\hat\phi'|_{\{a,b,c,d,e\}} =
	\begin{pmatrix}
		0 & \hat\mu &  &  & \lambda\\
		-\hat\mu & 0 &  &  & \lambda\\
		 & & 0 & \hat\sigma & \lambda\\
		 &  & -\hat\sigma & 0 & \lambda\\
		-\lambda & -\lambda & -\lambda & -\lambda & 0
	\end{pmatrix}
	\]
	Observe that $R'|_{\{a,b,e\}} = \hat R'|_{\{a,b,e\}}$ and $R'|_{\{c,d,e\}} = \hat R'|_{\{c,d,e\}}$ by construction.
	Since $f$ satisfies IIA, we get that $\phi'|_{\{a,b,e\}} = \hat\phi'|_{\{a,b,e\}}$ and $\phi'|_{\{c,d,e\}} = \hat\phi'|_{\{c,d,e\}}$.
	In particular, this means that $\mu = \hat\mu$ and $\sigma = \hat\sigma$.
	Since $\phi|_{\{a,b,c,d\}} \equiv \phi'|_{\{a,b,c,d\}}$ and $\hat\phi|_{\{a,b,c,d\}} \equiv \hat\phi'|_{\{a,b,c,d\}}$, there is $\alpha>0$ as required.
\end{proof}

\Cref{lem:mlscheme2} shows that $\phi(a,b)$ only depends on $N_{ab}$ and $I_{ab}$ and not on $a,b$ or $R$. 
Hence, there is a function $g\colon 2^N\times 2^N \rightarrow\mathbb{R}$ such that $g(N_{ab},I_{ab}) = \phi(a,b)$ for all $a,b\in \mathcal U$ and $R\in\mathcal D^N$ with $\phi \equiv f(R)$.
We now leverage Pareto indifference to show that $f$ is affine utilitarian.

\begin{lemma}\label{lem:mlscheme3}
	Let $f$ be an SWF satisfying Pareto indifference and IIA on some rich domain $\mathcal D\subseteq \mathcal{R}^\pc$ with $|\mathcal U|\ge 5$. 
	Then, $f$ is affine utilitarian.
\end{lemma}

\begin{proof}
	Let $R\in \mathcal{D}^N$ and $(\phi_i)_{i\in N}\in \Phi^N$ such that $(\phi_i)_{i\in N}\equiv R$.
	For all $G\subseteq N$, let $w_G = \nicefrac{1}{2}\, (g(N,\emptyset) + g(G,\emptyset))$.
	For convenience, we write $w_i$ for $w_{\{i\}}$.
	Since $\phi(x,y) = g(N_{xy},I_{xy})$ for all $x,y\in \mathcal U$, it suffices to show that
	\begin{equation}
		g(N_{xy},I_{xy}) = \sum_{i\in N} w_i \phi_i(x,y) = \sum_{i\in N_{xy}} w_i - \sum_{i\in N_{yx}} w_i\text,\label{eq:welfare1}
	\end{equation}
	for all $x,y\in \mathcal U$.
	To this end, we will first show that $w_G + w_{\hat G} = w_{G\cup\hat G}$ for all $G,\hat G\subseteq N$ with $G\cap \hat G= \emptyset$.
	Let $G,\hat G$ as above, $a,b,c,x,y\in \mathcal U$, and consider a preference profile $R\in\mathcal{D}^N$ such that
	\[
		R|_{\{a,b,c,x,y\}}=
		(
		\underbrace{
		\begin{bmatrix}
			x\\
			a\\
			b\\
			c\\
			y
		\end{bmatrix}
		,\dots
		}_{G}\,
		,
		\underbrace{
		\begin{bmatrix}
			b\\
			y\\
			c\\
			x\\
			a
		\end{bmatrix}
		,\dots
		}_{\hat G}\,
		,
		\begin{bmatrix}
			c\\
			x\\
			a\\
			y\\
			b
		\end{bmatrix}
		,\dots
		)\text,
	\]	
	which exists by \ref{eq:richbottom}.
	Let $\phi \equiv f(R)$.
	We have that, for $p = \nicefrac{1}{2}\, x + \nicefrac{1}{2}\, y$ and $q = \nicefrac{1}{3}\, a + \nicefrac{1}{3}\, b +\nicefrac{1}{3}\, c$, $\phi_i(p,q) = 0$ for all $i\in N$.
	Pareto indifference implies that $\phi(p,q) = 0$.
	Let $\mu = g(G,\emptyset)$, $\hat\mu = g(\hat G,\emptyset)$, and $\sigma = g(G\cup \hat G,\emptyset)$.
	By definition of $w$,
	\begin{align*}
		w_G + w_{\hat G} &= w_{G\cup\hat G}
	\intertext{ is equivalent to }
	(g(N,\emptyset) + g(G,\emptyset)) + (g(N,\emptyset) + g(\hat G,\emptyset)) &= g(N,\emptyset) + g(G\cup \hat G,\emptyset)\text.
	\end{align*}
	Hence, we have to show that $\mu + \hat\mu + g(N,\emptyset) = \sigma$.
	By definition of $g$, we get that $\phi$ takes the following form. 
	\[
	\phi|_{\{a,b,c,x,y\}} =
	\begin{pmatrix}
		0 &  &  &  -g(N,\emptyset) & -\hat\mu \\
		  & 0 &  & \hat\mu & \sigma\\
		  & & 0 & -\mu & -\hat\mu\\
		g(N,\emptyset) & -\hat\mu & \mu & 0 & \\
		\hat\mu  & -\sigma &  \hat\mu &  & 0
	\end{pmatrix}
	\]
	From $\phi(p,q) = 0$, it follows that $\nicefrac{1}{6}\,(\mu + \hat\mu + g(N,\emptyset) - \sigma) = 0$.
	This proves the desired relationship.

	Now we can rewrite~\eqref{eq:welfare1} as $g(N_{xy},I_{xy}) = w(N_{xy}) - w(N_{yx})$, which, by definition of $w$, is equivalent to
	\begin{equation}
		2g(N_{xy},I_{xy}) = g(N_{xy},\emptyset) - g(N_{yx},\emptyset)\text.\label{eq:welfare3}
	\end{equation}
	To prove~\eqref{eq:welfare3}, let $a,b,x,y\in \mathcal U$ and consider a preference profile $\hat R\in\mathcal{D}^N$ such that
	\[
		\hat R|_{\{a,b,x,y\}}=(
		\underbrace{
		\begin{bmatrix}
			a\\
			x,y\\
			b
		\end{bmatrix}
		,\dots
		}_{G}\,
		,
		\underbrace{
		\begin{bmatrix}
			x\\
			a\\
			b,y\\
		\end{bmatrix}
		,\dots
		}_{\hat G}\,,
		\begin{bmatrix}
			b\\
			x,y\\
			a\\
		\end{bmatrix}
		,\dots
		)\text,
	\]
	which exists by \ref{eq:richindiff}, \ref{eq:richinverse}, and \ref{eq:richbottom}.
	Let $\hat \phi \equiv f(\hat R)$.
	Observe that, for $p = \nicefrac{1}{3} \, x + \nicefrac{2}{3}\, y$ and $q = \nicefrac{1}{2}\, a + \nicefrac{1}{2}\, b$, $p \hat\sim_i q$ for all $i\in N$.
	Pareto indifference implies that $\hat\phi(p,q) = 0$.
	With the same definitions as before and $\epsilon = g(G,\hat G)$, $\hat\phi$ takes the following form.
	\[
	\hat\phi|_{\{a,b,x,y\}} \equiv
	\begin{pmatrix}
		0 & & \mu & \sigma \\
		 & 0 & -\sigma & -\epsilon\\
		-\mu & \sigma & 0 \\
		-\sigma & \epsilon & & 0
	\end{pmatrix}
	\]
	From $\hat\phi(p,q) = 0$, we get that $\nicefrac{1}{6}\,(-\mu + \sigma - 2\sigma + 2\epsilon) = 0$.
	Hence, $2\epsilon = \mu + \sigma$.
	This is equivalent to
	\[
		2g(G,\hat G) = g(G,\emptyset) + g(G\cup\hat G,\emptyset) = g(G,\emptyset) - g(N\setminus(G\cup\hat G),\emptyset)\text,
	\]
	where the last equality follows from skew-symmetry of $\hat\phi$ and the definition of $g$.
	This proves~\eqref{eq:welfare3}.
\end{proof}

Finally, the strict part of Pareto optimality implies that individual weights have to be positive.

\thmswf*

\begin{proof}
	From \Cref{lem:mlscheme3} we know that there are $w_1, \dots, w_n\in\mathbb{R}$ such that, for all $R\in\mathcal D^N$ and $\phi_1,\dots,\phi_n\in \Phi^\pc$ with $(\phi_i)_{i\in N}\equiv R$, $f(R) \equiv \sum_{i\in N} w_i \phi_i$.
	Assume for contradiction that $w_i \le 0$ for some $i\in N$.
	Let $G$ be the set of agents such that $w_i \le 0$ and consider a preference profile $R\in\mathcal{D}^N$ with $a,b\in \mathcal U$ such that
	\[
		R|_{\{a,b\}}=(
		\underbrace{
		\begin{bmatrix}
			a\\
			b
		\end{bmatrix}
		,\dots
		}_{G}\,
		,
		\begin{bmatrix}
			a,b
		\end{bmatrix}
		,\dots
		)\text,
	\]
	which exists by~\ref{eq:richindiff} and~\ref{eq:richbottom}.
	Let $\phi \equiv f(R)$.
	Then, we have that $\phi_i(a,b) = 1$ for all $i\in G$ and $\phi_i(a,b) = 0$ for all $i\in N\setminus G$. 
	Pareto optimality of $f$ implies that $\phi(a,b) > 0$.
	However, we have
	\[
		\phi(a,b) = \alpha \left(\sum_{i\in G} w_i \underbrace{\phi_i(a,b)}_{=1} + \sum_{i\in N\setminus G} w_i \underbrace{\phi_i(a,b)}_{= 0} \right) = \alpha \sum_{i\in G}w_i \le 0
	\]
	for some $\alpha>0$.
	This is a contradiction.
\end{proof}

\section{Characterizations for vNM preferences}
\label{app:dich}

In this section, we consider Arrovian SWFs with domain and range $\mathcal R^\vnm$ instead of $\mathcal R$.
For this case, we characterize both the largest domain of individual preferences that allows for anonymous Arrovian aggregation and the class of Arrovian SWFs on this domain, and thereby provide analogous statements to Theorems~\ref{thm:pcdomain} and~\ref{thm:swf} when restricting $\mathcal R$ to $\mathcal R^\vnm$.

A subdomain of vNM preferences is the domain of dichotomous preferences $\mathcal R^{\dich}=\mathcal R^\vnm \cap \mathcal R^\pc$, where every agent can assign one of only two different vNM utility values, say $0$ and $1$, to every pure outcome.
\Cref{thm:dichdomain} shows that this domain is the largest domain satisfying~\ref{eq:richneutral}, \ref{eq:richindiff}, and~\ref{eq:richinverse} on which anonymous Arrovian SWFs exist.

\begin{theorem}\label{thm:dichdomain}
	Let $f$ be an anonymous Arrovian SWF with range $\mathcal R^\vnm$ on some domain $\mathcal D\subseteq \mathcal R^\vnm$ satisfying~\ref{eq:richneutral}, \ref{eq:richindiff}, and~\ref{eq:richinverse}. Then, $\mathcal D\subseteq\mathcal R^\dich$.
\end{theorem}

\begin{proof}
	If $|\mathcal U|\le 2$, then $\mathcal R^\vnm = \mathcal R^\dich$, which immediately implies the statement of the theorem.
	So consider the case that $|\mathcal U|\ge 3$ and assume for contradiction that $\mathcal D\not\subseteq\mathcal R^\dich$, i.e., there is ${\pref_0}\in\mathcal D$ such that $a\pref_0 b\pref_0 c$ for some $a,b,c\in \mathcal U$.
	Observe that \Cref{lem:semidecisive} holds for $\mathcal D$, since $\mathcal D$ satisfies~\ref{eq:richneutral},~\ref{eq:richindiff}, and~\ref{eq:richinverse} and contains a preference relation with three indifference classes on pure outcomes, e.g., $\pref_0$.
	Hence, for all $G,H\subseteq N$ and $x,y\in \mathcal U$, if $(G,H)$ is decisive for $x$ against $y$, then $(G,H)$ is decisive for all pairs of alternatives, i.e., $x\mathrel{D_{G,H}} y$ implies $D_{G,H} = \mathcal U\times\mathcal U$.
	
	Now let $x,y\in \mathcal U$ and ${\pref_0^x},{\pref_0^y}\in\mathcal D$ such that $x\pref_0^x y$ and $y\pref_0^y x$, which exist by ${\pref_0}\in\mathcal D$ and~\ref{eq:richneutral}, and consider the preference profiles
	\[
		R = ({\pref_0^x},{\pref_0^y},\emptyset,\dots,\emptyset) \text{ and } \bar R = ({\pref_0^y},{\pref_0^x},\emptyset,\dots,\emptyset) \text,
	\]
	which exist by~\ref{eq:richindiff}.
	Let ${\pref} = f(R)$, ${\bar\pref} = f(\bar R)$ and observe that, by anonymity of $f$, ${\pref} = \bar\pref$.
	If $x\pref y$ or $y\pref x$, then \Cref{lem:semidecisive} implies that $y\mathrel{\bar\pref} x$ or $x\mathrel{\bar\pref} y$, respectively, which contradicts ${\pref} = \bar\pref$.
	Hence, $x\sim y$.
	
	Lastly, let ${\pref_1},{\pref_2}\in\mathcal D$ such that $a\pref_1 b\pref_1 c$ and $c\pref_2 a\pref_2 b$, which exist by ${\pref_0}\in\mathcal D$ and~\ref{eq:richneutral}, and consider the preference profile 
	\[
		\hat R = ({\pref_1},{\pref_2},\emptyset,\dots,\emptyset)\text,
	\]
	which exists by~\ref{eq:richindiff}.
	Let ${\hat\pref} = f(\hat R)$.
	Since $f$ satisfies IIA, it follows from what we have shown above that $a\mathrel{\hat\sim} c$ and $b\mathrel{\hat\sim} c$.
	Together with the fact that $\hat\pref\in \mathcal R^\vnm$, this implies $a\mathrel{\hat \sim} b$.
	However, since $f$ satisfies Pareto optimality, we have $a\mathrel{\hat\pref} b$, which is a contradiction.
\end{proof}

Secondly, we characterize the class of Arrovian SWFs on domains of dichotomous preferences.
To this end, we need to make a richness assumption for domains $\mathcal D\subseteq \mathcal R^\dich$, which prescribes that all dichotomous preferences relations on any set of up to four alternatives are possible.
\[
	\text{For all } {\hat\pref}\in\mathcal R^\dich\text{ and }X\subseteq\mathcal U\text{, }|X|\le 4 \text{, there is } {\pref}\in\mathcal D \text{ such that } {\pref}|_X = {\hat\pref}|_X\text.
	\tag{R5}
	\label{eq:richdich}
\]

We show that every Arrovian SWF on subdomains of $\mathcal R^\dich$ satisfying~\ref{eq:richdich} is affine utilitarian with positive weights.
Similar to \Cref{lem:mlscheme3}, we first prove that affine utilitarian SWFs are the only SWFs satisfying Pareto indifference and IIA.

\begin{lemma}\label{lem:dichmlvnm}
	Let $f$ be an SWF satisfying Pareto indifference and IIA with range $\mathcal R^{\vnm}$ on some domain $\mathcal D\subseteq\mathcal R^{\dich}$ satisfying~\ref{eq:richdich} with $|\mathcal U|\ge 4$.
	Then, $f$ is affine utilitarian.
\end{lemma}

\begin{proof}
	The proof is structured as follows.
	We start by defining a function $w$ that assigns a weight to every \emph{set} of agents based on the output of $f$ for specific profiles.
	It will turn out that $f$ is affine utilitarian for the weights that $w$ assigns to singleton sets.
	To prove this, we show that $w$ is additive, i.e., the weight of the union of two disjoint sets of agents is equal to the sum of the weights of both sets.

	First observe that, for all $R\in\mathcal D^N$, $f(R)\equiv\phi$, and $x,y,z\in \mathcal U$, we have
	\begin{align}
		\phi(x,z) = \phi(x,y) + \phi(y,z)\text,\label{eq:trans}
	\end{align}
	since $\phi\in\mathcal R^{\vnm}$ by assumption.
	If $f(R)\equiv 0$ for all $R\in\mathcal D^N$, we can choose $w_i = 0$ for all $i\in N$.
	Otherwise, there is $\hat R\in \mathcal D^N$ such that $f(\hat R)\equiv\hat\phi\neq 0$.
	Let $a,b,c,d\in \mathcal U$ be four distinct alternatives.
	We may assume without loss of generality that $\hat\phi(a,b)\neq 0$.
	Let $\bar R\in\mathcal D^N$, $f(\bar R)\equiv \bar\phi$, such that the preferences between $a$ and $b$ are as in $\hat R$ and both $a$ and $b$ are weakly preferred to $c$ by all agents. Formally, 
	$\bar R|_{\{a,b\}} = \hat R|_{\{a,b\}}$, $\bar N_{ac} = \hat N_{ab}$, $\bar N_{bc} = \hat N_{ba}$, and $\bar N_{ca} = \bar N_{cb} = \emptyset$, i.e.,
		\[
			\bar R|_{\{a,b,c\}}=(
			\underbrace{
			\begin{bmatrix}
				a\\
				b,c
			\end{bmatrix}
			,\dots
			}_{\hat N_{ab}}\,
			,
			\underbrace{
			\begin{bmatrix}
				b\\
				a,c
			\end{bmatrix}
			,\dots
			}_{\hat N_{ba}}\,
			,
			\begin{bmatrix}
				a,b,c
			\end{bmatrix}
			,\dots
			)\text.
		\]
		The profile $\bar R$ exists by~\ref{eq:richdich}.
	Since $f$ satisfies IIA, it follows that $\bar\phi(a,b)\neq 0$.
	By~\eqref{eq:trans}, we have that $\bar\phi(a,b) = \bar\phi(a,c) + \bar\phi(c,b)\neq 0$.
	Hence, by skew-symmetry of $\bar\phi$, either $\bar\phi(a,c)\neq 0$ or $\bar\phi(b,c)\neq 0$.
	Without loss of generality, we may assume that $\bar\phi(a,c)\neq 0$.
	Let $G^\ast = \bar N_{ac}\neq\emptyset$ and note that $\bar I_{ac} = N\setminus G^\ast$, since $\bar N_{ca} = \emptyset$.
	The set of agents $G^\ast$ will remain fixed for the rest of this proof. 
	Since $\bar\phi(a,c)\neq0$ and $f$ satisfies IIA, $G^\ast$ can be used to calibrate the utility values across different profiles.
	Based on $G^\ast$, we will now construct a function $w$ that assigns a weight to every set of agents.
	For every $G\subseteq N$, let $R^G\in\mathcal D^N$, $f(R^G)\equiv\phi^G$ such that $R^G|_{\{a,c\}} = \bar R|_{\{a,c\}}$, $N^G_{bc} = G$, and $I^G_{bc} = N\setminus G$.
	Hence,
	\[
		R^G|_{\{a,b,c\}}=(
		\lefteqn{
		\underbrace{
		\phantom{
			\begin{bmatrix}
				a\\
				b,c
			\end{bmatrix}
		,\dots\,
		,
			\begin{bmatrix}
				a,b\\
				c
			\end{bmatrix}
		,\dots
		}
		}_{G^\ast}
		}
			\begin{bmatrix}
				a\\
				b,c
			\end{bmatrix}
		,\dots,
		\overbrace{
			\begin{bmatrix}
				a,b\\
				c
			\end{bmatrix}
		,\dots,
			\begin{bmatrix}
				b\\
				a,c
			\end{bmatrix}
		,\dots
		}^{G}\,
		,
		\begin{bmatrix}
			a,b,c
		\end{bmatrix}
		,\dots
		)\text.
	\]
	The profiles $R^G$ exist by~\ref{eq:richdich}.
	Let
	\begin{align}
		w_G = \frac{\phi^G(b,c)}{\phi^G(a,c)}\text.\label{eq:wg}
	\end{align}
	Since $R^G|_{\{a,c\}} = \bar R|_{\{a,c\}}$ and $f$ satisfies IIA, we have that $\phi^G(a,c)\neq 0$ and hence, $w_G$ is well-defined.
	Intuitively, $w_G$ is the weight of the agents in $G$ relative to the weight of the agents in $G^\ast$. 
	
	We will first show that $w_G$ is independent of the choice of $a,b,c$ and $R^G$.
	To this end, let $R\in\mathcal D^N$, $f(R)\equiv \phi$, and $x,y,z\in \mathcal U$ (not necessarily distinct from $a,b,c$) such that $R^{(xa)(yb)(zc)}|_{\{a,b,c\}} = R^G|_{\{a,b,c\}}$, where $(xa)(yb)(zc)$ is the permutation that swaps $x$ with $a$, $y$ with $b$, and $z$ with $c$.
	We first consider the case that $x = a$, $y = b$, and $z\in \mathcal U\setminus\{a,b,c\}$.
	Let $R'\in\mathcal D^N$, $f(R')\equiv\phi'$, such that $R'|_{\{a,b,c\}} = R^G|_{\{a,b,c\}}$ and $R'|_{\{a,b,z\}} = R|_{\{a,b,z\}}$, which exists by~\ref{eq:richdich}.
	Since $f$ satisfies IIA, we have that $\phi'|_{\{a,b,c\}} \equiv \phi^G|_{\{a,b,c\}}$ and $\phi'|_{\{a,b,z\}} \equiv \phi|_{\{a,b,z\}}$.
	Note that, by the choice of $R$ and $R^G$, we have that $I'_{cz} = N$.
	Since $f$ satisfies Pareto indifference, it follows that $\phi'(c,z) = 0$.
	Using the definition of $\phi'$ for the first and the third equality and~\eqref{eq:trans} for the second equality, we get
	\[
		\frac{\phi(b,z)}{\phi(a,z)} = \frac{\phi'(b,z)}{\phi'(a,z)} = \frac{\phi'(b,c)}{\phi'(a,c)} = \frac{\phi^G(b,c)}{\phi^G(a,c)} = w_G\text.
	\]
	Repeated application of this case yields the desired statement for arbitrary $x,y,z$.
	
	Next we show that $w$ is additive, i.e., for all $G,\hat G\subseteq N$ with $G\cap \hat G= \emptyset$,
	\begin{align}
		w_G + w_{\hat G} = w_{G\cup \hat G}\text.\label{eq:additive}
	\end{align}
	To this end, let $R\in\mathcal D^N$, $f(R)\equiv\phi$, such that $N_{bd} = G$, $N_{cd} = \hat G$, $N_{ad} = G^\ast$, and $N_{da} = N_{db} = N_{dc} = \emptyset$, i.e., $d$ is least preferred among $\{a,b,c,d\}$ by all agents.
	The profile $R$ exists by~\ref{eq:richdich}.
	By \eqref{eq:wg}, we have that $\phi(b,d) = w_G\phi(a,d)$ and $\phi(c,d) = w_{\hat G}\phi(a,d)$.
	Moreover, let $\bar R\in\mathcal D^N$, $f(\bar R)\equiv \bar\phi$, such that $\bar N_{bd} = G$, $\bar N_{ad} = G\cup\hat G$, $\bar N_{cd} = G^\ast$, and $\bar N_{da} = \bar N_{db} = \bar N_{dc} = \emptyset$, which exists by~\ref{eq:richdich}.
	By \eqref{eq:wg}, we have that $\bar \phi(b,d) = w_G\bar \phi(c,d)$ and $\bar \phi(a,d) = w_{G\cup \hat G}\bar \phi(c,d)$.
	Let $\tilde R\in\mathcal D^N$, $f(\tilde R) \equiv \tilde\phi$, such that $\tilde R|_{\{b,c,d\}} = R|_{\{b,c,d\}}$ and $\tilde R|_{\{a,b,d\}} = \bar R|_{\{a,b,d\}}$, which completely determines $\tilde R|_{\{a,b,c,d\}}$, since $\mathcal D \subseteq \mathcal R^\dich$.
	The profile $\tilde R$ exists by~\ref{eq:richdich} and is depicted below.
		\[
			\tilde R|_{\{a,b,c,d\}}=(
			\underbrace{
			\begin{bmatrix}
				a,b\\
				c,d
			\end{bmatrix}
			,\dots}_{G}\,
	,
			\underbrace{
			\begin{bmatrix}
				a,c\\
				b,d
			\end{bmatrix}
			,\dots}_{\hat G}\,
			, 
			\begin{bmatrix}
				a,b,c,d
			\end{bmatrix}
			,\dots
			)
		\]
	Since $f$ satisfies IIA, we have that $\tilde\phi|_{\{b,c,d\}}\equiv\phi|_{\{b,c,d\}}$.
	Hence, there is $\alpha > 0$ such that $\tilde\phi(b,d) = \alpha w_G$ and $\tilde\phi(c,d) = \alpha w_{\hat G}$.
	Again, since $f$ satisfies IIA, we have that $\tilde\phi|_{\{a,b,d\}}\equiv\bar\phi|_{\{a,b,d\}}$.
	Hence, there is $\beta > 0$ such that $\tilde\phi(b,d) = \beta w_G$ and $\tilde\phi(a,d) = \beta w_{G\cup \hat G}$.
	Since $\mathcal D \subseteq \mathcal R^\dich\subseteq\mathcal R^\vnm$, $f(\tilde R)\in \mathcal R^{\vnm}$, and $f$ satisfies Pareto indifference, \citeauthor{Hars55a}'s Social Aggregation Theorem \citeyearpar{Hars55a} implies that there are $v_i\in\mathbb R$ for all $i\in N$ such that $\tilde\phi\equiv \sum_{i\in N} v_i\tilde\phi_i$ where $(\tilde\phi_i)_{i \in N}\equiv \tilde R$.
	Thus,
	\[
		\tilde\phi(a,d) = \sum_{i\in G\cup \hat G} v_i = \sum_{i\in G} v_i + \sum_{i\in\hat G} v_i = \tilde\phi(b,d) + \tilde\phi(c,d)\text.
	\]
	If $\tilde\phi(b,d) = \tilde\phi(c,d) = 0$, then $\tilde\phi(a,d) = 0$ and $w_G = w_{\hat G} = w_{G\cup\hat G} = 0$.
	In particular, $w_{G\cup \hat G} = w_G + w_{\hat G}$ and we are done.
	Otherwise, we may assume without loss of generality that $\tilde\phi(b,d)\neq 0$.
	This implies that $\alpha = \beta$.
	Hence, we have that
	\[
		\alpha w_{G\cup\hat G} = \tilde\phi(a,d) = \tilde\phi(b,d) + \tilde\phi(c,d) = \alpha(w_G + w_{\hat G}), 
	\]
	which proves the desired statement.
	
	Hence, to prove that $f$ is affine utilitarian with weights $w_i = w_{\{i\}}$ for all $i\in N$, it suffices to show that, for all $R\in\mathcal D^N$, $f(R) \equiv \phi$, there is $\alpha > 0$ such that $\phi(x,y) = \alpha (w_{N_{xy}} - w_{N_{yx}})$ for all $x,y\in \mathcal U$.
	To this end, we first show that $\phi(x,y) = 0$ if and only if $w_{N_{xy}} - w_{N_{yx}} = 0$.
	Without loss of generality, assume that $\{x,y\}\cap \{a,b\} = \emptyset$ and let $\bar R\in\mathcal D^N$ such that $\bar N_{xb} = N_{xy}$, $\bar N_{yb} = N_{yx}$, $\bar N_{ab} = G^\ast$, and $\bar N_{ba} = \bar N_{bx} = \bar N_{by}=\emptyset$, which exists by~\ref{eq:richdich}.
	Note that this implies $\bar R_{\{x,y\}} = R_{\{x,y\}}$ since $\mathcal D\subseteq\mathcal R^\dich$.
	Applying \eqref{eq:wg} and \eqref{eq:trans} for the first and second equality, respectively, we have that
	\[
		w_{N_{xy}} - w_{N_{yx}} = \frac{\bar\phi(x,b)}{\bar\phi(a,b)} - \frac{\bar\phi(y,b)}{\bar\phi(a,b)} = \frac{\bar\phi(x,y)}{\bar\phi(a,b)}\text.
	\]
	Since $f$ satisfies IIA, $\phi(x,y)$ is a positive multiple of $\bar\phi(x,y)$ from which the desired relationship follows.
	
	Now if $\phi = 0$, it follows that, for all $x,y\in \mathcal U$, $\phi(x,y) = w_{N_{xy}} - w_{N_{yx}} = 0$ and we can choose $\alpha > 0$ arbitrarily.
	Otherwise, we may assume without loss of generality that $\phi(a,b)\neq 0$, which implies that $w_{N_{ab}} - w_{N_{ba}}\neq 0$.
	Let $\alpha \in\mathbb R$ such that $\phi(a,b) = \alpha (w_{N_{ab}} - w_{N_{ba}})$.
	We aim to show that, for all $x,y\in \mathcal U$, $\phi(x,y) = \alpha (w_{N_{xy}} - w_{N_{yx}})$.
	To this end, let $x,y\in \mathcal U$.
	If $\{x,y\} = \{a,b\}$, this is clear by skew-symmetry of $\phi$.
	Otherwise we may assume without loss of generality that $x\in \mathcal U\setminus\{a,b\}$ and, by \eqref{eq:trans}, that $\phi(a,x)\neq 0$.
	Let $\hat R\in\mathcal D^N$, $f(\hat R)\equiv\hat\phi$, such that $R|_{\{a,b,x\}} = \hat R|_{\{a,b,x\}}$ and, $\hat N_{ya} = \hat N_{yb} = \hat N_{yc} = \emptyset$, which exists by~\ref{eq:richdich}.
	Note that $\hat N_{ya} = \emptyset$ implies that $\hat N_{by}\cap\hat N_{ba} = \hat N_{ba}$, since $\mathcal D\subseteq\mathcal R^\dich\subseteq\mathcal R^\vnm$ only contains transitive preference relations; likewise, $\hat N_{ay}\cap\hat N_{ab} = \hat N_{ab}$.
	From before, we have that there is $\gamma > 0$ such that $\hat\phi(a,y) = \gamma w_{\hat N_{ay}}$, $\hat\phi(b,y) = \gamma w_{\hat N_{by}}$, and $\hat\phi(x,y) = \gamma w_{\hat N_{xy}}$.
	Hence, we have that
	\begin{align*}
		\hat\phi(a,b) &= \hat\phi(a,y) - \hat\phi(b,y) = \gamma(w_{\hat N_{ay}} - w_{\hat N_{by}})\\ 
		&= \gamma ((w_{\underbrace{\hat N_{ay}\cap \hat N_{ab}}_{\hat N_{ab}}} + w_{\hat N_{ay}\cap \hat I_{ab}}) - (w_{\underbrace{\hat N_{by}\cap \hat N_{ba}}_{\hat N_{ba}}} + w_{\underbrace{\hat N_{by}\cap \hat I_{ab}}_{\hat N_{ay}\cap\hat I_{ab}}}))
		= \gamma (w_{\hat N_{ab}} - w_{\hat N_{ba}})\text,
	\end{align*}
	where the first equality follows from~\eqref{eq:trans} and skew-symmetry of $\hat\phi$ and the third equality follows from~\eqref{eq:additive}.
	Similarly, we get that $\hat\phi(a,x) = \gamma(w_{\hat N_{ax}} - w_{\hat N_{xa}})$.
	 
	 Since $f$ satisfies IIA, it follows that $\phi|_{\{a,b,x\}} \equiv \hat\phi|_{\{a,b,x\}}$.
	 Hence, $\phi(a,x) = \alpha (w_{N_{ax}} - w_{N_{xa}})$, and, by a similar argument, $\phi(a,y) = \alpha (w_{N_{ay}} - w_{N_{ya}})$.
	 Observe that by~\eqref{eq:additive} and $\mathcal D\subseteq\mathcal R^{\dich}$,
	 \[
	 w_{N_{xy}} = w_{N_{xa}\cap N_{xy}} + w_{N_{ay}\cap N_{xy}} = w_{N_{xa}} - w_{N_{ya}\cap N_{xa}} + w_{N_{ay}} - w_{N_{ax}\cap N_{ay}}
	 \]
	 and
	 \[
	 w_{N_{yx}} = w_{N_{ya}\cap N_{yx}} + w_{N_{ax}\cap N_{yx}} = w_{N_{ya}} - w_{N_{xa}\cap N_{ya}} + w_{N_{ax}} - w_{N_{ay}\cap N_{ax}}.
	 \]
	 Note that the minus terms are the same in both of the above equalities.
	 Then, we have that
	 \begin{align*}
	 	\phi(x,y) &= \phi(x,a) - \phi(y,a) = \alpha ( w_{N_{xa}} - w_{N_{ax}} - w_{N_{ya}} + w_{N_{ay}}) = \alpha( w_{N_{xy}} - w_{N_{yx}})\text,
	 \end{align*}
	 where the first equality follows from~\eqref{eq:trans} and skew-symmetry of $\phi$, the second equality follows from skew-symmetry of $\phi$, and the third equality follows from the statements about $w_{N_{xy}}$ and $w_{N_{yx}}$ derived above.
	 This proves the desired equation.
\end{proof}

When assuming that $f$ satisfies full Pareto optimality instead of Pareto indifference, the weights of all agents have to be positive and we obtain the following theorem.

\begin{theorem}\label{thm:dichswf}
	Let $f$ be an Arrovian SWF with range $\mathcal R^{\vnm}$ on some domain $\mathcal D\subseteq\mathcal R^{\dich}$ satisfying~\ref{eq:richdich} with $|\mathcal U|\ge 4$. 
	Then, $f$ is affine utilitarian with positive weights.
\end{theorem}

\Cref{thm:dichswf} follows from \Cref{lem:dichmlvnm} in the same way as \Cref{thm:swf} follows from \Cref{lem:mlscheme3}. 
Its proof is therefore omitted.
If we additionally assume that $f$ is anonymous, the weights of all agents have to be equal, and outcomes are ordered by the vNM utility function that assigns to each alternative the number of agents who approve it.

\begin{corollary}\label{cor:dichmlvnm}
	Let $|\mathcal U|\ge 4$ and $\mathcal D\subseteq\mathcal R^{\vnm}$ be some domain satisfying~\ref{eq:richdich}.
	An anonymous SWF with range $\mathcal R^{\vnm}$ is Arrovian if and only if it is the utilitarian SWF and $\mathcal{D}\subseteq \mathcal{R}^\dich$.\end{corollary}

\begin{remark}[Tightness of Bound]
\Cref{thm:dichswf} does not hold if~$\mathcal |U|<4$.
Let $\mathcal U = \{a,b,c\}$ and consider the SWF $f$ defined as follows.
For all $R\in\mathcal D^N$ and $x,y,z\in \mathcal U$, $f(R)\equiv \phi$ with $\phi(x,y)=\phi(y,z)=1$ and $\phi(x,z)=2$
if $|N_{xy}|>|N_{yx}|$ and $|N_{yz}|>|N_{zy}|$, and otherwise $f(R)\equiv \sum_{i\in N} \phi_i$ where $(\phi_i)_{i\in N}\in \Phi^N$ such that $(\phi_i)_{i\in N}\equiv R$.
In the former case, $|N_{xz}|>|N_{zx}|$, since the agents' preferences are dichotomous. 
It can be checked that $f$ satisfies Pareto optimality and IIA.
Note that, in line with \citeauthor{MaMo15a}'s \citeyear{MaMo15a} theorem, the collective preferences over \emph{pure} outcomes returned by $f$ coincide with the utilitarian SWF (see also \Cref{fn:MaMo}).
\end{remark}

\begin{remark}[SWFs with range $\mathcal R$]
	Theorems~\ref{thm:dichdomain} and~\ref{thm:dichswf} do not hold for SWFs with range $\mathcal{R}$. 
	To see this for \Cref{thm:dichdomain}, let $N = \{1,2\}$, $\mathcal U = \{a,b,c\}$, and ${\pref_0}\in \mathcal R^\vnm$ such that
	\[
		{\pref_0} 
		\equiv 
		\begin{pmatrix}
			0 & 1 & 2 \\
			-1 & 0 & 1\\
			-2 & -1 & 0\\
		\end{pmatrix}\text.
	\]
	Consider the domain $\mathcal D = \{{\pref_0}^\pi\colon \pi\in\Pi_\mathcal U\}\cup\{\emptyset\}$, which satisfies~\ref{eq:richneutral},~\ref{eq:richindiff}, and~\ref{eq:richinverse}.
	Then, the following SWF $f$ is an anonymous Arrovian SWF on $\mathcal D$.
	For all $R\in\mathcal D^N$ and $x,y,z\in\mathcal U$, $f(R)\equiv \phi$ with $\phi(x,y) = 1$ and $\phi(x,z) = \phi(y,z) = 0$

		if $x\pref_1 y\pref_1 z$ and $z\pref_2 x\pref_2 y$ and $f(R)\equiv \phi_1+\phi_2$ otherwise. IIA is easy to verify, since this SWF is consistent with majority rule on pairs of pure outcomes. 
		Pareto optimality is clearly satisfied for all profiles where $f$ coincides with the utilitarian rule. 
		In the remaining profiles, Pareto optimality can be verified by simple but tedious calculations.
		
		A counterexample using two agents and four alternatives can be constructed to show that \Cref{thm:dichswf} does not hold for SWFs with range $\mathcal{R}$. 
\end{remark}

\section{Rationalizability via Continuous and Convex Relations}
\label{app:choice}

Recall from \Cref{sec:ratchoice} that a choice function is an upper hemi-continuous function $S\colon \mathcal F(\Delta)\rightarrow \mathcal F(\Delta)$ such that for all $X\in \mathcal F(\Delta)$, $S(X)\subseteq X$, and for all $p,q\in\Delta$, $S([p,q])\in\{\{p\},\{q\},[p,q]\}$.
To formally define upper hemi-continuity, we need to introduce a notion of convergence for sequences in $\mathcal F(\Delta)$. 
To this end, for all $X,Y\in\mathcal F(\Delta)$, let $\dist(X,Y) = \max\{\sup_{x\in X}\inf_{y\in Y} |x-y|, \sup_{y\in Y}\inf_{x\in X} |x-y|\}$ be the Hausdorff distance of $X$ and $Y$. 
With this definition, $\dist$ is a metric on $\mathcal F(\Delta)$ and we say that a sequence $(X_i)_{i\in\mathbb N}\subseteq\mathcal F(\Delta)$ converges to $X\in\mathcal F(\Delta)$, written $X_i\rightarrow X$, if $\dist(X_i,X)$ goes to $0$ as $i$ goes to infinity.
A choice function $S$ is upper hemi-continuous if, for all $(X_i)_{i\in\mathbb N}$ with $X_i\rightarrow X\in\mathcal F(\Delta)$ and $(p_i)_{i\in \mathbb N}$ with $p_i\in S(X_i)$ for all $i\in \mathbb N$ and $p_i\rightarrow p\in\Delta$, $p\in S(X)$.

\proprationalizable*

\begin{proof}
	Let $S$ be a choice function.
	First we prove the ``only if'' part.
	To this end, assume that $S$ is rationalizable by a continuous and convex relation $\succ$, i.e., for all $X\in\mathcal F(\Delta)$, $S(X) = \max_{\succ} X$.
	For $X,Y\in\mathcal F(\Delta)$ with $X\cap Y\neq\emptyset$, $S(X) \cap Y = (\max_\succ X)\cap Y \subseteq \max_\succ (X\cap Y) = S(X\cap Y)$, since a maximal element in $X$ is also a maximal element in any subset of $X$ (without imposing any restrictions on $\succ$).
	Hence, $S$ satisfies contraction.
	To see that $S$ satisfies expansion, let $X,Y\in\mathcal F(\Delta)$ and observe that $S(X)\cap S(Y) = \max_\succ X\cap\max_\succ Y\subseteq \max_\succ(\conv(X\cup Y)) = S(X\cup Y)$, where the set inclusion follows from the fact that, by convexity of $\succ$, weak lower contour sets are convex.

	Second we prove the ``if part'', i.e., $S$ is rationalizable by a continuous and convex relation if it satisfies contraction and expansion. 
	To this end, define $\pref$ as the base relation of $S$, i.e., for all $p,q\in\Delta$, $p\succsim q$ if and only if $p\in S([p,q])$.
	Note that $\succsim$ is a complete relation, since, for all $p,q\in\Delta$, $\{p,q\}\cap S([p,q])\neq\emptyset$ by our definition of choice functions.
	
	First we show that $\pref$ rationalizes $S$, i.e., for all $X\in\mathcal F(\Delta)$, $S(X) = \max_{\pref} X$.
	To see that $S(X) \subseteq \max_{\pref} X$, let $p\in S(X)$.
	Since $S$ satisfies contraction, it follows that, for all $q\in X$, $p\in S([p,q])$ and hence, by definition of $\pref$, $p\succsim q$.
	This implies that $p\in\max_\pref X$.
	To prove that $\max_{\pref} X \subseteq S(X)$, let $p\in \max_\pref X$ and $(X^k)_{k\in\mathbb N}$ be a sequence of polytopes in $\Delta$ such that, for all $k\in\mathbb N$, $p\in X^k$, $X^k = \conv(p^{k,1},\dots,p^{k,l_k})\subseteq X$, $l_k\in\mathbb N$, and $X^k\rightarrow X$.
	Since $p\in\max_\pref X$, we have that, for all $k\in\mathbb N$ and $l\in\{1,\dots,l_k\}$, $p\succsim p^{k,l}$ and thus, by definition of $\pref$, $p\in S([p,p^{k,l}])$.
	Repeated application of expansion implies that, for all $k\in\mathbb N$, $p\in \conv(p^{k,1},\dots,p^{k,l_k}) = X^k$.
	Then, since $X^k\rightarrow X$, continuity of $S$ implies that $p\in S(X)$.
	
	It remains to be shown that $\succ$ satisfies continuity and convexity.
	To prove continuity, let $p\in\Delta$.
	First, assume for contradiction that $L(p)$ is not open at $q\in L(p)$.
	Then, there is a sequence $(q^k)_{k\in\mathbb N}\subseteq\Delta$ such that $q^k$ goes to $q$ as $k$ goes to infinity and $q^k\succsim p$ for all $k\in\mathbb N$.
	Since then $[p,q^k] \rightarrow [p,q]$ and $q^k\in S([p,q^k])$ for all $k\in\mathbb N$ by definition of $\succ$, continuity of $S$ implies that $q\in S([p,q])$, which contradicts $p\succ q$.
	Second, assume for contradiction that $U(p)$ is not open at $q\in U(p)$.
	Then, there is a sequence $(q^k)_{k\in\mathbb N}\subseteq\Delta$ such that $q^k$ goes to $q$ as $k$ goes to infinity and $p\succsim q^k$ for all $k\in\mathbb N$.
	Since then $[p,q^k] \rightarrow [p,q]$ and $p\in S([p,q^k])$ for all $k\in\mathbb N$ by definition of $\succ$, continuity of $S$ implies that $p\in S([p,q])$, which contradicts $q\succ p$.
	Hence, $\succ$ is continuous and, by \Cref{lem:archcont}, also satisfies Archimedean continuity.
	
	Proving convexity of $\succ$ turns out be rather involved. We first prove three useful auxiliary statements.
	We say that a preference relation $\succ$ is \emph{all indifferent} on $X\in\mathcal F(\Delta)$ if ${\succ}|_X = \emptyset$, i.e., for all $p,q\in X$, $p\sim q$.
	
	\begin{claim}\label{claim:allindifferent}
		For all $p,q,r\in\Delta$, if $p\sim q$, $p\sim r$, and $q\sim r$, then $\succ$ is all indifferent on $\conv(p,q,r)$.
	\end{claim}
	\begin{proof}[Proof of \Cref{claim:allindifferent}]	 
	To see this, observe that, by definition of $\succ$, $p\sim q$ and $p\sim r$ implies that $p\in S([p,q])$ and $p\in S([p,r])$. 
	Since $S$ satisfies expansion, we have that $p\in S(\conv([p,q]\cup [p,r])) = S(\conv(p,q,r))$.
	Similarly, $q,r\in S(\conv(p,q,r))$.
	Since $S(\conv(p,q,r))\in\mathcal F(\Delta)$ is convex, it follows that $S(\conv(p,q,r)) = \conv(p,q,r)$.
	Hence, since $S$ satisfies contraction, for all $s,t\in\conv(p,q,r)$, $[s,t] = S(\conv(p,q,r))\cap[s,t]\subseteq S([s,t])$.
	By definition of $\succ$, we have that $s\sim t$, which proves that $\succ$ is all indifferent on $\conv(p,q,r)$.
	\end{proof}
	
	We say that $\succ$ satisfies \emph{\ref{eq:weakbetweenness}} if for all $p,q\in\Delta$ with $p\succsim q$, there is $\lambda^\ast\in[0,1]$ such that, for all $\mu, \mu'\in[0,1]$ with $\mu > \mu'$,
	\begin{align}
		\begin{cases}
			p\mu q \succ p\mu' q\quad&\text{if } \mu > \lambda^\ast\text,\\
			p\mu q \sim p\mu' q&\text{if } \mu \le \lambda^\ast\text.\\
		\end{cases}
		\tag{lower betweenness}
		\label{eq:weakbetweenness}
	\end{align}
	If the above holds, we say that $\succ$ satisfies \emph{$\lambda^\ast$-\ref{eq:weakbetweenness}} on $[p,q]$.
	Intuitively, $\lambda^\ast$-\ref{eq:weakbetweenness} on $[p,q]$ prescribes that, when moving along the line from $p$ to $q$, preference strictly decreases until $p\lambda^\ast q$ is reached and then remains constant. 
	
	\begin{claim}\label{claim:lowerbetweenness}
		$\succ$ satisfies \ref{eq:weakbetweenness}.
	\end{claim}
	\begin{proof}[Proof of \Cref{claim:lowerbetweenness}]
		Let $p,q\in\Delta$ with $p\succsim q$.
		We distinguish two cases.
		If $S([p,q]) = [p,q]$, then $p\sim q$ by definition of $\succ$.
		From \Cref{claim:allindifferent} (with $r = p$), it follows that $\succ$ is all indifferent on $[p,q]$ and we may choose $\lambda^\ast = 1$.

		If $S([p,q]) \neq [p,q]$, we have that $p\not\sim q$ by definition of $\succ$.
		Since $p\succsim q$ by assumption, it follows that $p\succ q$.
		Let $\lambda^\ast = \sup_{\lambda\in[0,1]}\{q\in S([p\lambda q, q])\}$.
		Note that, since $p\succ q$ and $S$ satisfies continuity, $\lambda^\ast < 1$.
		By continuity of $S$ and definition of $\lambda^\ast$, $q\in S([p\lambda^\ast q,q])$ and, since $p\lambda'q\in S([p\lambda' q,q])$ for all $\lambda' > \lambda^\ast$, $p\lambda^\ast q\in S([p\lambda^\ast q,q])$.
		Together this yields that $S([p\lambda^\ast q,q]) = [p\lambda^\ast q,q]$ and hence, by \Cref{claim:allindifferent}, $\succ$ is all indifferent on $[p\lambda^\ast q,q]$.
		Now we show that, for all $\lambda\in [0,1]$, $p\lambda q \succsim q$.
		If $\lambda > \lambda^\ast$, this follows from the definition of $\lambda^\ast$.
		If $\lambda \le \lambda^\ast$, it follows from the fact that $\succ$ is all indifferent on $[p\lambda^\ast q,q]$ that $p\lambda q\sim q$.
		Now let $\mu,\mu'\in[0,1]$ such that $\mu > \mu'$.
		If $\mu > \lambda^\ast$, assume for contradiction that $p\mu' q\succsim p\mu q$.
		By definition of $\succ$, this implies that $p\mu' q\in S([p\mu q,p\mu' q])$.
		From $p\mu' q\succsim q$, it follows that $p\mu' q\in S([p\mu' q, q])$.
		Then, expansion implies that $p\mu' q \in S(\conv([p\mu q, p\mu' q]\cup [p\mu'q, q])) = S([p\mu q, q])$.
		By our definition of choice functions, $p\mu' q\in S([p\mu q, q])$ implies $S([p\mu q, q])=[p\mu q, q]$. Hence, in particular, $q\in S([p\mu q, q])$, which contradicts the assumption that $\mu>\lambda^\ast$.
		If $\mu\le\lambda^\ast$, $p\mu q\sim p\mu' q$ follows from the fact that $\succ$ is all indifferent on $[p\lambda^\ast q,q]$.     
	\end{proof}

	\begin{claim}\label{claim:convexindiff}
		For all $p\in\Delta$, $I(p)$ is convex.
	\end{claim}
	\begin{proof}[Proof of \Cref{claim:convexindiff}]
		Let $p,q,r\in\Delta$ and $\lambda\in[0,1]$ such that $p\sim q$ and $p\sim r$.
		We show that $p\sim q\lambda r$.
		The proof proceeds as follows.
		It is easy derived that $p$ is weakly preferred to $q\lambda r$.
		So we assume for contradiction that $p$ is strictly preferred to $q\lambda r$.
		By \Cref{claim:lowerbetweenness}, there is some $\lambda^\ast$ between $0$ and $1$ such that preference strictly decreases when moving along a straight line from $p$ to $p\lambda^\ast(q\lambda r)$.
		In particular, $\succ$ is not all indifferent on any non-trivial subinterval of this line segment.
		In the remainder of the proof, we construct a subset of $\conv(p,q,r)$ (which will turn out to be $\conv(p\bar \mu r^\ast,p\hat\mu r^\ast,t)$) on which $\succ$ is all indifferent that contains a non-trivial subinterval of the line segment from $p$ to $p\lambda^\ast(q\lambda r)$, which is a contradiction.
		The proof of this claim is illustrated in \Cref{fig:convexindiff}.
		The proofs of~\ref{item:l} and~\ref{item:u} below follow a similar structure. 
		
		If $q\sim r$, it follows from \Cref{claim:allindifferent} that $\succ$ is all indifferent on $\conv(p,q,r)$, which implies that $p\sim q\lambda r$.
		Otherwise, we may assume without loss of generality that $q\succ r$.
		Since $p\sim q$, $p\sim r$, and $S$ satisfies expansion, it follows that $p\in S(\conv(p,q,r))$.
		Hence, since $S$ satisfies contraction, for all $\lambda\in [0,1]$, $p\in S(\conv(p,q\lambda r))$, i.e., $p\succsim q\lambda r$.
		Assume for contradiction that $p\succ q\lambda r$ for some $\lambda\in [0,1]$.
		Let $\lambda^+ = \sup_{\lambda\in[0,1]} \{p\succ q\lambda r\}$, $\lambda^- = \inf_{\lambda\in[0,1]} \{p\succ q\lambda r\}$, $q^\ast = q\lambda^+ r$, and $r^\ast = q\lambda^- r$.
		Since, by continuity of $\succ$, $L(p)\cap [q,r]$ is open in $[q,r]$, it follows that $\lambda^+ > \lambda > \lambda^-$ and $p\sim q^\ast$ and $p\sim r^\ast$.
		Moreover, since $q\succ r$ it follows from \Cref{claim:lowerbetweenness} that $q^\ast\sim r^\ast$ or $q^\ast\succ r^\ast$.
		If $q^\ast\sim r^\ast$, $\succ$ is all indifferent on $\conv(p,q^\ast, r^\ast)$ by \Cref{claim:allindifferent}, which contradicts $p\succ q\lambda r$.
		
		In the rest of the proof of this claim, we will consider the case $q^\ast\succ r^\ast$.
		First we show that, for all $\mu\in[0,1)$, $q^\ast\succ p\mu r^\ast$.
		Since $p\sim q^\ast$ and $q^\ast\succ r^\ast$, the fact that $S$ satisfies expansion implies that $q^\ast\in S(\conv(p,q^\ast,r^\ast))$.
		Thus, the fact that $S$ satisfies contraction implies that $q^\ast\in S([q^\ast,s])$ for all $s\in \conv(p,q^\ast,r^\ast)$.
		If $q^\ast\sim p\mu r^\ast$ for some $\mu\in[0,1)$, then we have that $p\sim q^\ast$, $p\sim p\mu r^\ast$, and $q^\ast\sim p\mu r^\ast$, which, by \Cref{claim:allindifferent}, implies $\succ$ is all indifferent on $\conv(p,q^\ast,p\mu r^\ast)$.
		By \Cref{claim:lowerbetweenness}, there is $\lambda^\ast\in[0,1)$ such that $\succ$ satisfies $\lambda^\ast$-\ref{eq:weakbetweenness} on $[p,q\lambda r]$.
		Hence, for $\mu,\mu'$ close to $1$, $p\mu(q\lambda r)\succ p\mu'(q\lambda r)$ and $p\mu(q\lambda r), p\mu'(q\lambda r)\in\conv(p,q^\ast,p\mu r^\ast)$, which is a contradiction to the fact that $\succ$ is all indifferent on $\conv(p,q^\ast,p\mu r^\ast)$.
		Hence, for all $\mu\in[0,1)$, $q^\ast\succ p\mu r^\ast$.
		From this it follows by \Cref{claim:lowerbetweenness} that, for all $s\in \conv(p,q^\ast,r^\ast)\setminus[p,q^\ast]$, $q^\ast\succ s$.

		We now show that there are $\bar\lambda,\hat\lambda\in[0,1]$ and $\bar\mu,\hat\mu\in[0,1]$ such that $\bar\lambda > \hat\lambda > \lambda$ and $\bar\mu > \hat\mu$, $q^\ast\bar\lambda r^\ast\sim p\bar\mu r^\ast$ and $q^\ast\hat\lambda r^\ast\sim p\hat\mu r^\ast$, and $p\lambda^\ast(q\lambda r)$ is in the relative interior of $\conv(p\bar\mu r, q^\ast\bar\lambda r^\ast, r^\ast)$.
		By definition of $q^\ast$, there is a sequence $(\lambda_k)_{k\in\mathbb N}$ such that $\lambda_k$ goes to $1$ and, for all $k\in\mathbb N$, $p\succ q^\ast\lambda_k r^\ast$.
		Observe that, for all $k\in\mathbb N$, $p\succ q^\ast\lambda_k r^\ast\succsim r^\ast$.
		Using the fact that $\succ$ satisfies Archimedean continuity by \Cref{lem:archcont} it follows that, for all $k\in\mathbb N$, there is $\mu_k\in[0,1)$ such that $q^\ast\lambda_k r^\ast\sim p\mu_k r^\ast$.
		By definition of $\succ$, for all $k\in\mathbb N$, $p\mu_k r^\ast\in S([q^\ast\lambda_k r^\ast, p\mu_k r^\ast])$.
		Since $q^\ast\lambda_k r^\ast\rightarrow q^\ast$, it follows from continuity of $S$ that $p\mu^\ast r^\ast\in S([p\mu^\ast r^\ast,q^\ast])$, where $\mu^\ast$ is an accumulation point of $(\mu_k)_{k\in \mathbb N}$.
		Since $q^\ast\succ s$ for all $s\in \conv(p,q^\ast,r^\ast)\setminus[p,q^\ast]$, it follows that $\mu^\ast = 1$, i.e., $\mu_k$ goes to $1$ as $\lambda_k$ goes to $1$.
		Now let $\bar k,\hat k\in\mathbb N$ such that $\bar\lambda = \lambda_{\bar k} >\hat\lambda = \lambda_{\hat k} > \lambda$, $\bar\mu = \mu_{\bar k} > \hat\mu = \mu_{\hat k}$ and $p\lambda^\ast(q\lambda r)$ is in the relative interior of $\conv(p\bar\mu r, q^\ast\bar\lambda r^\ast, r^\ast)$.
		Such $\bar k,\hat k$ exist, since $\lambda_k$ and $\mu_{k}$ go to $1$ as $k$ goes to infinity.
		
		Lastly, we show that $\succ$ is all indifferent on a line segment that properly intersect that line segment from $p$ to $p\lambda^\ast(q\lambda r)$, which will yield a contradiction.
		Let $t\in [p\bar\mu r^\ast,q^\ast\bar\lambda r^\ast]\cap [p,q^\ast\hat\lambda r^\ast]$, which exists, since $\bar\lambda >\hat\lambda$.
		Since $p\hat\mu r^\ast\sim p$, $p\hat\mu r^\ast\sim q^\ast\hat\lambda r^\ast$, and $S$ satisfies expansion, we have that $p\hat\mu r^\ast\in S(\conv(p,p\hat\mu r^\ast, q^\ast\hat\lambda r^\ast))$.
		As $S$ satisfies contraction, it follows that $p\hat\mu r^\ast \succsim t$.
		
		Conversely, from $p\sim q^\ast$, $p\sim r^\ast$, and $q^\ast\succ r^\ast$, it follows that $p,q^\ast\in S(\conv(p,q^\ast,r^\ast))$ and hence, $[p,q^\ast]\subseteq S(\conv(p,q^\ast,r^\ast))$. 
		Let $\sigma\in(0,1)$ such that $t\in [p\sigma q^\ast,p\hat\mu r^\ast]$.
		Since $p\sigma q^\ast\in[p,q^\ast]$, it follows from the fact that $S$ satisfies contraction that $p\sigma q^\ast\in S([p\sigma q^\ast, p\hat\mu r^\ast])$ and so $p\sigma q^\ast \succsim p\hat\mu r^\ast$ by definition of $\succ$.
		Assuming that $p\hat\mu r^\ast\succ t$ would, by \Cref{claim:lowerbetweenness}, imply that $p\hat\mu r^\ast\succ p\sigma q^\ast$ and thus a contradiction. 
		So $t\succsim p\hat\mu r^\ast$.
		
		In summary, we have $p\hat\mu r^\ast\sim t$.
		Additionally, it holds that $p\bar\mu r^\ast\sim p\hat\mu r^\ast$ (since $p\sim r^\ast$) and $p\bar\mu r^\ast\sim t$ (since $p\bar\mu r^\ast\sim q^\ast \bar\lambda r^\ast$).
		Thus, $\succ$ is all indifferent on $\conv(p\bar\mu r, p\hat\mu r^\ast, t)$ by \Cref{claim:allindifferent}.
		By construction of $\bar\lambda$, the line segment from $p$ to $p\lambda^\ast(q\lambda r)$ intersects the relative interior of $\conv(p\bar\mu r, p\hat\mu r^\ast, t)$.
		Thus, there are $\lambda',\lambda''\in (0,1)$ such that $\lambda'' > \lambda'\ge \lambda^\ast$ and $p\lambda'(q\lambda r), p\lambda''(q\lambda r)\in \conv(p\bar\mu r^\ast, p\hat\mu r^\ast, t)$, i.e., $p\lambda'(q\lambda r), p\lambda''(q\lambda r)$ are contained in the intersection of the line segment from $p$ to $p\lambda^\ast(q\lambda r)$ and $\conv(p\bar\mu r, p\hat\mu r^\ast, t)$.
		Since $\succ$ satisfies $\lambda^\ast$-lower betweenness on $[p,q\lambda r]$, it holds that $p\lambda'(q\lambda r)\succ p\lambda''(q\lambda r)$, which contradicts the fact that $\succ$ is all indifferent on $\conv(p\bar\mu r^\ast, p\hat\mu r^\ast, t)$.
		Hence, for all $\lambda\in [0,1]$, $p\sim q\lambda r$, which proves the claim.
	\end{proof}

	\begin{figure}[H]
		
			\tikzset{->-/.style={decoration={
			  markings,
			  mark=at position #1 with {\arrow{latex}}},postaction={decorate}}}
			\def\arrow{1.5ex}
			\def\dist{.5ex}
			\def\sc{3}
			\centering
			\begin{subfigure}{0.48\textwidth}
			\centering
			\begin{tikzpicture}[scale = \sc, every node/.style={inner sep=0,outer sep=0}]
				
				\node[label={[label distance=\dist]-150:$p$}] (p) at (-150:1) {};
				\node[label={[label distance=\dist]90:$q$}] (q) at (90:1) {};
				\node[label={[label distance=\dist]-30:$r$}] (r) at (-30:1) {};
				
				\node[label={[label distance=\dist]15:$q\lambda r$}] (qr) at ($(q)!.65!(r)$) {};
				\node[label={[label distance=\dist]55:$q^\ast$}] (qast) at ($(q)!.1!(r)$) {};
				\node[label={[label distance=\dist]5:$r^\ast$}] (rast) at ($(q)!.9!(r)$) {};
				
				\node[label={[label distance=2*\dist]-88:$p\lambda^\ast(q\lambda r)$}] (lambdaast) at ($(p)!.8!(qr)$) {};
				\node[label={[label distance=\dist]5:$q^\ast\bar\lambda r^\ast$}] (lambdabar) at ($(q)!.23!(r)$) {};
				\node[label={[label distance=\dist]5:$q^\ast\hat\lambda r^\ast$}] (lambdahat) at ($(q)!.37!(r)$) {};
				
				\node[label={[label distance=\dist]-90:$p\bar\mu r^\ast$}] (mubar) at ($(p)!.25!(rast)$) {};
				\node[label={[label distance=\dist]-90:$p\hat\mu r^\ast$}] (muhat) at ($(p)!.45!(rast)$) {};
				
				\node[label={[label distance=\dist]150:$t$}] (t) at (intersection of p--lambdahat and lambdabar--mubar) {};
				\node[label={[label distance=\dist]90:$p\sigma q^\ast$}] (sigma) at (intersection of p--qast and t--muhat) {};
	
				\draw[-latex] (q.center) -- (lambdabar.center);
				\draw[-latex, dashed] (lambdabar.center) -- (r.center);
				\draw[-latex] (p.center) -- (lambdaast.center);
				\draw (p.center) -- (r.center);
				\draw (p.center) -- (q.center);
				\draw (p.center) -- (qast.center);
				\draw (p.center) -- (rast.center);
				\draw[dashed,-latex] (lambdaast.center) -- (qr.center);
				\draw (lambdabar.center) -- (mubar.center);
				\draw (lambdahat.center) -- (muhat.center);
				\draw[dashed,-latex] (p.center) -- (lambdahat.center);
				\draw[dashed,-latex] (sigma.center) -- (t.center);
				\draw (t.center) -- (muhat.center);

			\end{tikzpicture}
			\caption{}
			\label{fig:convexindiff}
			\rule{0ex}{3ex}
			\end{subfigure}
			\begin{subfigure}{0.48\textwidth}
				\centering
				\begin{tikzpicture}[scale = \sc, every node/.style={inner sep=0,outer sep=0}]
					
					\node[label={[label distance=\dist]-150:$p$}] (p) at (-150:1) {};
					\node[label={[label distance=\dist]90:$q$}] (q) at (90:1) {};
					\node[label={[label distance=\dist]-30:$r$}] (r) at (-30:1) {};
					
					\node[label={[label distance=\dist]30:$q\lambda r$}] (qr) at ($(q)!.5!(r)$) {};
					\node[label={[label distance=\dist]150:$p\sigma^\ast q$}] (sigmaast) at ($(p)!.5!(q)$) {};
					\node[label={[label distance=\dist]-90:$p\mu^\ast r$}] (muast) at ($(p)!.6!(r)$) {};
					
					\node[label={[label distance=\dist]165:$p\sigma q$}] (sigma) at ($(p)!.3!(q)$) {};
					\node[label={[label distance=\dist]-85:$p\mu_\sigma r$}] (mu) at ($(p)!.3!(r)$) {};
					
					\node[label={[label distance=1.5*\dist]-5:$s$}] (s) at (intersection of p--qr and sigma--mu) {};
					\node[label={[label distance=\dist]-90:$p\mu' r$}] (muprime) at (intersection of q--s and p--r) {};

					\draw[-latex] (p.center) -- (sigma.center);
					\draw[dashed,->-=.9] (sigma.center) -- (q.center);
					\draw[dashed,->-=.9] (muprime.center) -- (r.center);
					\draw[-latex] (p.center) -- (muprime.center);
					\draw[densely dotted] (q.center) -- (r.center);
					\draw (p.center) -- (qr.center);
					\draw[->-=.9] (sigmaast.center) -- (r.center);
					\draw[->-=.9] (muast.center) -- (q.center);
					\draw[-latex] (muprime.center) -- (s.center);
					\draw[->-=.9,dashed] (s.center) -- (q.center);
					\draw (sigma.center) -- (mu.center);

				\end{tikzpicture}	
				\caption{}
				\label{fig:convex1}
				\rule{0ex}{3ex}
			\end{subfigure}
			\begin{subfigure}{0.48\textwidth}
				\centering
				\begin{tikzpicture}[scale = \sc, every node/.style={inner sep=0,outer sep=0}]
					
					\node[label={[label distance=\dist]-150:$p$}] (p) at (-150:1) {};
					\node[label={[label distance=\dist]90:$q$}] (q) at (90:1) {};
					\node[label={[label distance=\dist]-30:$r$}] (r) at (-30:1) {};
					
					\node[label={[label distance=\dist]30:$q\lambda r$}] (qr) at ($(q)!.5!(r)$) {};
					\node[label={[label distance=\dist]-90:$p\mu r$}] (mu) at ($(p)!.5!(r)$) {};
					\node[label={[label distance=\dist]120:$s$}] (s) at ($(p)!.3!(qr)$) {};
					
					\node[label={[label distance=\dist]160:$p\sigma' q$}] (sigmaprime) at ($(p)!.4!(q)$) {};
					\node[label={[label distance=\dist]50:$r\mu^\ast(p\sigma'q)$}] (muast) at ($(r)!.3!(sigmaprime)$) {};

					\node[label={[label distance=\dist]-90:$t$}] (t) at ($(p)!.85!(r)$) {};
					\node[label={[label distance=\dist]170:$t'$}] (tprime) at (intersection of p--qr and q--t) {};

					\draw[-latex] (q.center) -- (sigmaprime.center);
					\draw[->-=.9,dashed] (sigmaprime.center) -- (p.center);
					\draw[-latex] (r.center) -- (t.center);
					\draw[->-=.9,dashed] (t.center) -- (p.center);
					\draw (q.center) -- (r.center);
					\draw (p.center) -- (qr.center);
					\draw[-latex] (q.center) -- (mu.center);
					\draw (mu.center) -- (qr.center);
					\draw (mu.center) -- (s.center);
					\draw[dashed,-latex] (muast.center) -- (sigmaprime.center);
					\draw[-latex] (r.center) -- (muast.center);
					\draw[dashed,-latex] (q.center) -- (tprime.center);
					\draw (t.center) -- (tprime.center);
					\draw (t.center) -- (qr.center);

				\end{tikzpicture}	
				\caption{}
				\label{fig:convex2a}
				\end{subfigure}
				\begin{subfigure}{0.48\textwidth}
					\centering
					\begin{tikzpicture}[scale = \sc, every node/.style={inner sep=0,outer sep=0}]
					
						\node[label={[label distance=\dist]-150:$p$}] (p) at (-150:1) {};
						\node[label={[label distance=\dist]90:$q$}] (q) at (90:1) {};
						\node[label={[label distance=\dist]-30:$r$}] (r) at (-30:1) {};
					
						\node[label={[label distance=\dist]30:$q\lambda r$}] (qr) at ($(q)!.5!(r)$) {};
						\node[label={[label distance=\dist]120:$p\sigma q$}] (sigma) at ($(p)!.5!(q)$) {};

						\node[label={[label distance=\dist]90:$s$}] (s) at (intersection of p--qr and sigma--r) {};
						
						\node[label={[label distance=\dist]150:$p\sigma' q$}] (sigma') at ($(p)!.25!(q)$) {};
						\node[label={[label distance=\dist]-100:$r\mu^\ast(p\sigma' q)$}] (muast) at ($(r)!.5!(sigma')$) {};
						
						\node[label={[label distance=\dist]-90:$t$}] (t) at ($(p)!.8!(r)$) {};

						\draw[-latex] (q.center) -- (sigma.center);
						\draw[-latex] (r.center) -- (t.center);
						\draw[-latex,dashed] (t.center) -- (p.center);
						\draw[-latex] (q.center) -- (qr.center);
						\draw[-latex,dashed] (qr.center) -- (r.center);
						\draw (p.center) -- (qr.center);
						\draw (sigma.center) -- (r.center);
						\draw (sigma.center) -- (p.center);
						
						\draw[-latex] (r.center) -- (muast.center);
						\draw[dashed,-latex] (muast.center) -- (sigma'.center);
						
						\draw (t.center) -- (s.center);
						\draw (t.center) -- (sigma.center);

					\end{tikzpicture}	
					\caption{}
					\label{fig:convex2b}
			\end{subfigure}
			\caption{Illustration of the proof of \Cref{prop:rationalizable}. A solid line with or without an arrowhead from one outcome $u$ to another outcome $v$ denotes that $u\succ v$ or $u\sim v$, respectively. A dashed line with an arrowhead from $u$ to $v$ denotes that $u\succsim v$. The dotted line indicates an indeterminate preference.}
		\label{fig:convex}
	\end{figure}
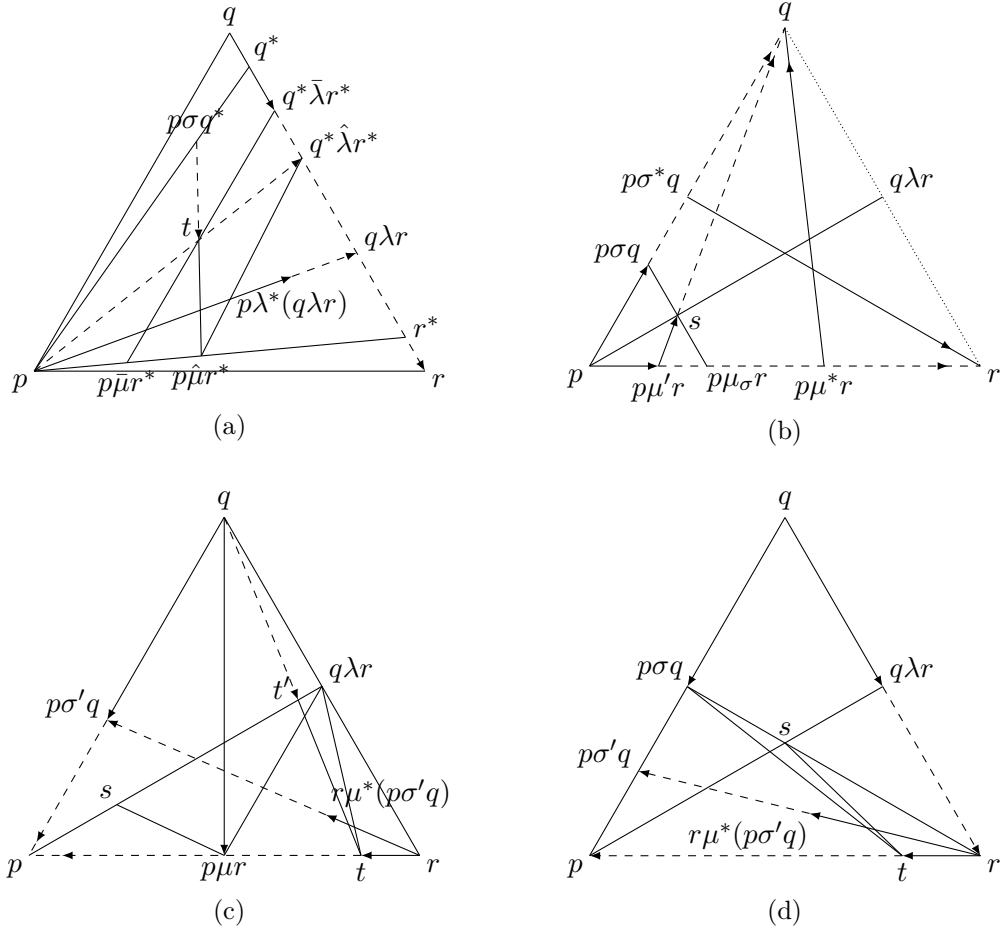

Now we are in a position to show that $\succ$ satisfies convexity.
This requires to show that, for all $p\in\Delta$, $L(p)$, $U(p)$, $L(p)\cup I(p)$, and $U(p)\cup I(p)$ are convex.

\begin{enumerate}[label=(\roman*)]
	\item $L(p)\cup I(p)$ is convex: Let $q,r\in L(p)\cup I(p)$ and $\lambda\in [0,1]$. 
By definition of $\succ$, $p\in S([p,q])\cap S([p,r])$.
Since $S$ satisfies expansion, it follows that $p\in S(\conv(p,q,r))$.
Then, $S$ satisfying contraction implies that $p\in S([p,q\lambda r])$, i.e., $p\succsim q\lambda r$.\label{item:lu}

	\item $L(p)$ is convex: Let $q,r\in L(p)$ and $\lambda\in [0,1]$. 
	The proof for this case is illustrated in \Cref{fig:convex1}.
By~\ref{item:lu}, we know that $p\succsim q\lambda r$.
Assume for contradiction that $p\sim q\lambda r$.
Let $\mu^\ast\in(0,1)$ such that, for all $\mu\in(\mu^\ast,1)$, $p\mu r\succ q$, which exists since, by continuity of $\succ$, $U(q)$ is open.
Similarly, there is $\sigma^\ast\in(0,1)$ such that, for all $\sigma\in(\sigma^\ast,1)$, $p\sigma q\succ r$.
For all $\sigma\in(\sigma^\ast,1)$, let $\mu_\sigma\in(0,1)$ such that $p\sigma q\sim p\mu_\sigma r$, which exists since, by \ref{eq:weakbetweenness}, $p\succ p\sigma q$, $p\sigma q\succ r$, and $\succ$ satisfies Archimedean continuity by \Cref{lem:archcont}.
Hence, for all $\sigma\in(\sigma^\ast,1)$, $p\mu_\sigma r\in S([p\sigma q,p\mu_\sigma r])$.
Since $p\sigma q$ goes to $p$ as $\sigma$ goes to $1$ and, by \ref{eq:weakbetweenness}, $p\succ p\mu r$ for all $\mu\in (0,1)$, continuity of $S$ implies that $\mu_\sigma$ goes to $1$ as $\sigma$ goes to $1$.
Thus, there is $\sigma\in(\sigma^\ast,1)$ such that $\mu_\sigma > \mu^\ast$.
By construction of $\sigma$, $[p\sigma q,p\mu_\sigma r]\cap [p,q\lambda r]\neq\emptyset$.
Let $s\in [p\sigma q,p\mu_\sigma r]\cap [p,q\lambda r]$ and $\mu'\in(\mu_\sigma,1)$ such that $s\in [p\mu' r,q]$.
Since $p\sim q\lambda r$, we have that $p\sim s$ and since $p\sigma q\sim p\mu_\sigma r$, we have that $s\sim p\mu_\sigma r$.
Since $I(s)$ is convex, it follows that $s\sim p\mu' q$.
However, by construction of $\mu'$, we have that $p\mu' r\succ q$.
Then, \ref{eq:weakbetweenness} implies that $p\mu' r\succ s$, which is a contradiction. 
Hence, $p\succ q\lambda r$.
\label{item:l}

	\item $U(p)$ is convex: Let $q,r\in U(p)$ and $\lambda\in [0,1]$. 
If $p\succ q\lambda r$, then since $\succ$ satisfies Archimedean continuity by \Cref{lem:archcont}, there is $\lambda'\in(\lambda,1)$ such that $p\sim q\lambda' r$.
Hence, we may assume without loss of generality that $p\sim q\lambda r$.

First consider the case that $q\sim r$.
The proof for this case is illustrated in \Cref{fig:convex2a}.
Since by continuity of $\succ$, $L(q)$ is open, there is $\mu\in (0,1)$ such that $q\succ p\mu r$.
Since $q\lambda r\sim p$, $q\lambda r\sim r$, and $I(q\lambda r)$ is convex by \Cref{claim:convexindiff}, it follows that $p\mu r\sim q\lambda r$.
As $\succ$ satisfies \ref{eq:weakbetweenness} by \Cref{claim:lowerbetweenness}, $r\succ p$ implies that $p\mu r\succsim p$.
If $p\mu r\succ p$, Archimedean continuity of $\succ$ implies that $p\mu r\sim p\sigma q$ for some $\sigma\in(0,1)$.
Then, for $s\in [p,q\lambda r]\cap [p\mu r,p\sigma q]$, we have that $p\mu r\sim s$.
Otherwise $p\mu r\sim p$.
In any case, there is $s\in[p,q\lambda r)$ such that $p\mu r\sim q\lambda r$, $p\mu r\sim s$, and $q\lambda r \sim s$.
Hence, by \Cref{claim:allindifferent}, $\succ$ is all indifferent on $\conv(p\mu r, q\lambda r, s)$.
Since by continuity of $\succ$, $L(r)$ is open, there is $\sigma'\in (0,1)$ such that $r\succ p\sigma' q$.
Let $\mu^\ast\in[0,1)$ such that $\succ$ satisfies $\mu^\ast$-lower betweenness on $[r,p\sigma'q]$, which exists by \Cref{claim:lowerbetweenness}.
Moreover, let $t\in(p,r)$ such that $[q,t]\cap [r,r\mu^\ast(p\sigma' q)]\neq\emptyset$ and $[q,t]\cap [s,q\lambda r]\neq\emptyset$, which exists since $\mu^\ast< 1$ and $s\in[p,q\lambda r)$.
Since we have that $q\lambda r\sim p$ and $q\lambda r\sim r$ and $I(q\lambda r)$ is convex by \Cref{claim:convexindiff}, we have that $q\lambda r\sim t$.

Let $t'\in[p,q\lambda r]\cap [q,t]$.
Since $\succ$ is all indifferent on $\conv(p\mu r, q\lambda r, s)$ as shown before, in particular, $\succ$ is all indifferent on $\conv(p\mu r, q\lambda r, s)\cap[t,t']$. 
From~\ref{item:lu}, we know that $L(q)\cup I(q)$ is convex.
Hence, $q\succsim t$.
By \Cref{claim:lowerbetweenness}, this implies that $t'\succsim t$.
If $t'\succ t$ then again \Cref{claim:lowerbetweenness} implies that $t'\succ t''$ for all $t''\in [t,t')$, which contradicts that $\succ$ is all indifferent on $\conv(p\mu r, q\lambda r, s)\cap[t,t']$.
Hence, $t\sim t'$.
Moreover, since $p\sim q\lambda r$, $q\lambda r\sim t'$.
Hence, by \Cref{claim:allindifferent}, $\succ$ is all indifferent on $\conv(q\lambda r,t,t')$.
This is a contradiction, since $\conv(q\lambda r,t,t')\cap [r,r\mu^\ast(p\sigma' r)]$ contains a sub-interval of $ [r,r\mu^\ast(p\sigma' r)]$ by the choice of $t$.

In the remaining case, we may assume without loss of generality that $q\succ r$.
Archimedean continuity of $\succ$ implies that there is $\sigma\in (0,1)$ such that $p\sigma q\sim r$.
The fact that $p\sim q\lambda r$ implies that $p\sim s$ for $s\in [p,q\lambda r]\cap [p\sigma q,r]$.
By \Cref{claim:lowerbetweenness}, it follows from $q\succ p$ that $p\sigma q\succsim p$.
If $p\sigma q\succ p$, we get a contradiction by applying the previous case to $p$, $p\sigma q$, and $r$.
So assume that $p\sim p\sigma q$.
The proof for this case is illustrated in \Cref{fig:convex2b}.
Observe that, since $p\sigma q\sim p$ and $p\sigma q \sim r$ and $I(p\sigma q)$ is convex by \Cref{claim:convexindiff}, it holds that, for all $s'\in\conv(p,p\sigma q, r)$, $p\sigma q\sim s'$.
Similarly, since $s\sim p$, $s\sim p\sigma q$, and $s\sim r$ and $I(s)$ is convex, it holds that, for all $s'\in \conv(p,p\sigma q,s)\cup\conv(p,r,s) = \conv(p,p\sigma q, r)$, $s\sim s'$.
Since, by continuity of $\succ$, $L(r)$ is open, there is $\sigma'\in (\sigma,1)$ such that $r\succ p\sigma' q$.
Let $\mu^\ast\in[0,1)$ such that $\succ$ satisfies $\mu^\ast$-lower betweenness on $[r,p\sigma' q]$, which exists by \Cref{claim:lowerbetweenness}.
Moreover, let $t\in [p,r]$ such that $[p\sigma q,t]\cap [r,r\mu^\ast(p\sigma' q)]\neq\emptyset$.
As shown above, we have that $p\sigma q\sim s$, $p\sigma q\sim t$, and $s\sim t$.
Thus, by \Cref{claim:allindifferent}, $\succ$ is all indifferent on $\conv(p\sigma q,s,t)$.
By the choice of $t$, there are $\mu,\mu'\in[\mu^\ast,1]$ such that $r\mu(p\sigma' q), r\mu'(p\sigma' q)\in\conv(p\sigma q, s,t)$ and $r\mu(p\sigma' q)\succ r\mu'(p\sigma' q)$, which is a contradiction to the fact that $\succ$ is all indifferent on $\conv(p\sigma q,s,t)$.
Hence, in any case, we have that $q\lambda r\succ p$.
\label{item:u}

	\item $U(p)\cup I(p)$ is convex: Let $q,r\in U(p)\cup I(p)$ and $\lambda\in [0,1]$. 
The cases that $q,r\in U(p)$ and $q,r\in I(p)$ are covered by~\ref{item:u} and the fact that $I(p)$ is convex, respectively.
Hence, we may assume without loss of generality that $q\in U(p)$ and $r\in I(p)$.
Assume for contradiction that $q\lambda r\not\in U(p)\cup I(p)$, i.e., $p\succ q\lambda r$.
Archimedean continuity of $\succ$ implies that $p\sim q\lambda' r$ for some $\lambda'\in (\lambda,1)$.
But then, since $I(p)$ is convex, $p\sim q\lambda' r$ and $p\sim r$ imply that $p\sim q\lambda r$, which is a contradiction.
\end{enumerate}
\end{proof}

\section{SSB Utility Representation Theorem}\label{app:ssb}

In this section, we show that every preference relation in $\mathcal R$ admits a representation through an SSB function.
This is proven by reduction to \citeauthor{Fish82c}'s~\citeyearpar{Fish82c} SSB representation theorem, which states that every relation satisfying Archimedean continuity, dominance, and symmetry can be represented by an SSB function.
A preference relation $\succ$ satisfies \emph{dominance} if, for all $p,q,r\in\Delta$ and $\lambda\in(0,1)$,
\begin{equation}
\begin{aligned}
	&p\succ q \text{ and } p\succsim r \text{ imply } p\succ q \lambda r\text,\\
	&q\succ p \text{ and } r\succsim p \text{ imply } q \lambda r \succ p\text{, and}\\
	&p\sim q \text{ and } p\sim r \text{ imply } p\sim q \lambda r\text.\\
\end{aligned}
\tag{Dominance}
\end{equation}
Informally, dominance requires that for every outcome $p$, $I(p)$ is a hyperplane through $p$ separating $U(p)$ and $L(p)$.

\propssb*

\begin{proof}
	It is easy to see that every preference relation that admits a representation through an SSB function satisfies continuity, convexity, and symmetry. 
	So we only prove the ``if part'' here.
	Let $\succ$ be a preference relation satisfying continuity, convexity, and symmetry.
	\citet[][Theorem 1]{Fish82c} has shown that every preference relation satisfying Archimedean continuity, dominance, and symmetry can be represented by an SSB function.
	By \Cref{lem:archcont}, $\succ$ satisfies Archimedean continuity.
	Hence, it suffices to show that $\succ$ satisfies dominance.
	
	First observe that, by convexity of $\succ$, $U(p)\cup I(p)$ and $L(p)\cup I(p)$ are convex.
	As the intersection of convex sets, $(U(p)\cup I(p))\cap (L(p)\cup I(p)) = I(p)$ is convex, too.
	Hence, for all $p,q,r\in\Delta$ with $p\sim q$ and $p\sim r$ and $\lambda\in(0,1)$, we have that $p\sim q\lambda r$.
	This establishes the indifference part of dominance.
	
	Next, we prove two auxiliary statements:
	\begin{enumerate}[label=(\roman*)]
	\item For all $p,q\in\Delta$ with $p\succ q$ and $\lambda\in (0,1)$, $p\succ p\lambda q\succ q$.
	This condition is known as \emph{betweenness} \citep[see, e.g.,][]{Chew89a}.
	Assume for contradiction that betweenness is not satisfied.
	Then, there are $p,q\in\Delta$ with $p\succ q$ and $\lambda\in(0,1)$ such that either $p\not\succ p\lambda q$ or $p\lambda q\not\succ q$.
	In the first case, convexity of $L(p)\cup I(p)$ implies that $p\succsim p\lambda q$, since $p\sim p$ and $p\succ q$ by assumption.
	Hence, $p\sim p\lambda q$.
	Let $\lambda^\ast = \inf_{\lambda\in[0,1]}\{ p\sim p\lambda q\}$.
	By continuity of $\succ$, $L(p)\cap [p,q]$ is open in $[p,q]$ and hence, $I(p)\cap [p,q] = [p,q]\setminus L(p)$ is closed in $[p,q]$.
	Hence, $p\sim p\lambda^\ast q$.
	In particular, $\lambda^\ast\in(0,1)$.
	Moreover, by convexity of $\succ$, $I(p)\cap [p,q]$ and $L(p)\cap [p,q]$ are convex.
	Hence, $I(p)\cap [p,q] = [p,p\lambda^\ast q]$ and $L(p)\cap [p,q] = (p\lambda^\ast q, q]$. Note that, since indifference sets are convex, it follows that $p\lambda' q\sim p\lambda''q$ for all $\lambda',\lambda''\in[\lambda^\ast,1]$.
	
	Let $\bar\lambda\in (\max\{0,2\lambda^\ast-1\},\lambda^\ast)$ and $\bar q = p\bar\lambda q$.
	So $\bar\lambda$ is chosen such that $\bar q$ is closer to $p\lambda^\ast q$ than $p\lambda^\ast q$ is to $p$.
	Since $\bar\lambda < \lambda^\ast$, $p\succ \bar q$.
	Now let $r = p$ and $\tilde\lambda \in (\bar\lambda,\lambda^\ast)$.
	Since $2\lambda^\ast < 1+\bar\lambda$ by construction, we have $\nicefrac12\,p+\nicefrac12\,\bar q = \nicefrac12\,(p1q) + \nicefrac12\,(p\bar\lambda q) = p((1 + \bar\lambda)/2)q\in [p,p\lambda^\ast q] \subseteq I(p)$.
	Thus, $r\sim \nicefrac 12\, p + \nicefrac12\, \bar q$.
	Moreover, $\bar q\tilde\lambda p = (p\bar\lambda q)\tilde\lambda p = p(\bar\lambda + 1-\tilde\lambda)q\sim \nicefrac12\, \bar q + \nicefrac12\, r$, since $\bar\lambda + 1 - \tilde\lambda > \bar\lambda  + 1 - \lambda^\ast > 2\lambda^\ast - 1 + 1 - \lambda^\ast = \lambda^\ast$.
	However, $\nicefrac12\, p + \nicefrac12\, r = p \succ p\tilde\lambda \bar q$, since $\tilde\lambda < \lambda^\ast$, which contradicts symmetry.
	An analogous argument takes care of the case $p\lambda q\not\succ q$.\label{enumi:betweenness}
	\item For all $p,q,r\in\Delta$, if there are $\bar p, \bar q, \bar r\in\conv(p,q,r)$ such that the affine hull of $\bar p$, $\bar q$, and $\bar r$ is equal to the affine hull of $p$, $q$, and $r$, and $\bar p\sim \bar q$, $\bar p\sim \bar r$, and $\bar q\sim \bar r$, then $p\sim q$, $p\sim r$, and $q\sim r$.
	For contradiction, assume without loss of generality that $p\succ q$.
	First, observe that the indifference part of dominance implies that, for all $s,t\in\conv(\bar p,\bar q,\bar r)$, $s\sim t$.
	Let $x$ be in the relative interior of $\conv(\bar p,\bar q,\bar r)$.
	For $y\in\conv(p,q,r)$, let $\lambda\in(0,1)$ such that $y\lambda x\in\conv(\bar p,\bar q,\bar r)$.
	If $y\not\sim x$, then by~\ref{enumi:betweenness}, $y\lambda x\not\sim x$, which is a contradiction.
	Hence, $x\sim y$, i.e., every outcome in the relative interior of $\conv(\bar p,\bar q,\bar r)$ is indifferent to every lottery in $\conv(p,q,r)$.
	Now let $\mu\in (0,1)$ such that $q^\ast = q\mu x$ is in the relative interior of $\conv(\bar p,\bar q,\bar r)$.
	Then, $p\sim q^\ast$, $p\sim x$, and $q^\ast\sim x$ by what we have shown before.
	This implies that, for all $s,t\in\conv(p,q^\ast, x)$, $s\sim t$.
	Since $\succ$ satisfies continuity, $U(q)$ is open and hence, $U(q)\cap [p,x]$ is open in $[p,x]$.
	Thus, there is $p^\ast\in(p,x)$ such that $p^\ast\succ q$.
	For $\sigma\in (0,1)$ close to one, we have that $p^\ast\sigma q\in\conv(p,q^\ast, x)$, which implies that $p^\ast\sim p^\ast\sigma q$.
	However, by~\ref{enumi:betweenness}, we have that $p^\ast\succ p^\ast\sigma q$, which is a contradiction.
	\label{enumi:triple}
	\end{enumerate}
	
	Now we are ready to show the missing parts of dominance.
	To this end, let $p,q,r\in\Delta$ and $\lambda\in (0,1)$.
	We first show that $p\succ q$ and $p\succsim r$ imply $p\succ q\lambda r$.
	If $p\succ q$ and $p\succ r$, it follows from convexity of $L(p)$ that $p\succ q\lambda r$.
	If $p\succ q$ and $p\sim r$, it follows from convexity of $L(p)\cup I(p)$ that $p\succsim q\lambda r$.
	Assume for contradiction that $p\sim q\lambda r$.
	We distinguish three cases, illustrated in \Cref{fig:ssbrepproof}:
	\begin{figure}
		
			\tikzset{->-/.style={decoration={
			  markings,
			  mark=at position #1 with {\arrow{latex}}},postaction={decorate}}}
			\def\arrow{1.5ex}
			\def\dist{.5ex}
			\def\sc{2.2}
			\centering
			\begin{subfigure}{0.32\textwidth}
				\centering
				\begin{tikzpicture}[scale = \sc, every node/.style={inner sep=0,outer sep=0}]
					
					\node[label={[label distance=\dist]-150:$p$}] (p) at (-150:1) {};
					\node[label={[label distance=\dist]90:$q$}] (q) at (90:1) {};
					\node[label={[label distance=\dist]-30:$r$}] (r) at (-30:1) {};
					
					\node[label={[label distance=\dist]30:$q\lambda r$}] (qr) at ($(q)!.5!(r)$) {};
		
					\draw[-latex] (p.center) -- (q.center);
					\draw (p.center) -- (r.center);
					\draw (p.center) -- (qr.center);
					\draw (q.center) -- (r.center);

				\end{tikzpicture}	
				\caption{The case $q\sim r$.}
				\label{fig:ssb1}
			\end{subfigure}
			\begin{subfigure}{0.32\textwidth}
				\centering
				\begin{tikzpicture}[scale = \sc, every node/.style={inner sep=0,outer sep=0}]
					
					\node[label={[label distance=\dist]-150:$p$}] (p) at (-150:1) {};
					\node[label={[label distance=\dist]90:$q$}] (q) at (90:1) {};
					\node[label={[label distance=\dist]-30:$r$}] (r) at (-30:1) {};
					
					\node[label={[label distance=\dist]30:$q\lambda r$}] (qr) at ($(q)!.5!(r)$) {};
					\node[label={[label distance=\dist]-90:$p\mu r$}] (pr) at ($(p)!.5!(r)$) {};
					\node[label={[label distance=\dist]150:$\bar q$}] (qbar) at (intersection of p--qr and q--pr) {};
					
					\draw[-latex] (p.center) -- (q.center);
					\draw (p.center) -- (r.center);
					\draw (p.center) -- (qr.center);
					\draw (q.center) -- (pr.center);
					\draw[-latex] (q.center) -- (r.center);

				\end{tikzpicture}
				\caption{The case $q\succ r$.}
				\label{fig:ssb2}
			\end{subfigure}
			\begin{subfigure}{0.32\textwidth}
				\centering
				\begin{tikzpicture}[scale = \sc, every node/.style={inner sep=0,outer sep=0}]
					
					\node[label={[label distance=\dist]-150:$p$}] (p) at (-150:1) {};
					\node[label={[label distance=\dist]90:$q$}] (q) at (90:1) {};
					\node[label={[label distance=\dist]-30:$r$}] (r) at (-30:1) {};
					
					\node[label={[label distance=\dist]30:$q\lambda^\ast r$}] (qr) at ($(q)!.5!(r)$) {};
					\node[label={[label distance=\dist]50:$\bar q$}] (qbar) at ($(q)!.2!(r)$) {};
					\node[label={[label distance=\dist]10:$\bar r$}] (rbar) at ($(q)!.8!(r)$) {};
		
					\draw[->-=.9] (p.center) -- (q.center);
					\draw (p.center) -- (r.center);
					\draw (p.center) -- (qr.center);
					\draw[->-=.9] (r.center) -- (q.center);
					\draw[-latex] (p.center) -- (qbar.center);
					\draw (p.center) -- (rbar.center);

				\end{tikzpicture}	
				\caption{The case $r\succ q$.}
				\label{fig:ssb3}
			\end{subfigure}
			\caption{Illustration of the proof of \Cref{prop:ssb}. A solid line with or without arrowhead from one outcome $u$ to another outcome $v$ denotes that $u\succ v$ or $u\sim v$, respectively.}
			
		\label{fig:ssbrepproof}
	\end{figure}
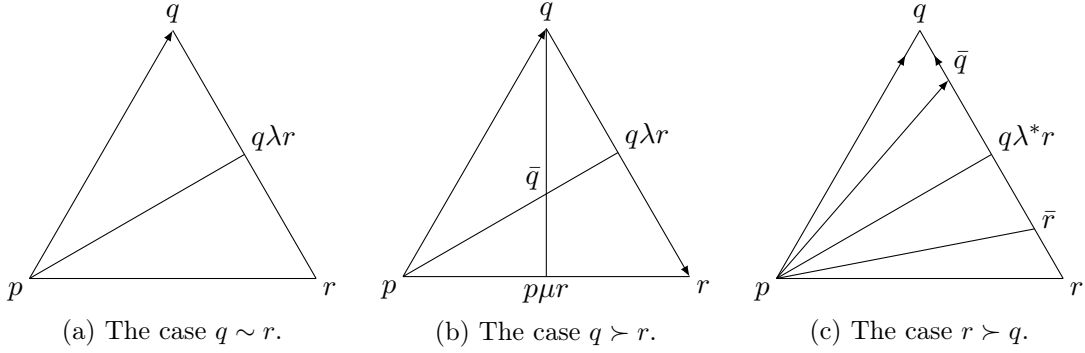
	\begin{enumerate}[label=(\roman*)]
		\setcounter{enumi}{2}
		\item $q\sim r$: 
		The indifference part of dominance implies that $q\lambda r\sim r$.
		Hence, we have that $p\sim r$, $p\sim q\lambda r$, and $q\lambda r\sim r$.
		By~\ref{enumi:triple}, we have $p\sim q$, which contradicts $p\succ q$.
		 
 	\label{item:ssb1}
		\item $q\succ r$:
		Since by assumption $p\succ q$ and $q\succ r$ and since, by \Cref{lem:archcont}, $\succ$ satisfies Archimedean continuity, there is $\mu\in (0,1)$ such that $q\sim p\mu r$. 
		Let $\bar q\in [p,q\lambda r]\cap [q,p\mu r]$.
		Then, by the indifference part of dominance, $p\sim \bar q$, $p\sim p\mu r$, and $\bar q\sim p\mu r$.
		By~\ref{enumi:triple}, we have $p\sim q$, which contradicts $p\succ q$.
		\label{item:ssb2}
		\item $r\succ q$: 
		Continuity of $\succ$ implies that $L(p)\cap [q,r]$ is open in $[q,r]$ and $I(p)\cap [q,r] = [q,r]\setminus L(p)$ is closed in $[q,r]$.
		Convexity of $\succ$ implies that $L(p)\cap [q,r]$ and $I(p)\cap [q,r]$ are convex.
		Hence, there is $\lambda^\ast\in (0,1)$ such that $L(p)\cap [q,r] = [q,q\lambda^\ast r)$ and $I(p)\cap[q,r] = [q\lambda^\ast r, r]$.
		Let $\bar q\in [q,q\lambda^\ast r)$ and $\bar r\in (q\lambda^\ast r,r]$ such that $q\lambda^\ast r = \nicefrac12\, \bar q + \nicefrac12\, \bar r$.
		Such $\bar q$ and $\bar r$ exist, since $\lambda^\ast\in (0,1)$.
		By construction, $p\succ \bar q$ and $p\sim \bar r$.
		Since $r\sim p$ and $r\succ \bar q$ by assumption and~\ref{enumi:betweenness}, convexity of $L(r)\cup I(r)$ implies that $r\succsim \nicefrac12\, p + \nicefrac12\, \bar q$.
		\begin{itemize}
			\item[--] If $r\sim \nicefrac12\, p + \nicefrac12\, \bar q$, let $\bar p\in [p,q\lambda^\ast r]\cap [r,\nicefrac12\, p + \nicefrac12\, \bar q]$.
			Then, $p\sim\bar p$, $p\sim r$, and $\bar p\sim r$.
			By~\ref{enumi:triple}, this contradicts $p\succ q$.
			\item[--] Now consider the case that $r\succ \nicefrac12\, p + \nicefrac12\, \bar q$.
			From~\ref{enumi:betweenness}, it follows that $\nicefrac12\, p + \nicefrac12\, \bar q\succ \bar q$, because $p \pref \bar q$.
			By \Cref{lem:archcont}, $\succ$ satisfies Archimedean continuity and hence, $\nicefrac12\, p + \nicefrac12\, \bar q\sim s$ for some $s\in [\bar q, r]$.
			If $s\in (q\lambda^\ast r, r)$, let $\bar p \in [p,q\lambda^\ast r]\cap  [\nicefrac12\, p + \nicefrac12\, \bar q, s]$, which is non-empty in this case.
			Then, $\bar p\sim  q\lambda^\ast r$ and $\bar p\sim \nicefrac12\, p + \nicefrac12\, \bar q$.
			Since $I(\bar p)$ is convex, it follows that $[q\lambda^\ast r, \nicefrac12\, p + \nicefrac12\, \bar q]\subseteq I(\bar p)$.
			If $\bar p\succ q$, then for $q'\in[q\lambda^\ast r, \nicefrac12\, p + \nicefrac12\, \bar q]\cap [\bar p, q]$, we have by~\ref{enumi:betweenness} that $\bar p\succ q'$, which is a contradiction.
			So $\bar p\sim q$.
			Then, again by~\ref{enumi:betweenness}, for $t\in [p,r]$ such that $\bar p\in [t,q]$, $t\sim q$.
			However, since $p\succ q$ and $r\succ q$ by assumption, it follows from convexity of $U(q)$ that $t\succ q$, which is a contradiction.
			
			If $s\in (\bar q, q\lambda^\ast r]$, it follows from symmetry of $\succ$ that $\nicefrac12\, p + \nicefrac12\, \bar r\sim t$ for some $t\in [q\lambda^\ast r, \bar r)$.
			Let $\bar p\in (p,r)$ such that $\nicefrac12\, p + \nicefrac12\, \bar r\in [\bar p, t]$.
			It follows from~\ref{enumi:betweenness} that $\bar p\sim t$.
			Hence, $p\sim \bar p$, $p\sim \nicefrac12\, p + \nicefrac12\, \bar r$, and $\bar p\sim \nicefrac12\, p + \nicefrac12\, \bar r$.
			By~\ref{enumi:triple}, this contradicts $p\succ q$.
		\end{itemize}\label{item:ssb3}
	\end{enumerate}
	The fact that $q\succ p$ and $r\succsim p$ imply $q\lambda r \succ p$ can be shown analogously.
\end{proof}

Noteworthily, \Cref{prop:ssb} does not hold if symmetry is replaced by the weaker notion of symmetry considered by \citet{Fish82c}, which only applies to cases where $p$, $q$, and $r$ as in the definition of symmetry in \Cref{sec:ssb} are linearly ordered (cf. \Cref{fn:symmetry}). The preference relation for $\mathcal U =\{a,b\}$ given in \Cref{sec:ssb} satisfies \citeauthor{Fish82c}'s notion of symmetry (since there are no three outcomes that are linearly ordered) but cannot be represented by an SSB function.

\end{document}